\documentclass[acmsmall,screen]{acmartFAKE}
\usepackage{booktabs}
\usepackage[ruled,vlined,linesnumbered,noend]{algorithm2e}

\SetAlFnt{\small}
\SetAlCapFnt{\small}
\SetAlCapNameFnt{\small}
\SetAlCapHSkip{0pt}
\IncMargin{-\parindent}

\usepackage{mleftright}
\usepackage[utf8]{inputenc}
\usepackage[T1]{fontenc}
\usepackage{amsmath}
\usepackage{amsthm}
\usepackage{bm}
\usepackage{comment}
\usepackage{tikz}
\usepackage{thm-restate}
\usepackage{mathtools}
\usepackage{bbm}
\usepackage[shortlabels]{enumitem}
\usepackage{hyperref}
\usepackage[capitalize]{cleveref}
\usepackage{nag}
\usepackage{etoolbox}
\usepackage{nicefrac}
\usepackage{wrapfig}
\usepackage{float}

\setlist{nosep,leftmargin=8mm,partopsep=0pt,topsep=0pt}

\makeatletter
\patchcmd\algocf@Vline{\vrule}{\vrule \kern-0.4pt}{}{}
\patchcmd\algocf@Vsline{\vrule}{\vrule \kern-0.4pt}{}{}
\makeatother

\SetKwComment{Hline}{}{\vspace{-3mm}\textcolor{gray}{\hrule}\vspace{1mm}}
\definecolor{darkgrey}{gray}{0.3}
\definecolor{commentcolor}{gray}{0.5}
\SetKwComment{Comment}{\color{commentcolor}[$\triangleright$\ }{}
\SetCommentSty{}
\SetNlSty{}{\color{darkgrey}}{}
\SetKwProg{Fn}{function}{}{}
\SetKwProg{Subr}{subroutine}{}{}
\crefalias{AlgoLine}{line}%
\crefname{algocf}{Algorithm}{Algorithms}

\makeatletter
\let\cref@old@stepcounter\stepcounter
\def\stepcounter#1{%
  \cref@old@stepcounter{#1}%
  \cref@constructprefix{#1}{\cref@result}%
  \@ifundefined{cref@#1@alias}%
    {\def\@tempa{#1}}%
    {\def\@tempa{\csname cref@#1@alias\endcsname}}%
  \protected@edef\cref@currentlabel{%
    [\@tempa][\arabic{#1}][\cref@result]%
    \csname p@#1\endcsname\csname the#1\endcsname}}
\makeatother

\definecolor{custom_color}{RGB}{225,231,255}
\usepackage[framemethod=tikz]{mdframed}

\crefname{algocf}{alg.}{algs.}
\Crefname{algocf}{Algorithm}{Algorithms}

\newcommand{\declarecolor}[2]{\definecolor{#1}{RGB}{#2}\expandafter\newcommand\csname #1\endcsname[1]{\textcolor{#1}{##1}}}
\declarecolor{White}{255, 255, 255}
\declarecolor{Black}{0, 0, 0}
\declarecolor{Maroon}{128, 0, 0}
\declarecolor{Coral}{255, 127, 80}
\declarecolor{LimeGreen}{50, 205, 50}
\declarecolor{DarkGreen}{0, 100, 0}
\declarecolor{Navy}{0, 0, 128}

\hypersetup{
	colorlinks=true,
	pdfpagemode=UseNone,
	citecolor=DarkGreen,
	linkcolor=Maroon,
	urlcolor=Navy,
	pdfstartview=FitW
}

\theoremstyle{plain}
\newtheorem{theorem}{Theorem}[section]
\newtheorem{lemma}[theorem]{Lemma}
\newtheorem{corollary}[theorem]{Corollary}
\newtheorem{proposition}[theorem]{Proposition}
\newtheorem{fact}[theorem]{Fact}

\theoremstyle{definition}
\newtheorem{definition}[theorem]{Definition}

\theoremstyle{remark}
\newtheorem{remark}[theorem]{Remark}

\usetikzlibrary{shapes, patterns, decorations, fit, intersections}

\newcommand*{\N}{{\mathbb{N}}}
\newcommand*{\R}{\mathbb{R}}
\newcommand{\defeq}{\coloneqq}

\let\co\relax
\let\polylog\relax

\DeclareMathOperator{\polylog}{polylog}

\DeclareMathOperator{\reg}{Reg}
\DeclareMathOperator*{\argmin}{arg\,min}
\DeclareMathOperator*{\argmax}{arg\,max}

\DeclareMathOperator{\adj}{adj}
\DeclareMathOperator{\co}{co}
\DeclareMathOperator{\relint}{relint}

\DeclareMathOperator{\map}{\mathtt{mp}}
\DeclareMathOperator{\minor}{\mathtt{mn}}

\DeclareMathOperator{\nextstr}{\textsc{NextStrategy}}
\DeclareMathOperator{\efce}{\textsc{EFCE}}

\DeclareMathOperator{\cfr}{\textsc{CFR}}

\DeclareMathOperator{\nash}{\textsc{NE}}
\DeclareMathOperator{\omw}{\textsc{OMWU}}
\DeclareMathOperator{\rmp}{\textsc{RM}^+}
\DeclareMathOperator{\obsut}{\textsc{ObserveUtility}}
\DeclareMathOperator{\fporacle}{\textsc{StableFPOracle}}

\let\E\relax
\DeclareMathOperator{\E}{\mathbb{E}}

\renewcommand{\vec}[1]{\bm{#1}}
\newcommand{\mat}[1]{\mathbf{#1}}

\DeclarePairedDelimiterX{\card}[1]{\lvert}{\rvert}{#1}
\DeclarePairedDelimiterX{\abs}[1]{\lvert}{\rvert}{#1}
\DeclarePairedDelimiterX{\tuple}[1]{\lparen}{\rparen}{#1}
\DeclarePairedDelimiterX{\parens}[1]{\lparen}{\rparen}{#1}
\DeclarePairedDelimiterX{\brackets}[1]{\lbrack}{\rbrack}{#1}
\DeclarePairedDelimiterX{\set}[1]\{\}{#1}
\let\Pr\relax
\DeclarePairedDelimiterXPP{\Pr}[1]{\mathbb{P}}[]{}{#1}
\DeclarePairedDelimiterXPP{\PrX}[2]{\mathbb{P}_{#1}}[]{}{#2}
\DeclarePairedDelimiterXPP{\Ex}[1]{\mathbb{E}}[]{}{#1}
\DeclarePairedDelimiterXPP{\ExX}[2]{\mathbb{E}_{#1}}[]{}{#2}

\usepackage{tikz}
\usetikzlibrary{arrows.meta,calc,shapes.misc,arrows}
\tikzset{cross/.style={path picture={
  \draw[black]
(path picture bounding box.south east) -- (path picture bounding box.north west) (path picture bounding box.south west) -- (path picture bounding box.north east);
}}}
\tikzstyle{chanode}=[fill=white,draw=black,circle,cross,inner sep=1mm]
\tikzstyle{pl1node}=[fill=black,draw=black,circle,inner sep=.8mm]
\tikzstyle{pl2node}=[fill=white,draw=black,circle,inner sep=.8mm]
\tikzstyle{termina}=[fill=white,draw=black,inner sep=.8mm]

\newcommand{\emptyseq}{\varnothing}

\usetikzlibrary{fit,shapes.misc,arrows.meta}
\makeatletter
\tikzset{
  fitting node/.style={
    inner sep=0pt,
    fill=none,
    draw=none,
    reset transform,
    fit={(\pgf@pathminx,\pgf@pathminy) (\pgf@pathmaxx,\pgf@pathmaxy)}
  },
  reset transform/.code={\pgftransformreset}
}
\tikzset{cross/.style={path picture={
  \draw[black]
(path picture bounding box.south east) -- (path picture bounding box.north west) (path picture bounding box.south west) -- (path picture bounding box.north east);
}}}
\tikzstyle{ox}=[semithick,draw=black,circle,cross,inner sep=1.2mm]
\makeatother

\makeatletter
\tikzset{
  fitting node/.style={
    inner sep=0pt,
    fill=none,
    draw=none,
    reset transform,
    fit={(\pgf@pathminx,\pgf@pathminy) (\pgf@pathmaxx,\pgf@pathmaxy)}
  },
  reset transform/.code={\pgftransformreset}
}
\tikzstyle{ox}=[semithick,draw=black,circle,cross,inner sep=1.2mm]
\makeatother
\newcommand{\seq}[1]{\texttt{#1}}

\setcitestyle{authoryear}

\title{Faster No-Regret Learning Dynamics for Extensive-Form Correlated and Coarse Correlated Equilibria}

\author{Ioannis Anagnostides}
\authornote{Both authors contributed equally to the paper.}
\email{ianagnos@cs.cmu.edu}
\affiliation{%
  \institution{Carnegie Mellon University}
  \streetaddress{5000 Forbes Avenue}
  \city{Pittsburgh}
  \state{Pennsylvania}
  \country{USA}
  \postcode{15232}
}

\author{Gabriele Farina}
\authornotemark[1]
\email{gfarina@cs.cmu.edu}
\affiliation{%
  \institution{Carnegie Mellon University}
  \streetaddress{5000 Forbes Avenue}
  \city{Pittsburgh}
  \state{Pennsylvania}
  \country{USA}
  \postcode{15232}
}
\author{Christian Kroer}
\email{christian.kroer@columbia.edu}
\affiliation{%
  \institution{Columbia University}
  \city{New York City}
  \state{New York}
  \country{USA}
  \postcode{10027}
}
\author{Andrea Celli}
\email{andrea.celli2@unibocconi.it}
\affiliation{%
 \institution{Bocconi University}
 \streetaddress{Via Roberto Sarfatti, 25}
 \city{Milan}
 \country{Italy}
 \postcode{20100}
}
\author{Tuomas Sandholm}
\email{sandholm@cs.cmu.edu}
\affiliation{%
  \institution{Carnegie Mellon University \& Strategy Robot, Inc. \& Optimized Markets, Inc. \& Strategic Machine, Inc.}
  \city{Pittsburgh}
  \state{Pennsylvania}
  \country{USA}
  \postcode{15232}
}

\begin{abstract}
    A recent emerging trend in the literature on learning in games has been concerned with providing faster learning dynamics for correlated and coarse correlated equilibria in normal-form games. Much less is known about the significantly more challenging setting of extensive-form games, which can capture both sequential and simultaneous moves, as well as imperfect information. In this paper we establish faster no-regret learning dynamics for \textit{extensive-form correlated equilibria (EFCE)} in multiplayer general-sum imperfect-information extensive-form games. When all players follow our accelerated dynamics, the correlated distribution of play is an $O(T^{-3/4})$-approximate EFCE, where the $O(\cdot)$ notation suppresses parameters polynomial in the description of the game. This significantly improves over the best prior rate of $O(T^{-1/2})$. To achieve this, we develop a framework for performing accelerated \emph{Phi-regret minimization} via predictions. One of our key technical contributions---that enables us to employ our generic template---is to characterize the stability of fixed points associated with \emph{trigger deviation functions} through a refined perturbation analysis of a structured Markov chain. Furthermore, for the simpler solution concept of extensive-form \emph{coarse} correlated equilibrium (EFCCE) we give a new succinct closed-form characterization of the associated fixed points, bypassing the expensive computation of stationary distributions required for EFCE. Our results place EFCCE closer to \emph{normal-form coarse correlated equilibria} in terms of the per-iteration complexity, although the former prescribes a much more compelling notion of correlation. Finally, experiments conducted on standard benchmarks corroborate our theoretical findings.
\end{abstract}

\begin{document}

\begin{titlepage}

\maketitle

\end{titlepage}

\tableofcontents
\clearpage

\section{Introduction}

Game-theoretic solution concepts describe how rational agents should act in games. Over the last two decades there has been tremendous progress in imperfect-information game solving and algorithms based on game-theoretic solution concepts have become the state of the art. Prominent milestones include an optimal strategy for Rhode Island hold'em poker~\citep{Gilpin07:Lossless}, a near-optimal strategy for limit Texas hold'em~\citep{Bowling15:Heads}, and a superhuman strategy for no-limit Texas hold'em~\citep{Brown17:Superhuman,Moravvcik17:DeepStack}. In particular, these advances rely on algorithms that approximate \emph{Nash equilibria} (\emph{NE}) of two-player zero-sum \emph{extensive-form games} (\emph{EFGs}). EFGs are a broad class of games that capture sequential and simultaneous interaction, and imperfect information.
For two-player zero-sum EFGs, it is by now well-understood how to compute a Nash equilibrium at scale: in theory this can be achieved using accelerated uncoupled no-regret learning dynamics, for example by having each player use an \emph{optimistic} regret minimizer and leveraging suitable \emph{distance-generating functions}~\citep{Hoda10:Smoothing,Kroer20:Faster,Farina21:Better} for the EFG decision space. Such a setup converges to an equilibrium at a rate of $O(T^{-1})$. In practice, modern variants of the \emph{counterfactual regret minimization (CFR)} framework~\citep{Zinkevich07:Regret} typically lead to better practical performance, although the worst-case convergence rate known in theory remains inferior. CFR is also an uncoupled no-regret learning dynamic.

However, many real-world applications are not two-player zero-sum games, but instead have \emph{general-sum} utilities and often more than two players. 
In such settings, Nash equilibrium suffers from several drawbacks when used as a prescriptive tool. 
First, there can be multiple equilibria, and an equilibrium strategy may perform very poorly when played against the ``wrong'' equilibrium strategies of the other player(s).
Thus, the players effectively would need to communicate in order to find an equilibrium, or hope to converge to it via some sort of decentralized learning dynamics.
Second, finding a Nash equilibrium is computationally hard both in theory~\citep{Daskalakis06:Complexity,Etessami07:Complexity} and in practice~\citep{Berg17:Exclusion}.  
This effectively squashes any hope of developing efficient learning dynamics that converge to Nash equilibria in general games.

A competing notion of rationality proposed by \citet{Aumann74:Subjectivity} is that of \emph{correlated equilibrium} (\emph{CE}). 
%
%
Unlike $\nash$, it is known that the former can be computed in polynomial time and, perhaps even more importantly, it can be attained through \emph{uncoupled} learning dynamics where players only need to reason about their own observed utilities. This overcomes the often unreasonable presumption that players have knowledge about the other players' utilities. At the same time, uncoupled learning algorithms have proven to be a remarkably \emph{scalable} approach for computing equilibria in large-scale games, as described above. In normal-form games (NFGs), a \emph{correlated strategy} is defined as a probability distribution over joint action profiles, customarily modeled via a trusted external mediator that draws an action profile from this distribution and then privately recommends to each player their component. A correlated strategy is a CE if, for each player, the mediator's recommendation is the best action in expectation, assuming that all the other players follow their recommended actions~\cite{Aumann74:Subjectivity}. In NFGs it has long been known that uncoupled no-regret learning dynamics can converge to CE and \emph{coarse correlated equilibria} (\emph{CCE}) at a rate of $O(T^{-1/2})$~\citep{Foster97:Calibrated,Hart00:Simple}. 
More recently, it has been established that accelerated dynamics can converge at a rate of $\widetilde{O}(T^{-1})$~\citep{Daskalakis21:near,Anagnostides21:Near} in NFGs, where the notation $\widetilde{O}(\cdot)$ suppresses $\polylog(T)$ factors.

However, in the context of EFGs the idea of correlation is much more intricate, and there are several notions of correlated equilibria based on when the mediator gives recommendations and how the mediator reacts to players who disregard the advice. 
%
Three natural extensions of CE to extensive-form games are the \emph{extensive-form correlated equilibrium (EFCE)} by~\citet{Stengel08:Extensive}, the \emph{extensive-form coarse correlated equilibrium (EFCCE)} by~\citet{Farina20:Coarse}, and the \emph{normal-form coarse correlated equilibrium (NFCCE)} by~\citet{Celli2018:Computing}. The set of those equilibria are such that, for any extensive-form game, EFCE $\subseteq$ EFCCE $\subseteq$ NFCCE. In an EFCE, the stronger of those notions of correlation, the mediator forms recommendations for each of the possible decision points an agent may encounter in the game, and recommended actions are gradually revealed to players as they reach new information sets; thus, the mediator must take into account the \emph{evolution} of the players' beliefs throughout the game. Because of the sequential nature, the presence of private information in the game, and the gradual revelation of recommendations, the constraints associated with $\efce$ are significantly more complex than for normal-form games. 
For these reasons, the question of whether uncoupled learning dynamics can converge to an $\efce$ was only recently resolved by \citet{Celli20:No-Regret}. Moreover, in a follow-up work the authors also established an explicit rate of convergence of $O(T^{-1/2})$~\citep{Farina21:Simple}. Our paper is primarily concerned with the following fundamental question: 
\vspace{.3cm}
\begin{quote}
    \centering
    \emph{Can we develop faster uncoupled no-regret learning dynamics for EFCE?}
\end{quote}
\vspace{.3cm}

We affirmatively answer this question by developing dynamics converging at a rate of $O(T^{-3/4})$ to an EFCE. Furthermore, we also study learning dynamics for the simpler solution concept of EFCCE. More precisely, although accelerated learning dynamics for EFCE can be automatically employed for EFCCE (since the set of EFCEs forms a subset of the set of EFCCEs), all the known learning dynamics for EFCE have large per-iteration complexity. Indeed, they require as an intermediate step the expensive computation of the stationary distributions of multiple Markov chains. Thus, the following natural question arises: \emph{Are there learning dynamics for EFCCE that avoid the expensive computation of stationary distributions?} We answer this question in the positive. Our results reveal that EFCCE is more akin to NFCCE than to EFCE from a learning perspective, although EFCE prescribes a much more compelling notion of correlation than NFCCE. 

\subsection{Contributions}

Our first primary contribution is to develop faster no-regret learning dynamics for EFCE:

\begin{theorem}
    \label{theorem:main}
    On any general-sum multiplayer extensive-form game, there exist uncoupled no-regret learning dynamics which lead to a correlated distribution of play that is an $O(T^{-3/4})$-approximate $\efce$. Here the $O(\cdot)$ notation suppresses game-specific parameters polynomial in the size of the game.
\end{theorem}
This substantially improves over the prior best known rate of $O(T^{-1/2})$ recently established by~\citet{Farina21:Simple}. To achieve this result we employ the framework of \emph{predictive} (also known as \textit{optimistic}) regret minimization~\citep{Chiang12:Online,Rakhlin13:Optimization}. One of our conceptual contributions is to connect this line of work with the framework of \emph{Phi-regret} minimization~\citep{Greenwald03:General,Gordon08:No} by providing a general template for stable-predictive Phi-regret minimization (\Cref{theorem:accelerating-Phi}). The importance of Phi-regret is that it leads to substantially more compelling notions of hindsight rationality, well-beyond the usual \emph{external} regret \citep{Gordon08:No}, including the powerful notion of \emph{swap regret}~\citep{Blum07:From}. Moreover, one of the primary insights behind the result of \citet{Farina21:Simple} is to cast convergence to an EFCE as a Phi-regret minimization problem. Given these prior connections, we believe that our stable-predictive template is of independent interest, and could lead to further applications in the future.


From a technical standpoint, in order to apply our generic template for accelerated Phi-regret minimization (\Cref{theorem:accelerating-Phi}), we establish two separate ingredients. First, we develop a \emph{predictive} external regret minimizer for the set of transformations associated with $\efce$. This deviates from the construction of \citet{Farina21:Simple} in that we have to additionally guarantee and preserve the predictive bounds throughout the construction. Further, our algorithm combines optimistic regret minimization---under suitable DGFs---for the sequence-form polytope, with \emph{regret decomposition} in the style of CFR. While these have been the two main paradigms employed in EFGs, they were used separately in the past. We refer to \Cref{fig:algo} for a detailed description of our algorithm. 

The second central component consists of sharply characterizing the stability of fixed points of \emph{trigger deviation functions}. This turns out to be particularly challenging, and direct extensions of prior techniques only give a bound that is \emph{exponential} in the size of the game. In this context, one of our key technical contributions is to provide a refined perturbation analysis for a Markov chain consisting of a rank-one stochastic matrix (\Cref{lemma:convex_characterization}). To do this, we deviate from prior techniques (\emph{e.g.}, \citep{Candogan13:Dynamics,Chen20:Hedging}) that used the Markov chain tree theorem, and instead use an alternative linear-algebraic characterization for the eigenvectors of the underlying Laplacian system. This leads to a rate of convergence that depends \emph{polynomially} on the description of the game, which is crucial for the practical applicability of the dynamics. 

Next, we shift our attention to learning dynamics for EFCCE. We first introduce the notion of \emph{coarse trigger deviation functions}, and we show that if each player employs a no-coarse-trigger-regret algorithm, the correlated distribution of play converges to an EFCCE (\Cref{theorem:EFCCE-convergence}). This allows for a unifying treatment of EFCE and EFCCE. Moreover, we show that, unlike all existing methods for computing fixed points of trigger deviation functions, the fixed points of \emph{coarse} trigger deviation functions admit a succinct closed-form characterization (\Cref{thm:closedForm}); in turn, this enables us to obtain a much more efficient algorithm for computing them (\Cref{algo:FP-EFCCE}). From a practical standpoint, this is crucial as it substantially reduces the per-iteration complexity of the dynamics, placing EFCCE closer to NFCCE in terms of the underlying complexity, even though EFCCE prescribes a stronger notion of correlation. Another implication of our closed-form characterization is an improved stability analysis for the fixed points, which is much less technical than the one we give for EFCE (\Cref{proposition:FP-EFCCE}). Finally, we support our theoretical findings with experiments on several general-sum benchmarks.

\subsection{Further Related Work} 
\label{sec:related}

The line of work on accelerated no-regret learning was pioneered by \citet{Daskalakis15:Near}, showing that one can bypass the adversarial $\Omega(T^{-1/2})$ barrier for the incurred average regret if \emph{both} players in a zero-sum game employ an uncoupled variant of Nesterov's excessive gap technique~\citep{Nesterov05:Excessive}, leading to a near-optimal rate of $O(\log T/T)$. Subsequently, \citet{Rakhlin13:Online} showed that the optimal rate of $O(T^{-1})$ can be obtained with a remarkably simple variant of mirror descent which incorporates a \emph{prediction} term in the update step. While these results only hold for zero-sum games, \citet{Syrgkanis15:Fast} showed that an $O(T^{-3/4})$ rate can be obtained for multiplayer general-sum normal-form games. In a recent result, \citet{Chen20:Hedging} strengthened the regret bounds in \cite{Syrgkanis15:Fast} from external to swap regret using the celebrated construction of \citet{Blum07:From}, thereby establishing a rate of convergence of $O(T^{-3/4})$ to CE. 
%
Even more recent work~\citep{Daskalakis21:near,Anagnostides21:Near} has established a near-optimal rate of convergence of $\widetilde{O}(T^{-1})$ to correlated equilibria in normal-form games when all players leverage \emph{optimistic multiplicative weights update (OMWU)}, where $\widetilde{O}(\cdot)$ hides $\polylog(T)$ factors. Extending these results to EFCE presents a considerable challenge since their techniques crucially rely on the softmax-type strucure of OMWU on the simplex, as well as the particular structure of the associated fixed points.   



Correlated equilibria in extensive-form games are much less understood than Nash equilibria. 
It is known that a feasible $\efce$ can also be computed efficiently through a variant of the \emph{Ellipsoid algorithm} \citep{Papadimitriou08:Computing,Jiang15:Polynomial-Time}, while an alternative sampling-based approach was given by \citet{Dudik09:SamplingBased}. However, those approaches perform poorly in large-scale problems, and do not allow the players to arrive at $\efce$ via distributed learning. \citet{Celli19:Learning} devised variants of the $\cfr$ algorithm that provably converge to an NFCCE, a solution concept much less appealing than $\efce$ in extensive-form games~\citep{Gordon08:No}. Finally, \citet{Morrill21:Efficient, Morrill21:Hindsight} characterize different hindsight rationality notions in EFGs, associating each solution concept with suitable $O(T^{-1/2})$ no-regret learning dynamics. 
\section{Preliminaries}
\label{section:prel}

In this section we introduce the necessary background related to extensive-form games (EFGs), correlated equilibria in EFGs, and regret minimization. A comprehensive treatment on EFGs can be found in \cite{Shoham09:Multiagent}, while for an introduction to the theory of learning in games the reader is referred to the excellent book of \citet{Cesa06:Prediction}.

\paragraph{Conventions} In the sequel we use the $O(\cdot)$ notation to suppress (universal) constants. We typically use subscripts to indicate the player or some element in the game tree uniquely associated with a given player, such as a decision point; to lighten our notation, the associated player is not made explicit in the latter case. Superscripts are reserved almost exclusively for time indexes. Finally, the $k$-th coordinate of a vector $\vec{x} \in \R^d$ will be denoted by $\vec{x}[k]$.

\subsection{Extensive-Form Games} 

An extensive-form game is abstracted on a directed and rooted \emph{game tree} $\mathcal{T}$. The set of nodes of $\mathcal{T}$ is denoted with $\mathcal{H}$. Non-terminal nodes are referred to as \emph{decision nodes} and are associated with a player who acts by selecting an action from a set of possible actions $\mathcal{A}_h$, where $h \in \mathcal{H}$ represents the decision node. By convention, the set of players $[n] \cup \{c\}$ includes a \emph{fictitious} agent $c$ who ``selects'' actions according to some fixed probability distributions dictated by the nature of the game (\emph{e.g.}, the roll of a dice); this intends to model external stochastic phenomena occurring during the game. For a player $i \in [n] \cup \{c\}$, we let $\mathcal{H}_{i} \subseteq \mathcal{H}$ be the subset of decision nodes wherein a player $i$ makes a decision. The set of \emph{leaves} $\mathcal{Z} \subseteq \mathcal{H}$, or equivalently the \emph{terminal nodes}, correspond to different outcomes of the game. Once the game transitions to a terminal node $z \in \mathcal{Z}$ payoffs are assigned to each player based on a set of (normalized) utility functions $\{ u_{i} : \mathcal{Z} \to [-1,1] \}_{i \in [n]}$. It will also be convenient to represent with $p_{c}(z)$ the product of probabilities of ``chance'' moves encountered in the path from the root until the terminal node $z \in \mathcal{Z}$. In this context, the set of nodes in the game tree can be expressed as the (disjoint) union $\mathcal{H} \defeq \bigcup_{i \in [n] \cup \{c\}} \mathcal{H}_{i} \cup \mathcal{Z}$. 

\paragraph{Imperfect Information} To model imperfect information, the set of decision nodes $\mathcal{H}_{i}$ of player $i$ are partitioned into a collection of sets $\mathcal{J}_{i}$, which are called \emph{information sets}. Each information set $j \in \mathcal{J}_{i}$ groups nodes which cannot be distinguished by player $i$. Thus, for any nodes $h, h' \in j$ we have $\mathcal{A}_h = \mathcal{A}_{h'}$. As usual, we assume that the game satisfies \emph{perfect recall}: players never forget information once acquired. This implies, in particular, that for any nodes $h, h' \in j$ the sequence of $i$'s actions from the root until $h$ must coincide with the sequence from the root to node $h'$; otherwise, $i$ would be able to distinguish between nodes $h$ and $h'$ by virtue of perfect recall. We will also define a partial order $\prec$ on $\mathcal{J}_{i}$ so that $j \prec j'$, for $j, j' \in \mathcal{J}_{i}$, if there exist nodes $h \in j$ and $h' \in j'$ such that the path from the root to $h'$ passes through $h$. If $j \prec j'$, we will say that $j$ is an \emph{ancestor} of $j'$, or equivalently, $j'$ is a \emph{descendant} of $j$. 

\paragraph{Sequence-form Strategies} For a player $i \in [n]$, an information set $j \in \mathcal{J}_{i}$, and an action $a \in \mathcal{A}_j$, we will denote with $\sigma = (j, a)$ the \emph{sequence} of $i$'s actions encountered on the path from the root of the game until (and included) action $a$. 
For notational convenience, we will use the special symbol $\emptyseq$ to denote the \emph{empty sequence}. Then, $i$'s set of sequences is defined as $\Sigma_{i} \defeq \{(j, a) : j \in \mathcal{J}_{i}, a \in \mathcal{A}_j \} \cup \{ \emptyseq \}$; we will also use the notation $\Sigma_i^* \defeq \Sigma_{i} \setminus \{ \emptyseq \}$. For a given information set $j \in \mathcal{J}_{i}$ we will use $\sigma_{j} \in \Sigma_{i}$ to represent the \emph{parent sequence}; \emph{i.e.} the last sequence encountered by player $i$ before reaching any node in the information set $j$, assuming that it exists. Otherwise, we let $\sigma_{j} = \emptyseq$, and we say that $j$ is a \emph{root information set} of player $i$. 
A \emph{strategy} for a player specifies a probability distribution for every possible information set encountered in the game tree. For perfect-recall EFGs, strategies can be equivalently represented in \emph{sequence-form}:

\begin{definition}[Sequence-form Polytope]
    \label{definition:sequence_form}
The \emph{sequence-form strategy polytope} for player $i \in [n]$ is defined as the following (convex) polytope:
\begin{equation*}
    \mathcal{Q}_i \defeq \bigg\{ \vec{q}_i \in \mathbb{R}^{|\Sigma_{i}|}_{\geq 0} : \vec{q}_i[\emptyseq] = 1, \quad  \vec{q}_i[\sigma_j] = \sum_{a \in \mathcal{A}_j} \vec{q}_i[(j, a)], \quad \forall j \in \mathcal{J}_{i} \bigg\}.
\end{equation*}
\end{definition}

This definition ensures the probability mass conservation for the sequence-form strategies along every possible decision point. The probability of playing action $a$ at information set $j \in \mathcal{J}_{i}$ can be obtained by dividing $\vec{q}[(j, a)]$ by $\vec{q}[\sigma_j]$. Analogously, one can define the sequence-form strategy polytope for the \emph{subtree} of the partially ordered set $(\mathcal{J}_{i}, \prec)$ \emph{rooted} at $j \in \mathcal{J}_{i}$, which will be denoted by $\mathcal{Q}_j$. Moreover, the set of \emph{deterministic} sequence-form strategies for player $i \in [n]$ is the set $\Pi_{i} = \mathcal{Q}_{i} \cap \{0, 1\}^{|\Sigma_{i}|}$, and similarly for $\Pi_j$. A well-known implication of Kuhn's theorem~\cite{Kuhn53:Extensive} is that $\mathcal{Q}_{i} = \co \Pi_{i}$, and $\mathcal{Q}_j = \co \Pi_j$, for any $i \in [n]$ and $j \in \mathcal{J}_{i}$. 
The \emph{joint} set of deterministic sequence-form strategies of the players will be represented with $\Pi \defeq \bigtimes_{i \in [n]} \Pi_{i}$. As such, an element $\vec{\pi} \in \Pi$ is an $n$-tuple $(\vec{\pi}_{1}, \dots, \vec{\pi}_{n})$ specifying a deterministic sequence-form strategy for every player $i \in [n]$. Finally, we overload notation by representing the utility of player $i \in [n]$ under a profile $\vec{\pi} \in \Pi$ as
\begin{equation*}
    u_{i}(\vec{\pi}) \defeq \sum_{z \in \mathcal{Z}} p_{c}(z) u_{i}(z) \mathbbm{1} \{ \vec{\pi}_{k}[\sigma_{k, z}] = 1, \forall k \in [n]\},
\end{equation*}
where $\sigma_{i, z}$ denotes the last sequence of player $i$ before reaching the terminal node $z \in \mathcal{Z}$. For the convenience of the reader, in \Cref{tab:notation} we have summarized the main notation related to EFGs used throughout this paper.

\begin{figure}[t]
{\begin{minipage}{8.2cm}\vspace*{-5mm}
\begin{table}[H]
    \centering
    \setlength{\tabcolsep}{1mm}
    \scalebox{.9}{\begin{tabular}{lp{7.9cm}}
     & \textbf{Description}  \\ \toprule
     $\mathcal{J}_i$ & Information sets of player $i$ \\
     $\mathcal{A}_j$ & Set of actions at information set $j$ \\
     $\Sigma_i$ & Set of sequences of player $i$ \\
     $\Sigma_i^*$ & Set of sequences of player $i$ excluding $\emptyseq$ \\ 
     $\Sigma_j$ & Set of sequences at $j \in \mathcal{J}_i$ and all of its descendants \\
     $\mathfrak{D}_{i}$ & Maximum depth of any $j\in\mathcal{J}_i$ \\
     $\mathcal{Z}$ & Number of leaves \\
     $\| \mathcal{Q}_i \|_1$ & Upper bound on the $\ell_1$-norm of any $\vec{x} \in \mathcal{Q}_i$ \\
     \hline
     $\Pi_{i}$ & Deterministic sequence-form strategies of player $i$ \\
     $\Pi_j$ & Deterministic sequence-form strategies rooted at $j \in \mathcal{J}_{i}$ \\
     $\mathcal{Q}_{i}$ & Sequence-form strategies of player $i$ \\
     $\mathcal{Q}_{j}$ & Sequence-form strategies rooted at $j \in \mathcal{J}^{(i)}$ \\
     $\Pi$ & Set of joint deterministic sequence-form strategies \\ \bottomrule
    \end{tabular}}
    \caption{Summary of EFG notation.}
    \label{tab:notation}
\end{table}\vspace*{-9mm}
\end{minipage}}\hfill
{\begin{minipage}{5cm}\vspace*{-5mm}
\begin{figure}[H]
    \centering
    \scalebox{.87}{\begin{tikzpicture}[baseline=0cm]
       \def\dcha{1.5}
       \def\done{.7}
       \def\dtwo{.4}
       \node[chanode] (A) at (0, -.2) {};
       \node[pl1node] (B) at (-\dcha,-1) {};
       \node[pl1node] (C) at (\dcha,-1) {};
       \node[pl2node] (D) at (-\dcha-\done,-2) {};
       \node[pl2node] (E) at (-\dcha+\done,-2) {};
       \node[pl2node] (F) at (\dcha-\done,-2) {};
       \node[pl2node] (G) at (\dcha+\done,-2) {};
       \node[termina] (l1) at (-\dcha-\done-\dtwo,-3) {};
       \node[termina] (l2) at (-\dcha-\done+\dtwo,-3) {};
       \node[termina] (l3) at (-\dcha+\done-\dtwo,-3) {};
       \node[termina] (l4) at (-\dcha+\done+\dtwo,-3) {};
       \node[termina] (l5) at (\dcha-\done-\dtwo,-3) {};
       \node[termina] (l6) at (\dcha-\done+\dtwo,-3) {};
       \node[termina] (l7) at (\dcha+\done- \dtwo,-3) {};
       \node[termina] (l8) at (\dcha+\done+\dtwo,-3) {};

       \draw[semithick] (A) -- (B)
                  --node[fill=white,inner sep=.9] {\seq{1}} (D)
                  --node[fill=white,inner sep=.9] {\seq{5}} (l1);
       \draw[semithick] (B) --node[fill=white,inner sep=.9] {\seq{2}} (E)
                  --node[fill=white,inner sep=.9] {\seq{5}} (l3);
       \draw[semithick] (D) --node[fill=white,inner sep=.9] {\seq{6}} (l2);
       \draw[semithick] (E) --node[fill=white,inner sep=.9] {\seq{6}} (l4);
       \draw[semithick] (A) -- (C)
                  --node[fill=white,inner sep=.9] {\seq{3}} (F)
                  --node[fill=white,inner sep=.9] {\seq{7}} (l5);
       \draw[semithick] (C) --node[fill=white,inner sep=.9] {\seq{4}} (G)
                  --node[fill=white,inner sep=.9] {\seq{7}} (l7);
       \draw[semithick] (F) --node[fill=white,inner sep=.9] {\seq{8}} (l6);
       \draw[semithick] (G) --node[fill=white,inner sep=.9] {\seq{8}} (l8);

       \draw[black!60!white] (B) circle (.25);
       \node[black!60!white]  at ($(B) + (-.4, 0)$) {\Large\textsc{a}};

       \draw[black!60!white] (C) circle (.25);
       \node[black!60!white]  at ($(C) + (.4, 0)$) {\Large\textsc{b}};

       \draw[black!60!white] ($(D) + (0, .2)$) arc (90:270:.2);
       \draw[black!60!white] ($(D) + (0, .2)$) -- ($(E) + (0, .2)$);
       \draw[black!60!white] ($(D) + (0, -.2)$) -- ($(E) + (0, -.2)$);
       \draw[black!60!white] ($(E) + (0, -.2)$) arc (-90:90:.2);
       \node[black!60!white]  at ($(D) + (-.4, 0)$) {\Large\textsc{c}};

       \draw[black!60!white] ($(F) + (0, .2)$) arc (90:270:.2);
       \draw[black!60!white] ($(F) + (0, .2)$) -- ($(G) + (0, .2)$);
       \draw[black!60!white] ($(F) + (0, -.2)$) -- ($(G) + (0, -.2)$);
       \draw[black!60!white] ($(G) + (0, -.2)$) arc (-90:90:.2);
       \node[black!60!white]  at ($(G) + (.4, 0)$) {\Large\textsc{d}};
\end{tikzpicture}}  
    \caption{Example of a two-player EFG.}
    \label{fig:example}
\end{figure}
\end{minipage}}
\end{figure}

\paragraph{An Illustrative Example.} To further clarify some of the concepts we have introduced so far, we illustrate a simple two-player EFG in \Cref{fig:example}. Black nodes belong to player~$1$, white round nodes to player~$2$, square nodes are terminal nodes (aka leaves), and the crossed node is a chance node. Player $2$ has two information sets, $\mathcal{J}_2 \defeq \{\textsc{C}, \textsc{D} \}$, each containing two nodes. This captures the lack of knowledge regarding the action played by player $1$. In contrast, the outcome of the chance move is observed by both players. At the information set $\textsc{C}$, player $2$ has two possible actions, $\mathcal{A}_{\textsc{C}} \defeq \{\seq{5}, \seq{6}\}$. Thus, one possible sequence for player $2$ is the pair $\sigma = (\textsc{C}, \seq{5}) \in \Sigma_{2}$.

\subsection{Online Learning and Optimistic Regret Minimization}

Consider a convex and compact set $\mathcal{X} \subseteq \R^d$ representing the set of strategies of some agent. In the online decision making framework, a \emph{regret minimizer} $\mathcal{R}$ can be thought of as a black-box device which interacts with the external environment via the following two basic subroutines:

\begin{itemize}[nolistsep,itemsep=1mm,leftmargin=.8cm]
    \item $\mathcal{R}.\nextstr()$: The regret minimizer returns a strategy $\vec{x}^{(t)} \in \mathcal{X}$ at time $t$;
    \item $\mathcal{R}.\obsut(\ell^{(t)})$: The regret minimizer receives as feedback a \emph{linear utility function} $\ell^{(t)}: \mathcal{X} \ni \vec{x} \mapsto \langle \vec{\ell}^{(t)}, \vec{x} \rangle$, and may alter its internal state accordingly. 
\end{itemize}
The utility function $\ell^{(t)}$ could depend adversarially on the previous outputs $\vec{x}^{(1)}, \dots, \vec{x}^{(t-1)}$, but not on $\vec{x}^{(t)}$. The decision making is \emph{online} in the sense that the regret minimizer can adapt to previously received information, but no information about future utilities is available. The performance of a regret minimizer is typically measured in terms of its \emph{cumulative external regret} (or simply regret), defined, for a time horizon $T \in \N$, as follows.
\begin{equation}
    \label{eq:external_regret}
    \reg^T \defeq \max_{\vec{x}^* \in \mathcal{X}} \sum_{t=1}^T \langle \vec{x}^*, \vec{\ell}^{(t)} \rangle - \sum_{t=1}^T \langle \vec{x}^{(t)}, \vec{\ell}^{(t)} \rangle.
\end{equation}
That is, the performance of the online algorithm is compared to the best \emph{fixed} strategy in \emph{hindsight}. A regret minimizer is called \emph{Hannan consistent} if, under any sequence of (bounded) utility functions, its regret grows sublinearly with $T$; that is, $\reg^T = o(T)$. It is well-known that broad families of learning algorithms incur $O(\sqrt{T})$ regret under \emph{any} sequence of utility functions, which also matches the lower bound in the adversarial regime (see~\cite{Cesa06:Prediction}).

\paragraph{\texorpdfstring{Phi}{Phi}-Regret.} A conceptual generalization of external regret is the so-called \emph{Phi-regret}. In this framework the performance of the learning algorithm is measured based on a \emph{set of transformations} $\Phi \ni \phi: \mathcal{X} \to \mathcal{X}$, leading to the notion of (cumulative) $\Phi$-regret:
\begin{equation*}
    \reg^T \defeq \max_{\phi^* \in \Phi} \sum_{t=1}^T \langle \phi^*(\vec{x}^{(t)}), \vec{\ell}^{(t)} \rangle - \sum_{t=1}^T \langle \vec{x}^{(t)}, \vec{\ell}^{(t)} \rangle.
\end{equation*}

When the set of transformations $\Phi$ coincides with the set of \emph{constant functions} we recover the notion of external regret given in \eqref{eq:external_regret}. However, Phi-regret is substantially stronger and it yields more appealing notions of hindsight rationality~\citep{Gordon08:No}, incorporating the notion of \emph{swap regret}~\citep{Blum07:From}.

\paragraph{Optimistic Regret Minimization.} An emerging subfield of online learning (\cite{Chiang12:Online,Rakhlin13:Online}) studies the improved performance guarantees one can obtain when the utilities observed by the regret minimization algorithm possess additional structure, typically in the form of \emph{small variation}. Such considerations diverge from the adversarial regime we previously described, and are motivated---among others---by the fact that in many settings the utility functions are themselves selected by \emph{regularized learning algorithms}. For our purposes we shall employ the following definition, which is a modification of the $\texttt{RVU}$ property~\citep{Syrgkanis15:Fast}.
\begin{definition}[Stable-Predictive]
    \label{definition:stable-predictive}
Let $\mathcal{R}$ be a regert minimizer and $\|\cdot\|$ be any norm. $\mathcal{R}$ is said to be $\kappa$-\emph{stable} with respect to $\|\cdot\|$ if for all $t \geq 2$ the strategies output by $\mathcal{R}$ are such that
\begin{equation*}
    \| \vec{x}^{(t)} - \vec{x}^{(t-1)} \| \leq \kappa.
\end{equation*}
Moreover, $\mathcal{R}$ is said to be $(\alpha, \beta)$-\emph{predictive} with respect to $\|\cdot\|$ if its regret $\reg^T$ can be bounded as
\begin{equation}
    \label{eq:predictive}
    \reg^T \leq \alpha + \beta \sum_{t=1}^T \| \vec{\ell}^{(t)} - \vec{m}^{(t)} \|_*^2,
\end{equation}
for any sequence of utilities $\vec{\ell}^{(1)}, \dots, \vec{\ell}^{(T)}$, where $\|\cdot\|_*$ is the dual norm of $\|\cdot\|$.
\end{definition}
In the above definition $\vec{m}^{(t)}$ serves as the \emph{prediction} of the regret minimizer $\mathcal{R}$ at time $t \geq 1$. While traditional online algorithms are not known to satisfy \eqref{eq:predictive}, we will next present natural variants which are indeed stable-predictive in the sense of \Cref{definition:stable-predictive}. 


\paragraph{Optimistic Follow the Regularized Leader.} 
Let $d$ be a $1$-strongly convex \emph{distance generating function (DGF)} with respect to a norm $\|\cdot\|$, and $\eta > 0$ be the \emph{learning rate}. The update rule of \emph{optimistic follow the regularized leader (OFTRL)} takes the following form for $t \geq 2$:
\begin{equation}
    \label{eq:oftrl}
    \tag{OFTRL}
    \vec{x}^{(t)} \defeq \argmax_{\vec{x} \in \mathcal{X}} \left\{ \left\langle \vec{x}, \vec{m}^{(t)} + \sum_{\tau=1}^{t-1} \vec{\ell}^{(\tau)} \right\rangle - \frac{d(\vec{x})}{\eta} \right\},
\end{equation}
where $\vec{m}^{(t)}$ is the prediction at time $t$, and $\vec{x}^{(1)} \defeq \argmin_{\vec{x} \in \mathcal{X}} d(\vec{x})$. Unless specified otherwise, it will be tacitly assumed that $\vec{m}^{(t)} \defeq \vec{\ell}^{(t-1)}$, for $t \geq 1$, where we conventionally let $\vec{\ell}^{(0)} \defeq \vec{0}$. \citet{Syrgkanis15:Fast} established the following property:
\begin{lemma}
    \label{lemma:oftrl-predictive}
    \eqref{eq:oftrl} is $(\Omega_d/\eta, \eta)$-predictive\footnote{\citet{Syrgkanis15:Fast} only stated this for the simplex, but their proof readily extends to arbitrary convex and compact sets.} with respect to any norm $\|\cdot\|$ for which $d$ is $1$-strongly convex, where $\Omega_d$ is the range of $d$ on $\mathcal{X}$, that is, $\Omega_d \defeq \max_{\vec{x},\vec{x}'\in\mathcal{X}} \{d(\vec{x}) - d(\vec{x}')\}$.
\end{lemma}
The \emph{entropic regularizer} on the simplex is defined as $d(\vec{x}) \defeq \sum_{k=1}^d \vec{x}[k] \log \vec{x}[k]$, and it is well-known to be $1$-strongly convex with respect to the $\ell_1$-norm. 
This OFTRL setup will be referred to as \emph{optimistic multiplicative weights updates (OMWU)}.\footnote{When $\vec{m}^{(t)} \defeq \vec{0}$, for all $t \geq 1$, we recover \emph{multiplicative weights updates (MWU)}.}


We will also require a suitable DGF for the sequence-form polytope. To this end, we will employ the \emph{dilatable global entropy} DGF, recently introduced by~\citet{Farina21:Better}. 

\begin{definition}[\citep{Farina21:Better}]
    \label{definition:global-dgf}
The \emph{dilatable global entropy distance generating function} $d : \mathcal{Q} \to \R_{\geq 0}$ is defined as 
\begin{equation*}
    d(\vec{x}) \defeq \sum_{\sigma \in \Sigma } \vec{w}[\sigma] \vec{x}[\sigma] \log (\vec{x}[\sigma]). 
\end{equation*}
The vector $\vec{w} \in \R^{|\Sigma|}$ is defined recursively as 
\begin{align*}
    \vec{w}[\emptyseq] &= 1; \\
    \vec{w}[(j,a)] &= \vec{\gamma}[j] - \sum_{j': \sigma_{j'} = (j,a)} \vec{\gamma}[j'], \quad \forall (j, a) \in \Sigma,
\end{align*}
where
\begin{equation}
    \label{eq:gamma}
    \vec{\gamma}[j] = 1 + \max_{a \in \mathcal{A}_j} \left\{ \sum_{j': \sigma_{j'} = (j,a) } \vec{\gamma}[j'] \right\}, \quad \forall j \in \mathcal{J}.
\end{equation}
\end{definition}

This DGF is ``nice'' (in the parlance of \citet{Hoda10:Smoothing}) since its gradient, as well as the gradient of its convex conjugate, can be computed \emph{exactly} in linear time in $|\Sigma|$---the dimension of the domain. Our analysis will require the following characterization.

\begin{lemma}[\cite{Farina21:Better}]
    \label{lemma:dge}
    The dilatable global entropy $d$ of \Cref{definition:global-dgf} is a DGF for the sequence-form polytope $\mathcal{Q}$. Moreover, it is $1/\|\mathcal{Q} \|_1$-strongly convex on $\relint \mathcal{Q}$ with respect to the $\|\cdot\|_1$ norm, where $\|\mathcal{Q} \|_1 = \max_{\vec{q} \in \mathcal{Q}}\|\vec{q}\|_1$. Finally, the $d$-diameter $\Omega_d$ of $\mathcal{Q}$ is at most $\| \mathcal{Q} \|_1^2 \max_{j \in \mathcal{J}} \log |\mathcal{A}_j|$.
\end{lemma}

In the sequel we will instantiate \eqref{eq:oftrl} with dilatable global entropy as DGF to construct a stable-predictive regret minimizer for the sequence-form strategy polytope. 

\subsection{Extensive-Form Correlated and Coarse Correlated Equilibrium} 

In this subsection we introduce the notion of an \emph{extensive-form correlated and coarse correlated equilibrium} (henceforth EFCE and EFCCE respectively). First, for EFCE we will work with the definition used in \citep{Farina19:Correlation}, which is equivalent to the original one due to \citet{Stengel08:Extensive}. To this end, let us introduce the concept of a \emph{trigger deviation function}.
\begin{definition}
    \label{definition:trigger_deviation_functions}
Consider some player $i \in [n]$. A \emph{trigger deviation function} with respect to a \emph{trigger sequence} $\hat{\sigma} = (j,a) \in \Sigma_i^*$ and a \emph{continuation strategy} $\hat{\vec{\pi}}_i \in \Pi_j$ is any linear function $f : \R^{|\Sigma_{i}|} \to \R^{|\Sigma_{i}|}$ with the following properties.
\begin{itemize}
    \item Any strategy $\vec{\pi}_i \in \Pi_{i}$ which does not prescribe the sequence $\hat{\sigma}$ remains invariant. That is, $f(\vec{\pi}_i) = \vec{\pi}_i$ for any $\vec{\pi}_i \in \Pi_{i}$ such that $\vec{\pi}_i[\hat{\sigma}] = 0$;
    \item Otherwise, the prescribed sequence $\hat{\sigma} = (j, a)$ is modified so that the behavior at $j$ and all of its descendants is replaced by the behavior specified by the continuation strategy:
    \begin{equation*}
        f(\vec{\pi}_i)[\sigma] = 
        \begin{cases}
        \vec{\pi}_i[\sigma] & \textrm{  if  } \sigma \not\succeq j ;\\
        \hat{\vec{\pi}}_i[\sigma] & \textrm{  if  } \sigma \succeq j,
        \end{cases}
    \end{equation*}
    for all $\sigma \in \Sigma_{i}$.
\end{itemize}
\end{definition}
%
We will let $\Psi_{i} \defeq \{ \phi_{\hat{\sigma} \rightarrow \hat{\vec{\pi}}_i} : \hat{\sigma} = (j, a) \in \Sigma_{i}^*, \hat{\vec{\pi}}_i \in \Pi_j \}$ be the set of all possible (linear) mappings defining trigger deviation functions for player $i$. We are ready to introduce the concept of $\efce$.


\begin{definition}[EFCE]
    \label{definition:efce}
A probability distribution $\vec{\mu} \in \Delta(\Pi)$ is an $\epsilon$-EFCE, for $\epsilon \geq 0$, if for every player $i \in [n]$ and every trigger deviation function $\phi_{\hat{\sigma} \rightarrow \hat{\vec{\pi}}_i} \in \Psi_{i}$,
\begin{equation*}
    \E_{\vec{\pi} \sim \vec{\mu}} \left[ u_{i} \left( \phi_{\hat{\sigma} \rightarrow \hat{\vec{\pi}}_i} (\vec{\pi}_{i}), \vec{\pi}_{-i} \right) - u_{i}(\vec{\pi}) \right] \leq \epsilon,
\end{equation*}
where $\vec{\pi} = (\vec{\pi}_1, \dots, \vec{\pi}_n) \in \Pi$. We say that $\vec{\mu} \in \Delta(\Pi)$ is an $\efce$ if it is a $0$-EFCE.
\end{definition}

\begin{theorem}[\cite{Farina21:Simple}]
    \label{theorem:EFCE-convergence}
    Suppose that for every player $i \in [n]$ the sequence of deterministic sequence-form strategies $\vec{\pi}_i^{(1)}, \dots, \vec{\pi}_i^{(T)} \in \Pi_{i}$ incurs $\Psi_{i}$-regret at most $\reg_i^{T}$ under the sequence of linear utility functions
    \begin{equation*}
        \ell_i^{(t)} : \Pi_{i} \ni \vec{\pi}_{i} \mapsto u_{i} \left( \vec{\pi}_{i}, \vec{\pi}_{-i}^{(t)} \right).
    \end{equation*}
    Then, the correlated distribution of play $\vec{\mu} \in \Delta(\Pi)$ is an $\epsilon$-EFCE, where $\epsilon \defeq \frac{1}{T} \max_{i \in [n]} \reg_i^{T}$.
\end{theorem}

Similarly, we introduce the closely related notion of a \emph{coarse} trigger deviation function.

\begin{definition}[Coarse Trigger Deviation Functions]
    \label{definition:coarse_trigger_deviation_functions}
Consider some player $i \in [n]$. A \emph{coarse trigger deviation function} with respect to an information set $j \in \mathcal{J}_i$ and a continuation strategy $\hat{\vec{\pi}}_i \in \Pi_j$ is any linear function $f : \R^{|\Sigma_{i}|} \to \R^{|\Sigma_{i}|}$ with the following properties:
\begin{itemize}
    \item For any $\vec{\pi}_i \in \Pi_{i}$ such that $\vec{\pi}_i[\sigma_j] = 0$ it holds that $f(\vec{\pi}_i) = \vec{\pi}_i$;
    \item Otherwise,
    \begin{equation*}
        f(\vec{\pi}_i)[\sigma] = 
        \begin{cases}
        \vec{\pi}_i[\sigma] & \textrm{  if  } \sigma \not\succeq j ;\\
        \hat{\vec{\pi}}_i[\sigma] & \textrm{  if  } \sigma \succeq j,
        \end{cases}
    \end{equation*}
    for all $\sigma \in \Sigma_i$.
\end{itemize}
\end{definition}
We will also let $\widetilde{\Psi}_{i} \defeq \{ \phi_{j \rightarrow \hat{\vec{\pi}}_i} : j\in \mathcal{J}_{i}, \hat{\vec{\pi}}_i \in \Pi_j \}$ be the set of all (linear) mappings inducing coarse trigger deviations functions for player $i$.
\begin{definition}[EFCCE]
    \label{definition:efcce}
A probability distribution $\vec{\mu} \in \Delta(\Pi)$ is an $\epsilon$-EFCCE, for $\epsilon \geq 0$, if for every player $i \in [n]$ and every coarse trigger deviation function $\phi_{j \rightarrow \hat{\vec{\pi}}_i} \in \widetilde{\Psi}_{i}$,
\begin{equation*}
    \E_{\vec{\pi} \sim \vec{\mu}} \left[ u_{i} \left( \phi_{j \rightarrow \hat{\vec{\pi}}_i} (\vec{\pi}_{i}), \vec{\pi}_{-i} \right) - u_{i}(\vec{\pi}) \right] \leq \epsilon,
\end{equation*}
where $\vec{\pi} = (\vec{\pi}_1, \dots, \vec{\pi}_n) \in \Pi$. We say that $\vec{\mu} \in \Delta(\Pi)$ is an EFCCE if it is a $0$-EFCCE.
\end{definition}

Analogously to \Cref{theorem:EFCE-convergence}, we show (in \Cref{appendix:proofs_prel}) that if all players employ a $\widetilde{\Psi}_i$-regret minimizer, the correlated distribution of play converges to an EFCCE. 

\begin{restatable}{theorem}{efccereg}
    \label{theorem:EFCCE-convergence}
    Suppose that for every player $i \in [n]$ the sequence of deterministic sequence-form strategies $\vec{\pi}_i^{(1)}, \dots, \vec{\pi}_i^{(T)} \in \Pi_{i}$ incurs $\widetilde{\Psi}_{i}$-regret at most $\reg_i^{T}$ under the sequence of linear utility functions
    \begin{equation*}
        \ell_i^{(t)} : \Pi_{i} \ni \vec{\pi}_{i} \mapsto u_{i} \left( \vec{\pi}_{i}, \vec{\pi}_{-i}^{(t)} \right).
    \end{equation*}
    Then, the correlated distribution of play $\vec{\mu} \in \Delta(\Pi)$ is an $\epsilon$-EFCCE, where $\epsilon \defeq \frac{1}{T} \max_{i \in [n]} \reg_i^{T}$.
\end{restatable}





\section{Accelerating Phi-Regret Minimization with Optimism}
\label{section:accelerating-Phi}

In this section we present a general construction for obtaining improved Phi-regret guarantees. Our template is then instantiated in \Cref{sec:efce,section:efcce} to obtain faster dynamics for EFCE and EFCCE.

Our approach combines the framework of \citet{Gordon08:No} with stable-predictive (aka. optimistic) regret minimization. 
As in \cite{Gordon08:No}, we combine 1) a regret minimizer that outputs a linear transformation $\phi^{(t)} \in \Phi$ at every time $t$, and 2) a fixed-point oracle for each $\phi^{(t)} \in \Phi$. However, our construction further requires that 2) is stable (in the sense of \Cref{definition:stable-predictive}). To achieve this, we will focus on regret minimizers having the following property. 
\begin{definition}
    \label{definition:FP-smooth}
    Consider a set of functions $\Phi$ such that $\phi(\mathcal{X}) \subseteq \mathcal{X}$ for all $\phi \in \Phi$, and a no-regret algorithm $\mathcal{R}_{\Phi}$ for the set of transformations $\Phi$ which returns a sequence $(\phi^{(t)})$. We say that $\mathcal{R}_{\Phi}$ is \emph{fixed point $\kappa$-stable} with respect to a norm $\|\cdot\|$ if the following conditions hold.
\begin{itemize}
    \item Every $\phi^{(t)}$ admits a fixed point. That is, there exists $\vec{x}^{(t)} \in \mathcal{X}$ such that $\phi^{(t)}(\vec{x}^{(t)}) = \vec{x}^{(t)}$.
    \item For $\vec{x}^{(t)}$ with $\vec{x}^{(t)} = \phi^{(t)}(\vec{x}^{(t)})$, there is $\vec{x}^{(t+1)} = \phi^{(t+1)}(\vec{x}^{(t+1)})$ such that $\|\vec{x}^{(t+1)} - \vec{x}^{(t)} \| \leq \kappa$.
\end{itemize}
\end{definition}

In this context, we will show how to construct a stable-predictive $\Phi$-regret minimizer starting from the following two components.
\begin{enumerate}
    \item $\mathcal{R}_{\Phi}$: An $(A, B)$-predictive fixed point $\kappa$-stable regret minimizer for the set $\Phi$;
    \item \label{item:FP} $\fporacle(\phi ; \widetilde{\vec{x}}, \kappa, \epsilon)$: A \emph{stable fixed point oracle} which returns a point $\vec{x} \in \mathcal{X}$ such that (i) $\|\phi(\vec{x}) - \vec{x}\| \leq \epsilon$, and (ii) $\|\vec{x} - \widetilde{\vec{x}}\| \leq \kappa$ (the existence of such a fixed point is guaranteed by the fixed point $\kappa$-stability assumption on the regret minimizer). 
\end{enumerate}


\begin{restatable}[Stable-Predictive Phi-Regret Minimization]{theorem}{phiaccel}
    \label{theorem:accelerating-Phi}
    Consider an $(A, B)$-predictive regret minimizer $\mathcal{R}_{\Phi}$ with respect to $\|\cdot\|_{1}$ for a set of linear transformations $\Phi$ on $\mathcal{X}$.  Moreover, suppose that $\mathcal{R}_{\Phi}$ is fixed point $\kappa$-stable. Then, if we have access to a $\fporacle$, we can construct a $\kappa$-stable algorithm with $\Phi$-regret $\reg^T$ bounded as
\begin{equation*}
    \reg^T \leq A + 2 B \sum_{t=1}^T \| \vec{\ell}^{(t)} - \vec{\ell}^{(t-1)}\|^2_{\infty} + 2 B \|\vec{\ell}\|_\infty^2 \sum_{t=1}^T \|\vec{x}^{(t)} - \vec{x}^{(t-1)}\|^2_\infty  + \|\vec{\ell}\|_\infty \sum_{t=1}^T \epsilon^{(t)},
\end{equation*}
where $\epsilon^{(t)}$ is the error of $\fporacle$ at time $t$, and $\|\vec{\ell}^{(t)}\|_\infty \leq \|\vec{\ell}\|_\infty$ for any $t \geq 1$. It is also assumed that $\|\vec{x}\|_{\infty} \leq 1$ for all $\vec{x} \in \mathcal{X}$.
\end{restatable}
The $\ell_1$ norm is used only for convenience; the theorem readily extends under any equivalent norm. The proof of \Cref{theorem:accelerating-Phi} builds on the construction of \citet{Gordon08:No}, and it is included in \Cref{appendix:proof-accelerating_phi}.

\section{Faster Convergence to EFCE}
\label{sec:efce}

Our framework (\Cref{theorem:accelerating-Phi}) reduces accelerating $\Phi$-regret minimization to (i) developing a predictive regret minimizer for the set $\Phi$, and (ii) establishing the stability of the fixed points ($\fporacle$). In this section we establish these components for the set of all possible trigger deviations functions (\Cref{definition:trigger_deviation_functions}), leading to faster convergence to EFCE. In particular, \Cref{section:Phi-regret} is concerned with the former task while \Cref{section:stability} is concerned with the latter.

\subsection{Constructing a Predictive Regret Minimizer for \texorpdfstring{$\Psi_{i}$}{Psi-i}}
\label{section:Phi-regret}

Here we develop a regret minimizer for the set $\co \Psi_{i}$, the convex hull of all trigger deviation functions (\Cref{definition:trigger_deviation_functions}) of player $i \in [n]$. Given that $\co \Psi_{i} \supseteq \Psi_{i}$, this will immediately imply a $\Psi_i$-regret minimizer---after applying \Cref{theorem:accelerating-Phi}. To this end, the set $\co \Psi_{i}$ can be evaluated in two stages. First, for a fixed sequence $\hat{\sigma} = (j, a) \in \Sigma_{i}^*$ we define the set $\Psi_{\hat{\sigma}} \defeq \co \left\{ \phi_{\hat{\sigma} \rightarrow \hat{\vec{\pi}}_i} : \hat{\vec{\pi}}_i \in \Pi_j \right\}$. Then, we take the convex hull of all $\Psi_{\hat{\sigma}}$; that is, $\co \Psi_{i} = \co \{ \Psi_{\hat{\sigma}} : \hat{\sigma} \in \Sigma^*_i \}$. In light of this, we first develop a predictive regret minimizer for the set $\Psi_{\hat{\sigma}}$, for any $\hat{\sigma} \in \Sigma_i^*$. These individual regret minimizers are then combined using a \emph{regret circuit} to conclude the construction in \Cref{theorem:co-circuit}. The overall algorithm is illustrated in \Cref{fig:algo}. All of the omitted poofs and pseudocode for this section are included in \Cref{appendix:proof-Phi-regret}.

\begin{figure}[ht]
    \centering
    \scalebox{.90}{\begin{tikzpicture}[yscale=.87]
    \tikzstyle{arrow}=[semithick,->];
    \tikzstyle{box}=[black!60,thick,rounded corners=1mm];
    \tikzstyle{loss}=[red];
    \tikzstyle{iter}=[blue];
    \tikzstyle{lbl}=[fill=white,inner ysep=0pt];
    \tikzstyle{hl}=[line width=2mm,white];
    \tikzstyle{hlbox}=[box,line width=1.8mm,gray!15];
    \tikzstyle{dot}=[circle,fill=black,draw=none,inner sep=.5mm];
    \draw[hlbox] (-.5,-.1) rectangle (12.2,5.4);
    \draw[box,thick] (-.5,-.1) rectangle (12.2,5.4) node[fitting node] (outer) {};
    
    \draw[hlbox] (1.2,.5) rectangle (9.5,5);
    \draw[box,thick] (1.2,.5) rectangle (9.5,5) node[fitting node] (external) {};
    
    \draw[hlbox] (1.7, 3.0) rectangle +(4,1.5);
    \draw[box] (1.7, 3.0) rectangle +(4,1.5) node[fitting node] (Rsigma1) {};
    
    \draw[hlbox] (1.7, 0.7) rectangle +(4,1.5);
    \draw[box] (1.7, 0.7) rectangle +(4,1.5) node[fitting node] (Rsigma2) {};
    \draw[box] ($(Rsigma1.center)!.5!(Rsigma2.center)+(3.0,-.5)$) rectangle +(1.8,1) node[fitting node] (mixer) {};
    \node(XXX) at ({$(external.east)!.51!(outer.east)$} |- mixer.center) {};
    \draw[box] ($(XXX)-(.7,.50)$) rectangle +(1.4,1.0) node[fitting node] (fp) {};
    
    \draw[box,semithick,fill=black!5] ($(Rsigma1.center)-(.95,.55)$) rectangle +(1.9,.9) node[fitting node] (CFR1) {};
    \node[text width=1.2cm,align=center] at (CFR1.center) {\fontsize{8}{6}\selectfont OFTRL};
    
    \draw[box,semithick,fill=black!5] ($(Rsigma2.center)-(.95,.55)$) rectangle +(1.9,.9) node[fitting node] (CFR2) {};
    \node[text width=1.2cm,align=center] at (CFR2.center) {\fontsize{8}{6}\selectfont OFTRL};

    \coordinate (Rsigma1_in)  at (CFR1 -| Rsigma1.west);
    \coordinate (Rsigma1_out) at (CFR1 -| Rsigma1.east);
    \coordinate (Rsigma2_in)  at (CFR2 -| Rsigma2.west);
    \coordinate (Rsigma2_out) at (CFR2 -| Rsigma2.east);
    
    \node[ox] (t4) at ($(Rsigma1_in)!.5!(CFR1.west)$) {};   
    \node[ox] (t5) at ($(CFR1.east)!.5!(Rsigma1_out)$) {};   
    \node[ox] (t6) at ($(Rsigma2_in)!.5!(CFR2.west)$) {};   
    \node[ox] (t7) at ($(CFR2.east)!.5!(Rsigma2_out)$) {};  
    
    \node at (mixer.center) {\small$\mathcal{R}_{\Delta}$ (OMWU)};

    \draw[arrow,loss] (t4) -- (CFR1.west);
    \draw[arrow,loss] (t6) -- (CFR2.west);
    
    \draw[arrow,iter] (CFR1.east) -- (t5);
    \draw[arrow,iter] (CFR2.east) -- (t7);
    
    \node at (Rsigma2.center) {};
    \node[lbl,yshift=-.2mm] at (Rsigma1.north) {\small$\mathcal{R}_{\seq{1}}$ (\cref{proposition:R_sigma})};
    \node[lbl,yshift=-.1mm] at (Rsigma2.north) {\small$\mathcal{R}_{\seq{m}}$ (\cref{proposition:R_sigma})};

    \node[ox] (t1) at ({$(outer.west)!.25!(external.west)$}|-fp) {};
    \node[dot] (d1) at ($(t1)+(.8,0)$) {};
    \draw[hl] ($(t1) - (1.3,0)$) -- (t1);
    \draw[arrow,loss] ($(t1) - (1.3,0)$) node[above right,xshift=-1mm] {$\vec{\ell}_i^{(t)}$} -- (t1);
    \draw[arrow,loss] (t1) --node[pos=.8,above]{$L_i^{(t)}$} (d1);
    \draw[hl] (d1) to[out=30,in=180] (Rsigma1_in) -- (t4);
    \draw[hl] (d1) to[out=-30,in=180] (Rsigma2_in) -- (t6);
    \draw[arrow,loss] (d1) to[out=30,in=180] (Rsigma1_in) -- (t4);
    \draw[arrow,loss] (d1) to[out=-30,in=180] (Rsigma2_in) -- (t6);
    
    \node[ox] (t2) at (mixer.east-|{$(external.east)-(.4,0)$}) {};

    \draw[arrow,iter] (mixer.east) -- (t2);
    \draw[hl] (t2) -- (fp.west);
    \node[above left,iter,inner xsep=0pt,fill=white,xshift=-.8mm] at (fp.west) {\small$\phi_i^{(t)}$};
    \draw[arrow,iter] (t2) -- (fp.west);
    \draw[hl] (fp.east) -- +(.9,0);
    \node[above,iter,inner xsep=0,fill=white,xshift=0mm] at ($(fp.east) + (.7,0)$) {$\vec{x}_i^{(t)}$};
    \node[dot] (d2) at ($(fp.east) + (.3,0)$) {};
    \draw[arrow,iter] (fp.east) -- (d2);
    \draw[arrow,iter] (d2) -- +(0.7,0);
    \draw[arrow,iter,rounded corners=.5mm] (d2) |- ($(t1)+(0,-2.4)$) -- (t1);
    
    \node[ox] (t3) at ($(mixer.west) - (.5,0)$) {};
    \draw[hl,shorten <= 2mm] (d1) -- (t3);
    \draw[arrow,loss] (d1) -- (t3);
    
    \node[fill=white,inner ysep=0pt,yshift=-.3mm] at ($(external.north)+(2,0)$) {\small$\mathcal{R}_{\Psi_i}$ (\cref{theorem:co-circuit})};
    \node[fill=white,inner ysep=0pt] at ($(outer.north) -(3.6,0)$) {\small$\Psi_{i}$-Regret Minimizer for $\mathcal{Q}_{i}$};

    \draw[white,line width=3.5mm] ($(Rsigma1.south)!.70!(Rsigma2.north)$) -- ($(Rsigma1.south)!.30!(Rsigma2.north)$);
    \fill[black] ($(Rsigma1.south)!.70!(Rsigma2.north)$) circle (.3mm);
    \fill[black] ($(Rsigma1.south)!.50!(Rsigma2.north)$) circle (.3mm);
    \fill[black] ($(Rsigma1.south)!.30!(Rsigma2.north)$) circle (.3mm);
    
    \node[dot] (d3) at (t3 |- Rsigma1_out) {};
    \node[dot] (d4) at (t3 |- Rsigma2_out) {};
    \draw[arrow,iter] (d3) -- (t3);
    \draw[arrow,iter] (d4) -- (t3);
    \draw[arrow,loss] (t3) -- (mixer.west);

    \node[iter] at (10.3,.6) {$\vec{x}_i^{(t)}$};
    
    \draw[arrow,iter,rounded corners=.5mm] (d3) --node[above]{$\phi_{\seq{1}\to \vec{q}^{(t)}_{\seq{1}}}$} (Rsigma1_out -| t2) -- (t2);
    \draw[arrow,iter,rounded corners=.5mm] (d4) --node[below=-.5mm]{$\phi_{\seq{m}\to \vec{q}^{(t)}_{\seq{m}}}$} (Rsigma2_out -| t2) -- (t2);
    
    \draw[hl] (t5) -- (d3);
    \draw[hl] (t7) -- (d4);
    
    \draw[arrow,iter] (t5) -- (d3);
    \draw[arrow,iter] (t7) -- (d4);

    \node[above=-.6mm] at (fp.center) {\small Fixed};
    \node[below=-.6mm] at (fp.center) {\small point};
\end{tikzpicture}}
    \caption{An overview of the overall construction. For notational convenience we have let $\Sigma^*_i \defeq \{ \seq{1}, \seq{2}, \dots, \seq{m}\}$. The symbol $\otimes$ in the figure denotes a multilinear transformation. We have used blue color for the iterates and red for the utilities. The algorithm first constructs a regret minimizer $\mathcal{R}_{\Psi_i}$ for the set $\Psi_i$ (\Cref{theorem:co-circuit}). This internally uses a regret minimizer $\mathcal{R}_{\Delta}$ which ``mixes'' the strategies of $\mathcal{R}_{\seq{1}}, \dots, \mathcal{R}_{\seq{m}}$. In turn, the latter regret minimizers internally employ \eqref{eq:oftrl} with dilatable global entropy as DGF (\Cref{proposition:R_sigma}). The last step can also be implemented using stable-predictive CFR (\Cref{theorem:opt-cfr}), as we leverage for our experiments. Finally, $\mathcal{R}_{\Psi_i}$ is used to construct a stable-predictive $\Psi_i$-regret minimizer using the construction of \Cref{theorem:accelerating-Phi}.}
    \label{fig:algo}
\end{figure}
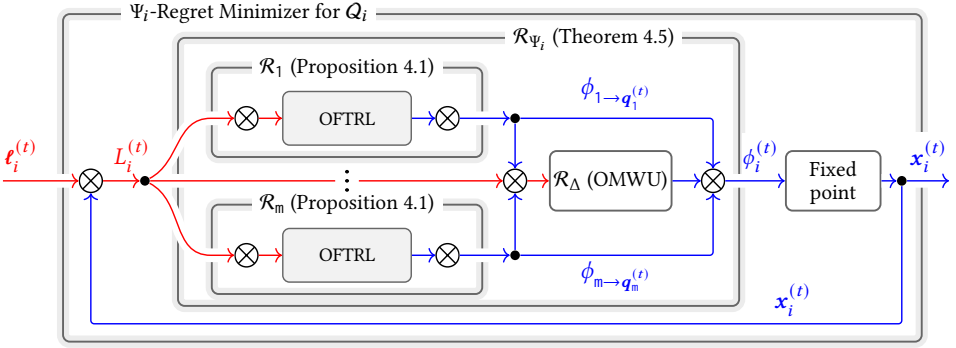

\subsubsection{Predictive Regret Minimizer for the set \texorpdfstring{$\Psi_{\hat{\sigma}}$}{Psi}.} 

Consider a sequence $\hat{\sigma} \in \Sigma_i^*$. We claim that the set of transformations $\Psi_{\hat{\sigma}} \defeq \co \left\{ \phi_{\hat{\sigma} \rightarrow \hat{\vec{\pi}}_i} : \hat{\vec{\pi}}_i \in \Pi_j \right\}$
is the image of $\mathcal{Q}_j$ under the affine mapping $h_{\hat{\sigma}} : \vec{q} \mapsto \phi_{\hat{\sigma} \rightarrow \vec{q}}$. 
Hence, it is not hard to see that a regret minimizer for $\Psi_{\hat{\sigma}}$ can be constructed starting from a regret minimizer for $\mathcal{Q}_j$. We now show that the predictive bound is preserved through this construction. 

\begin{restatable}{proposition}{firstcircuit}
    \label{proposition:R_sigma}
    Consider a player $i \in [n]$ and any trigger sequence $\hat{\sigma} = (j, a) \in \Sigma^*_i$. There exists an algorithm which constructs a regret minimizer $\mathcal{R}_{\hat{\sigma}}$ with access to an $(A, B)$-predictive regret minimizer $\mathcal{R}_{\mathcal{Q}_j}$ for the set $\mathcal{Q}_j$ such that $\mathcal{R}_{\hat{\sigma}}$ is $(A, B)$-predictive. 
\end{restatable}

This proposition requires a predictive regret minimizer for the set $\mathcal{Q}_j$, for each $j \in \mathcal{J}_i$. To this end, we instantiate \eqref{eq:oftrl} with dilatable global entropy as DGF (\Cref{definition:global-dgf}). Then, combining \Cref{lemma:oftrl-predictive} with \Cref{lemma:dge} leads to the following predictive bound.

\begin{lemma}
    \label{lemma:predictive-OMD}
    Suppose that the regret minimizer $\mathcal{R}_{\mathcal{Q}_j}$ is instantiated with dilatable global entropy. Then, $\mathcal{R}_{\mathcal{Q}_j}$ is $(A, B)$-predictive with respect to $\|\cdot\|_1$, where $A = \frac{\|\mathcal{Q}_i\|^2_1 \max_{j \in \mathcal{J}_i} \log |\mathcal{A}_j|}{\eta}$ and $B = \eta \|\mathcal{Q}_i\|_1$. 
\end{lemma}

The discrepancy between this bound and the one in \Cref{lemma:oftrl-predictive} derives from the fact that the modulus of convexity with respect to $\|\cdot\|_1$ for the dilatable global entropy is $1/\|\mathcal{Q}_i\|_1$ instead of $1$. Alternatively, we also establish a predictive variant of CFR which can be used in place of OFTRL for performing regret minimization over the set $\mathcal{Q}_j$.

\begin{proposition}[Predictive CFR; Full Version in \Cref{theorem:opt-cfr}]
    \label{proposition:cfr}
    There exists a variant of CFR using OMWU which is $(A, B)$-predictive, where $A = O(\frac{ \max_{j \in \mathcal{J}} \log|\mathcal{A}_j|}{\eta} \|\mathcal{Q}\|_1 )$ and $B = O(\eta \|\mathcal{Q}\|_1^3)$. 
\end{proposition}

This construction follows the approach of~\citet{Farina19:Stable}, but here we make the dependencies on the size of the game explicit. The predictive bound we obtain for CFR is inferior to the one in \Cref{lemma:predictive-OMD}, so the rest of our theoretical analysis will follow the ``global'' approach. 



\subsubsection{Predictive Regret Minimizer for \texorpdfstring{$\co \Psi_{i}$}{co Psi}.} 

The next step consists of appropriately combining the regret minimizers $\Psi_{\hat{\sigma}}$, for all $\hat{\sigma} \in \Sigma_i^*$, to a composite regret minimizer for the set $\co \Psi_{i}$. To this end, we will use a \emph{regret circuit} for the convex hull, formally introduced below.

\begin{proposition}[\cite{Farina19:Regret}]
    \label{proposition:regret_circuit-co}
    Consider a collection of sets $\mathcal{X}_1, \dots, \mathcal{X}_m$, and let $\mathcal{R}_i$ be a regret minimizer for the set $\mathcal{X}_i$, for each $i \in [m]$. Moreover, let $\mathcal{R}_{\Delta}$ be a regret minimizer for the $m$-simplex $\Delta^m$. A regret minimizer $\mathcal{R}_{\co}$ for the set $\co \{\mathcal{X}_1, \dots, \mathcal{X}_m\}$ can be constructed as follows.
    \begin{itemize}
        \item $\mathcal{R}_{\co}.\nextstr$ obtains the next strategy $\vec{x}_i^{(t)}$ of each regret minimizer $\mathcal{R}_i$, as well as the next strategy $\vec{\lambda}^{(t)} = (\vec{\lambda}^{(t)}[1], \dots, \vec{\lambda}^{(t)}[m]) \in \Delta^m$ of $\mathcal{R}_{\Delta}$, and returns the corresponding convex combination: $\vec{\lambda}^{(t)}[1] \vec{x}_1^{(t)} + \dots + \vec{\lambda}^{(t)}[m] \vec{x}_m^{(t)}$.
        \item $\mathcal{R}_{\co}.\obsut(L^{(t)})$ forwards $L^{(t)}$ to each of the regret minimizers $\mathcal{R}_1, \dots, \mathcal{R}_m$, while it forwards the utility function $(\vec{\lambda}[1], \dots, \vec{\lambda}[m]) \mapsto \vec{\lambda}[1] L^{(t)}(\vec{x}_1^{(t)}) + \dots + \vec{\lambda}[m] L^{(t)}(\vec{x}_m^{(t)})$ to $\mathcal{R}_{\Delta}$.
    \end{itemize}
    Then, if $\reg^T_1, \dots, \reg^T_m$ are the regrets accumulated by the regret minimizers $\mathcal{R}_1, \dots, \mathcal{R}_m$, and $\reg^T_{\Delta}$ is the regret of $\mathcal{R}_{\Delta}$, the regret $\reg^T_{\co}$ of the composite regret minmizers $\mathcal{R}_{\co}$ can be bounded as
    \begin{equation*}
        \reg^T_{\co} \leq \reg^T_{\Delta} + \max \{\reg_1^T, \dots, \reg_m^T\}.
    \end{equation*}
\end{proposition}

Next, we leverage this construction to obtain the main result of this subsection: a predictive regret minimizer for the set of transformations $\co \Psi_i$.

\begin{restatable}{theorem}{secondcircuit}
    \label{theorem:co-circuit}
    There exists a regret minimization algorithm $\mathcal{R}_{\Psi_i}$ for the set $\co \Psi_{i}$ (\Cref{fig:algo}) such that under any sequence of utility vectors $\vec{L}_i^{(1)}, \dots, \vec{L}_i^{(T)}$ its regret $\reg_{\Psi_i}^T$ can be bounded as
    \begin{equation*}
    \reg_{\Psi_i}^T \leq \frac{ \log |\Sigma_i| + \|\mathcal{Q}_i\|^2_1 \max_{j \in \mathcal{J}_i} \log |\mathcal{A}_j|}{\eta} + \eta (\|\mathcal{Q}_i\|_1 + 4|\Sigma_{i}|^2) \sum_{t=1}^T \| \vec{L}_i^{(t)} - \vec{L}_i^{(t-1)} \|_{\infty}^2.
    \end{equation*}
\end{restatable}

As illustrated in \Cref{fig:algo}, the ``mixer'' $\mathcal{R}_{\Delta}$ is instantiated with OMWU, while each regret minimizer $\mathcal{R}_{\hat{\sigma}}$, for $\hat{\sigma} \in \hat{\sigma} \in \Sigma_i^*$, internally employs the dilatable global entropy as DGF to construct a regert minimizer over $\mathcal{Q}_j$. A notable ingredient of our predictive regret circuit (\Cref{proposition:circ-pred}) is that we employ an advanced prediction mechanism in place of the usual ``one-recency bias'' wherein the prediction is simply the previously observed utility. This leads to an improved regret bound as we further explain in \Cref{remark:better_prediction}.

\subsection{Stability of the Fixed Points}
\label{section:stability}

As suggested by \Cref{theorem:accelerating-Phi}, employing a predictive regret minimizer is of little gain if we cannot guarantee that the observed utilities will be stable. For this reason, in this subsection we focus on characterizing the stability of the fixed points, eventually leading to our stable-predictive $\co \Psi_{i}$-regret minimizer. In the context of \Cref{theorem:accelerating-Phi}, this establishes the \emph{stable} fixed point oracle. All of the omitted proofs of this section are included in \Cref{appendix:section_5}.

\paragraph{Multiplicative Stability.} Our analysis will reveal a particularly strong notion of stability we refer to as \emph{multiplicative stability}. More precisely, we say that a sequence $( \vec{z}^{(t)} )$, with $\vec{z}^{(t)} \in \mathbb{R}^d_{>0}$, is $\kappa$\emph{-multiplicative-stable}, with $\kappa \in (0,1)$, if $(1 + \kappa)^{-1} \vec{z}^{(t-1)}[k]  \leq \vec{z}^{(t)}[k] \leq (1 + \kappa) \vec{z}^{(t-1)}[k]$, for any $k \in [d]$ and for all $t \geq 2$. 
When $\vec{z}^{(t)}[k]$ and $\vec{z}^{(t-1)}[k]$ are such that $ (1 + \kappa)^{-1} \vec{z}^{(t-1)}[k] \leq \vec{z}^{(t)}[k] \leq (1 + \kappa) \vec{z}^{(t-1)}[k]$, we say that they are $\kappa$-\emph{multiplicative-close}. We begin by showing that OMWU on the simplex and OFTRL with dilatable global entropy as DGF guarantee multiplicative stability.

\begin{restatable}{lemma}{mulstabsimplex}
    \label{lemma:OMW-simplex-stability}
    Consider the OMWU algorithm on the simplex $\Delta^m$ with $\eta > 0$. If all the observed utilities and the predictions are such that $\|\vec{\ell}^{(t)}\|_{\infty}, \| \vec{m}^{(t)}\|_\infty \leq \|\vec{\ell}\|_\infty$, and $\eta < 1/(12\|\vec{\ell}\|_\infty)$, then the sequence $(\vec{x}^{(t)})$ produced by $\omw$ is $(12 \eta \|\vec{\ell}\|_\infty)$-multiplicative-stable.
\end{restatable}

\begin{restatable}{lemma}{mulstabseq}
    \label{lemma:OMD-stability}
    Consider the \eqref{eq:oftrl} algorithm on the sequence-form strategy polytope $\mathcal{Q}$ with dilatable global entropy as DGF and $\eta > 0$. If all the utility functions are such that $\|\vec{\ell}^{(t)}\|_{\infty} \leq 1$, and $\eta = O(1/\mathfrak{D})$ is sufficiently small, then the sequence $(\vec{x}^{(t)})$ produced is $O(\eta \mathfrak{D})$-multiplicative-stable.
\end{restatable}

To establish multiplicative stability of \eqref{eq:oftrl} under the dilatable global entropy DGF we first derive a closed-form solution which reveals the multiplicative structure of the update rule for the behavioral strategies at every ``local'' decision point. Then, the conversion to the sequence-form representation leads to a slight degradation of an $O(\mathfrak{D})$ (depth) factor in the multiplicative stability. Next, we use \Cref{lemma:OMW-simplex-stability,lemma:OMD-stability} to arrive at the following conclusion.

\begin{corollary}
    \label{corollary:mul-stability}
    Consider the regret minimization algorithm of \Cref{fig:algo}, and suppose that $\mathcal{R}_{\Delta}$ is instantiated using OMWU with $\eta > 0$, while each $\mathcal{R}_{\hat{\sigma}}$ is instantiated using \eqref{eq:oftrl} with dilatable global entropy as DGF and $\eta > 0$, for all $\hat{\sigma} \in \Sigma_i^*$. Then, for a sufficiently small $\eta = O(1/\|\mathcal{Q}_i\|_1)$,
    \begin{itemize}
        \item[(i)] The output sequence of each $\mathcal{R}_{\hat{\sigma}}$ is $O(\eta \mathfrak{D}_i)$-multiplicative-stable;
        \item[(ii)] The output sequence of $\mathcal{R}_{\Delta}$ is $O(\eta \|\mathcal{Q}_i\|_1)$-multiplicative-stable.
    \end{itemize}
\end{corollary}

Armed with this characterization, we will next establish the multiplicative stability of the fixed points associated with trigger deviation functions. To this end, building on the approach of \citet{Farina21:Simple}, let us introduce the following definitions.

\begin{definition}
    Consider a player $i \in [n]$ and let $J \subseteq \mathcal{J}_{i}$ be a subset of $i$'s information sets. We say than $J$ is a \emph{trunk} of $\mathcal{J}_{i}$ if, for every $j \in J$, all predecessors of $j$ are also in $J$. 
\end{definition}

\begin{definition} 
    \label{definition:partial_fixed_point}
    Consider a player $i \in [n]$, a trunk $J \subseteq \mathcal{J}_{i}$, and $\phi_i \in \co \Psi_{i}$. A vector $\vec{x}_i \in \R_{\geq 0}^{|\Sigma_{i}|}$ is a $J$-\emph{partial fixed point} of $\phi_i$ if the following conditions hold:
    \begin{itemize}
        \item $\vec{x}_i[\emptyseq] = 1$ and $\vec{x}_i[\sigma_j] = \sum_{a \in \mathcal{A}_j} \vec{x}_i[(j, a)]$, for all $j \in J$;
        \item $\phi_i(\vec{x}_i)[\emptyseq] = \vec{x}_i[\emptyseq] = 1$, and $\phi_i(\vec{x}_i)[(j, a)] = \vec{x}_i[(j, a)]$, for all $j \in J$ and $a \in \mathcal{A}_j$.
    \end{itemize}
\end{definition}

An important property is that a $J$-partial fixed point can be efficiently ``promoted'' to a $J \cup \{j^*\}$-partial fixed point by computing the stationary distribution of a certain Markov chain (see \Cref{algo:Extend}). However, it is a priori unclear how this fixed point operation would affect the stability of the produced strategies. In fact, even for a $2$-state Markov chain, the stationary distribution could behave very unsmoothly under slight perturbations in the transition probabilities; \emph{e.g.}, see~\citep{Meyer80:The,Haviv84:Perturbation,Chen20:Hedging}. This is where the stronger notion of multiplicative stability comes into play. Indeed, it turns out that as long as the transition probabilities are multiplicative-stable, the stationary distribution will also be stable~\citep{Candogan13:Dynamics}. This observation was also leveraged by \citet{Chen20:Hedging} to obtain an $O(T^{-3/4})$ rate of convergence to correlated equilibria in normal-form games. 

However, our setting is substantially more complex, and direct extensions of those prior techniques appears to only give a bound \emph{exponential} in the size of the game. In light of this, one of our key observations is that the associated Markov chains has a particular structure which enables us to establish a polynomial degradation in terms of stability. At a high level, we observe that the underlying Markov chain can be expressed as the convex combination of a stable chain with a much less stable rank-one component. The main concern is that the unstable rank-one chain could cause a substantial degradation in terms of the stability of the fixed points. We address this by proving the following key lemma.

\begin{restatable}{lemma}{rankone}
    \label{lemma:convex_characterization}
    Let $\mat{M}$ be the transition matrix of an $m$-state Markov chain such that $\mat{M} := \vec{v} \vec{1}^\top + \mat{C}$, where $\mat{C}$ is a matrix with strictly positive entries and columns summing to $1 - \lambda$, and $\vec{v}$ is a vector with strictly positive entries summing to $\lambda$. Then, if $\vec{\pi}$ is the stationary distribution of $\mat{M}$, there exists, for each $i \in [m]$, a (non-empty) finite set $F_i$ and $F = \bigcup_i F_i$, and corresponding parameters $b_j \in \{0, 1\}, 0 \leq p_j \leq m-2, |S_j| = m - p_j - b_j - 1$, for each $j \in F_i$, such that
    \begin{equation*}
        \vec{\pi}[i] = \frac{\sum_{j \in F_i} \lambda^{p_j + 1} (\vec{v}[q_j])^{b_j}  \prod_{(s,w) \in S_j} \mat{C}[(s,w)]}{\sum_{j \in F} C_j \lambda^{p_j + b_j} \prod_{(s,w) \in S_j}\mat{C}[(s,w)]},
    \end{equation*}
    where $C_j = C_j(m)$ is a positive parameter.
\end{restatable}

The main takeaway of this lemma is that the stationary distribution has only an \emph{affine dependence} on the vector $\vec{v}$. This will be crucial as $\vec{v}$ will be much less stable than the entries of $\mat{C}$, as we make precise in the sequel. Naturally, \Cref{lemma:convex_characterization} is not at all apparent from the Markov chain tree theorem, and derives from the particular structure of the Markov chain. Indeed, to establish \Cref{lemma:convex_characterization} we deviate from the existing techniques which are relying on the Markov chain tree theorem, and we instead leverage linear-algebraic techniques to characterize the corresponding eigenvector of the underlying Laplacian system. As a result, using a slight variant of \Cref{lemma:convex_characterization} (see \Cref{corollary:w_formula}) leads to the following stability bound.

\begin{restatable}{corollary}{markovstable}
    \label{corollary:Markov_stability}
Let $\mat{M}, \mat{M}'$ be the transition matrices of $m$-state Markov chains such that $\mat{M} = \vec{v} \vec{1}^{\top} + \mat{C}$ and $\mat{M}' = \vec{v}' \vec{1}^{\top} + \mat{C}'$, where $\mat{C}$ and $\mat{C}'$ are matrices with strictly positive entries, and $\vec{v}, \vec{v}'$ are vectors with strictly positive entries such that $\vec{v} = \vec{r}/l$ and $\vec{v}' = \vec{r}'/l'$, for some $l > 0$ and $l' > 0$. If $\vec{\pi}$ and $\vec{\pi}'$ are the stationary distributions of $\mat{M}$ and $\mat{M}'$, let $\vec{w} \defeq l \vec{\pi}$ and $\vec{w}' \defeq l' \vec{\pi}'$. Finally, let $\lambda$ and $\lambda'$ be the sum of the entries of $\vec{v}$ and $\vec{v}'$ respectively. Then, if (i) the matrices $\mat{C}$ and $\mat{C}'$ are $\kappa$-multiplicative-close; (ii) the scalars $\lambda$ and $\lambda'$ are $\kappa$-multiplicative-close; (iii) the vectors $\vec{r}$ and $\vec{r}'$ are $\gamma$-multiplicative-close; and (iv) the scalars $l$ and $l'$ are also $\gamma$-multiplicative-close, then the vectors $\vec{w}$ and $\vec{w}'$ are $(\gamma + O(\kappa m))$-multiplicative-close, for a sufficiently small $\kappa = O(1/m)$.
\end{restatable}

Under the assertion that $\gamma \gg \kappa$, the key takeaway is that the ``closeness'' of $\vec{w}$ and $\vec{w}'$ does not scale with $O((\gamma + \kappa) m)$, but only as $\gamma + O(\kappa m)$. Using this bound we are ready to characterize the degradation in stability after a ``promotion'' (\Cref{algo:Extend}) of a partial fixed point (in the formal sense of \Cref{definition:partial_fixed_point}).

\begin{restatable}{proposition}{stabextend}
    \label{proposition:stability_extend}
    Consider a player $i \in [n]$, and let $\phi^{(t)}_i = \sum_{\hat{\sigma} \in \Sigma_i^*} \vec{\lambda}_i^{(t)}[\hat{\sigma}] \phi_{\hat{\sigma} \rightarrow \vec{q}^{(t)}_{\hat{\sigma}}}$ be a transformation in $\co \Psi_{i}$ such that the sequences $(\vec{\lambda}_i^{(t)})$ and $(\vec{q}_{\hat{\sigma}}^{(t)})$ are $\kappa$-multiplicative-stable, for all $\hat{\sigma} \in \Sigma_i^*$. If $(\vec{x}_i^{(t)})$ is a $\gamma$-multiplicative-stable $J$-partial fixed point sequence, the sequence of $(J \cup \{j^*\})$-partial fixed points of $\phi_i$ is $(\gamma + O(\kappa |\mathcal{A}_{j^*}|))$-multiplicative-stable, for any sufficiently small $\kappa = O(1/|\mathcal{A}_{j^*}|)$.
\end{restatable}

Moreover, we employ this proposition as the inductive step to derive sharp multiplicative-stability bounds for the sequence of fixed points. The underlying algorithm gradually invokes the ``promotion'' subroutine (\Cref{algo:Extend}) in a top-down traversal of the tree, and it is given in \Cref{algo:FP-EFCE}.  

\begin{restatable}{theorem}{fpth}
    \label{theorem:stability-EFCE}
    Consider a player $i \in [n]$, and let $\phi_i^{(t)} = \sum_{\hat{\sigma} \in \Sigma_i^{*}} \vec{\lambda}_i^{(t)}[\hat{\sigma}] \phi_{\hat{\sigma} \rightarrow \vec{q}^{(t)}_{\hat{\sigma}}}$ be a transformation in $\co \Psi_{i}$ such that the sequences $(\vec{\lambda}_i^{(t)})$ and $(\vec{q}_{\hat{\sigma}}^{(t)})$ are $\kappa$-multiplicative-stable, for all $\hat{\sigma} \in \Sigma_i^*$. Then, the sequence of fixed points $\vec{q}_i^{(t)} \in \mathcal{Q}_{i}$ of $\phi_i^{(t)}$ is $O(\kappa |\mathcal{A}_i| \mathfrak{D}_{i})$-multiplicative-stable, for a sufficiently small $\kappa = O(1/(|\mathcal{A}_{i}| \mathfrak{D}_{i}))$, where $|\mathcal{A}_{i}| \defeq \max_{j \in \mathcal{J}_{i}} |\mathcal{A}_j|$.
\end{restatable}
A more refined bound is discussed in \Cref{remark:precise_bound}. The important insight of \Cref{theorem:stability-EFCE} is that the stability degrades according to the \emph{sum} of the actions at the decision points encountered along each path, and not as the \emph{product} of the actions. This is crucial as the latter bound---which would follow from prior techniques---need not be polynomial in the description of the game. At the heart of this improvement lies our refined characterization obtained in \Cref{lemma:convex_characterization}. Using the stability bounds derived in \Cref{corollary:mul-stability}, we are ready to establish the multiplicative-stability of the sequence of fixed points. 

\begin{corollary}[Stability of Fixed Points]
    \label{corollary:FP-stability}
For any sufficiently small $\eta = O(1/(\mathfrak{D}_i |\mathcal{A}_i| \| \mathcal{Q}_i \|_1 ))$, the sequence of fixed points $(\vec{q}_i^{(t)})$ of player $i \in [n]$ is $O(\eta \mathfrak{D}_i |\mathcal{A}_i| \| \mathcal{Q}_i \|_1 )$-multiplicative-stable.
\end{corollary}

\subsection{Completing the Proof}
\label{subsection:completing}

Finally, we combine all of the previous pieces to complete the construction. First, we apply \Cref{theorem:accelerating-Phi} using the predictive bound obtained in \Cref{theorem:co-circuit} to conclude that the $\Psi_i$-regret of each player $i \in [n]$ can be bounded as 

\begin{equation}
    \label{eq:psi-reg}
    \reg_i^T \leq \frac{\log |\Sigma_i| + \|\mathcal{Q}_i\|^2_1 \log |\mathcal{A}_i|}{\eta} + 10\eta |\Sigma_i|^2 \sum_{t=1}^T \|\vec{\ell}_i^{(t)} - \vec{\ell}_i^{(t-1)}\|_\infty^2 + 10\eta |\Sigma_i|^2 \sum_{t=1}^T \|\vec{q}_i^{(t)} - \vec{q}_i^{(t-1)}\|_\infty^2,
\end{equation}
where we assumed---for simplicity---exact computation of each fixed point, \emph{i.e.}, $\epsilon^{(t)} = 0$ for any $t \geq 1$, while we also used the fact that $\|\vec{\ell}_i^{(t)}\|_\infty \leq 1$ which follows from the normalization assumption. So far we have focused on bounding the regret of each player without making any assumptions about the stability of the observed utility functions. A crucial observation is that if all players employ a regularized (or smooth) learning algorithm, then the observed utility functions will also change slowly over time. This is formalized in the following auxiliary claim.

\begin{restatable}{claim}{utilinfty}
    \label{claim:aux}
    For any player $i \in [n]$ the observed utilities satisfy
    \begin{equation*}
        \| \vec{\ell}_i^{(t)} - \vec{\ell}_i^{(t-1)} \|^2_{\infty} \leq (n-1) |\mathcal{Z}|^2 \sum_{k \neq i} \| \vec{q}_k^{(t)} - \vec{q}_k^{(t-1)} \|^2_1.
    \end{equation*}
\end{restatable}
Thus, plugging this bound to \eqref{eq:psi-reg} yields that the $\Psi_i$-regret $\reg_i^T$ of each player $i$ can be bounded as
\begin{equation*}
    \frac{\log |\Sigma_i| + \|\mathcal{Q}_i\|^2_1 \log |\mathcal{A}_i|}{\eta} + 10 \eta (n-1) |\Sigma_i|^2 |\mathcal{Z}|^2 \sum_{t=1}^T \sum_{k \neq i} \|\vec{q}_k^{(t)} - \vec{q}_k^{(t-1)}\|_1^2 + 10\eta |\Sigma_i|^2 \sum_{t=1}^T \|\vec{q}_i^{(t)} - \vec{q}_i^{(t-1)}\|_\infty^2.
\end{equation*}
As a result, using \Cref{corollary:FP-stability} we arrive at the following conclusion.

\begin{corollary}
    \label{corollary:1/4}
    Suppose that each player follows the dynamics of \Cref{fig:algo} with a sufficiently small learning rate $\eta = O(1/(T^{1/4} \mathfrak{D}_i |\mathcal{A}_i| \|\mathcal{Q}_i\|_1))$. Then, the $\Psi_i$-regret of each player will be bounded as $\reg_i^T \leq \mathcal{P} T^{1/4}$, where $\mathcal{P}$ is independent on $T$ and polynomial on the description of the game.
\end{corollary}

Finally, \Cref{theorem:main} follows from \Cref{theorem:EFCE-convergence} after performing sampling in order to transition to deterministic strategies, as we explain in \Cref{appendix:everything}. We also point out that the complexity of every iteration in the proposed dynamics is analogous to that in~\citep{Farina21:Simple}, although the dynamics developed in the latter paper only attain a rate of convergence of $O(T^{-1/2})$. Finally, we remark that it is easy to make the overall regret minimizer robust against adversarial losses using an adaptive choice of learning rate.

\NewDocumentCommand{\pure}{O{i}}{%
\ifthenelse{\equal{#1}{}}%
    {\vec{\pi}}%
    {\vec{\pi}^{(#1)}}%
}
\NewDocumentCommand{\hatpure}{O{i}}{%
\ifthenelse{\equal{#1}{}}%
    {\hat{\vec{\pi}}}%
    {\hat{\vec{\pi}}^{(#1)}}%
}
\newcommand{\Z}{\mathscr{Z}}
\newcommand{\bbR}{\mathbb{R}}
\NewDocumentCommand{\puret}{O{i}O{t}}{\vec{\pi}^{(#1),\,#2}}
\NewDocumentCommand{\Pure}{O{i}}{\Pi^{(#1)}}
\NewDocumentCommand{\tdev}{O{\hatinfo}O{\hatpure[]}}{%
    \phi_{#1\to#2}
}
\NewDocumentCommand{\Mdev}{O{\hatinfo}O{\hatpure[]}}{%
    \vec{M}_{#1\to#2}
}
\NewDocumentCommand{\cJ}{O{i}}{\mathscr{J}^{(#1)}}
\newcommand{\info}{j}
\newcommand{\emptyinfo}{\info_{\emptyseq}}
\newcommand{\emptyaction}{a_{\emptyseq}}
\newcommand{\lambdav}{\vec{\lambda}}
\newcommand{\hatinfo}{j}
\NewDocumentCommand{\Info}{O{i}}{\mathcal{J}_{#1}}
\newcommand{\cS}{\mathscr{S}}
\NewDocumentCommand{\Seqs}{O{i}}{\Sigma_{#1}}
\NewDocumentCommand{\pc}{O{z}}{p^{(c)}(#1)}
\NewDocumentCommand{\ut}{O{i}}{u^{(#1)}}
\NewDocumentCommand{\bbone}{O{XXX}}{\mathds{1}_{\{#1\}}}
\NewDocumentCommand{\Ph}{O{i}}{\Psi_{#1}}
\NewDocumentCommand{\Tph}{O{i}}{(\co\Ph)}
\newcommand{\A}{\mathscr{A}}
\newcommand{\alp}{\alpha_z^{(i),\,t}}

\NewDocumentCommand{\Q}{O{i}}{\mathcal{Q}}
\NewDocumentCommand{\q}{}{\vec{q}}
\NewDocumentCommand{\qt}{O{i}O{t}}{\vec{q}^{(#1),\,#2}}

\NewDocumentCommand{\parseq}{O{i}O{\info}}{\sigma^{(#1)}(#2)}

\NewDocumentCommand{\rdev}{O{i}O{\hatinfo}O{\hat{\vec{\pi}}}}{%
    r^{(#1)}_{\muv,#2\to#3}
}
\NewDocumentCommand{\cumr}{O{i}O{T}}{R^{(#1),\,#2}}

\section{Faster Convergence to EFCCE}
\label{section:efcce}

In this section we turn our attention to learning dynamics for extensive-form \emph{coarse} correlated equilibrium (EFCCE). While the dynamics we previously developed for EFCE would also trivially converge to EFCCE, as the former is a subset of the latter~\citep{Farina20:Coarse}, our main contribution is to show that each iteration of EFCCE dynamics can be substantially more efficient compared to EFCE. Indeed, unlike all known methods for EFCE, we obtain in \Cref{sec:closed form} a succinct closed-form solution for the fixed points associated with EFCCE which does not require the expensive computation of the stationary distribution of a Markov chain. This places EFCCE closer to normal-form coarse correlated equilibria (NFCCE) in terms of the per-iteration complexity, even thought EFCCE prescribes a much more compelling notion of correlation. Furthermore, we use this closed-form characterization in \Cref{sec:stab} to obtain improved stability bounds for the fixed points associated with EFCCE, and with a much simpler analysis compared to the one for EFCE. 

\subsection{Closed-Form Fixed Point 
Computation}\label{sec:closed form}

As suggested by our general template introduced in \Cref{theorem:accelerating-Phi}, we first have to construct a predictive regret minimizer for the set of \emph{coarse} trigger deviation functions $\widetilde{\Psi}_i$ (\Cref{definition:coarse_trigger_deviation_functions}). This construction is very similar to the one for $\Psi_i$ we previously described in detail in \Cref{section:Phi-regret}. For this reason, here we focus on the computation and the stability properties of the fixed points associated with any $\phi_i \in \co \widetilde{\Psi}_i$. Specifically, we will first show that it it possible to compute a sequence-form strategy $\vec{q}_i$ such that $\phi_i(\vec{q}_i) = \vec{q}_i$ in linear time on $O(|\Sigma_i| \mathfrak{D}_i)$.

Indeed, let $\phi_i =\sum_{\hatinfo\in\Info}\lambdav_i[\hatinfo]\tdev[\hatinfo][\q_{\hatinfo}]$ be any coarse trigger deviation function, where $\lambdav_i\in\Delta({\Info})$, and $\q_{\hatinfo}\in\Q_{\hatinfo}$ for each $\hatinfo\in\Info$. \Cref{algo:FP-EFCCE} describes an efficient procedure to compute a fixed point of a given transformation $\phi_i \in \co \widetilde{\Psi}_i$. In particular, the algorithm iterates over the sequences of player $i$ according to their partial ordering $\prec$. That is, it is never the case that a sequence $\sigma=(j,a)$ is considered before $\sigma_j$. For every sequence $\sigma=(j,a) \in \Sigma_i^* $ the algorithm computes $d_{\sigma}\in\bbR_{\ge0}$ as the sum of the weights corresponding to information sets preceding $j$ (\Cref{line:dsigma}). If $d_{\sigma}=0$, the choice we make at $\sigma$ is indifferent as long as the resulting vector $\q_i$ is a well-formed sequence-form strategy. For this reason, we simply set $\q_i[\sigma]$ so that the probability-mass flow is evenly divided among sequences originating in $j$ (\Cref{line: set to uniform}).
Otherwise, when $d_{\sigma}\ne 0$, \Cref{eq: fixed point} assigns to $\q_i[\sigma]$ a value equal to the weighted sum of $\q_{j'}[\sigma]\q_i[\sigma']$ for sequences $\sigma'=(j',a')$ preceding information set $j \in \mathcal{J}_i$. In the next theorem we show that \Cref{algo:FP-EFCCE} is indeed correct, and runs in time $O(|\Sigma_i| \mathfrak{D}_i)$.

\begin{restatable}{theorem}{closedForm}
\label{thm:closedForm}
For any player $i\in[n]$ and any transformation $\phi_i = \sum_{\hatinfo\in\Info}\lambdav_i[\hatinfo]\tdev[\hatinfo][\q_{\hatinfo}]\in\co \widetilde{\Psi}_i$, the output $\q_i \in \bbR^{|\Seqs|}$ of \Cref{algo:FP-EFCCE} is such that $\q_i \in \mathcal{Q}_i$ and $\phi_i(\q_i)= \q_i$. Furthermore, \Cref{algo:FP-EFCCE} runs in $O(|\Sigma_i|\mathfrak{D}_i)$.
\end{restatable}


\begin{algorithm}[ht]
\SetAlgoLined
\DontPrintSemicolon
\KwIn{$\phi_i = \sum_{j \in \mathcal{J}_{i}} \vec{\lambda}_i[j] \phi_{j \rightarrow \vec{q}_j} \in \co \widetilde{\Psi}_{i}$}    
\KwOut{$\vec{q}_i \in \mathcal{Q}_i$ such that $\phi_i(\vec{q}_i) = \vec{q}_i$}
$\vec{q}_i \leftarrow \mathbf{0} \in \R^{|\Sigma_{i}|}$, $\vec{q}_i[\emptyseq] \leftarrow 1$\label{line:q star init} \\
\For{$\sigma = (j, a) \in \Sigma^*_{i}$ in top-down ($\prec$) order}{
$d_{\sigma} \leftarrow \sum_{j' \preceq j} \vec{\lambda}_i[j']$ \label{line:dsigma} \\
\If{$d_{\sigma} = 0$}{$\vec{q}_i[\sigma] \leftarrow \frac{\vec{q}_i[\sigma_j]}{|\mathcal{A}_j|}$} \label{line: set to uniform}
\Else{$\vec{q}_i[\sigma] \leftarrow \frac{1}{d_{\sigma}} \sum_{j' \preceq j} \vec{\lambda}_i[j'] \vec{q}_{j'}[\sigma] \vec{q}_i[\sigma_{j'}]$ \label{eq: fixed point} \\
}
\textbf{return} $\vec{q}_i$
}
\caption{$\textsc{FixedPoint}(\phi_i)$ for $\phi_i \in \co \widetilde{\Psi}_{i}$}
\label{algo:FP-EFCCE}
\end{algorithm}

\subsection{Stability of the Fixed Points}
\label{sec:stab}

Another important application of our closed-form solution in \Cref{algo:FP-EFCCE} is that it allows us to obtain through a simple analysis sharp bounds on the stability of the fixed points. Indeed, we show that the fixed point operation only leads to (multiplicative) degradation linear in the depth of each player's subtree.

\begin{restatable}{proposition}{stabilityefcce}
    \label{proposition:FP-EFCCE}
    Suppose that the sequences $(\vec{\lambda}_i^{(t)})$ and $(\vec{q}_j^{(t)})$, for all $j \in \mathcal{J}_i$, are $\kappa$-multiplicative-stable. Then, \Cref{algo:FP-EFCCE} yields a sequence of $(12 \kappa \mathfrak{D}_{i})$-multiplicative-stable strategies, assuming that $\kappa < 1/(12 \mathfrak{D}_{i})$.
\end{restatable}

Observe that the derived bound on stability is slightly better compared to that for $\efce$ (\Cref{theorem:stability-EFCE}). Consequently, having established the stability of the fixed points, we can apply \Cref{theorem:accelerating-Phi} to derive a stable-predictive $\widetilde{\Psi}_{i}$-regret minimizer for each player $i \in [n]$. This leads to a result analogous to \Cref{corollary:1/4} we showed for EFCE, but our dynamics for EFCCE have a substantially improved per-iteration complexity.


\section{Experiments}

In this section we support our theoretical findings through experiments conducted on benchmark general-sum games. Namely, we experiment with the following popular games: (i) a three-player variant of \emph{Kuhn poker}~\citep{Kuhn50:Simplified}; (ii) a two-player bargaining game known as \emph{Sheriff}~\citep{Farina19:Correlation}; (iii) a three-player version of \emph{Liar's dice}~\citep{Lisy15:Online}; and (iv) three-player Goofspiel~\citep{Ross71:Goofspiel}. A detailed description of each of these games and the precise parameters we use is given in \cref{app:games}. The rest of this section is organized as follows. \Cref{sec:exp-EFCE} shows the convergence of our accelerated dynamics for EFCE (as presented in \Cref{sec:efce}) compared to the prior state of the art. Next, \Cref{sec:exp-EFCCE} illustrates the convergence of our dynamics for EFCCE.

\subsection{Convergence to EFCE}
\label{sec:exp-EFCE}

Here we investigate the performance of our accelerated dynamics for EFCE (\Cref{fig:algo}) compared to the existing algorithm by~\citet{Farina21:Simple}. Both of these dynamics will be based on a CFR-style decomposition into ``local'' regret-minimization problems. In particular, our stable-predictive dynamics will use OMWU at every local decision point (as in \Cref{proposition:cfr}), while the algorithm of~\citet{Farina21:Simple} will be instantiated with (i) \emph{regret matching$^+$} ($\rmp$)~\citep{Tammelin14:Solving} for each simplex (in place of regret matching), and (ii) using the vanilla MWU algorithm for each simplex. In accordance to our theoretical predictions (\Cref{corollary:1/4}), the stepsize for $\omw$ is set as $\eta^{(t)} = \tau \cdot t^{-1/4}$, while for MWU it is set as $\eta^{(t)} = \tau \cdot t^{-1/2}$. Here $\tau$ is treated as a hyperparameter, chosen by picking the best-performing value among $\{0.01, 0.1, 1, 10, 100 \}$. 



\begin{figure}[thp]
    \centering
    \def\scale{0.43}
    \includegraphics[scale=\scale]{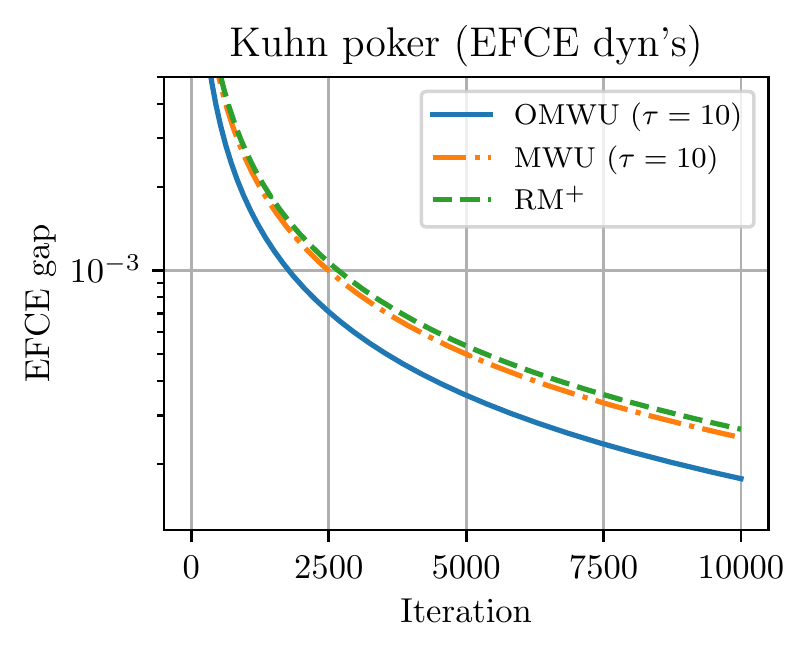}
    \includegraphics[scale=\scale]{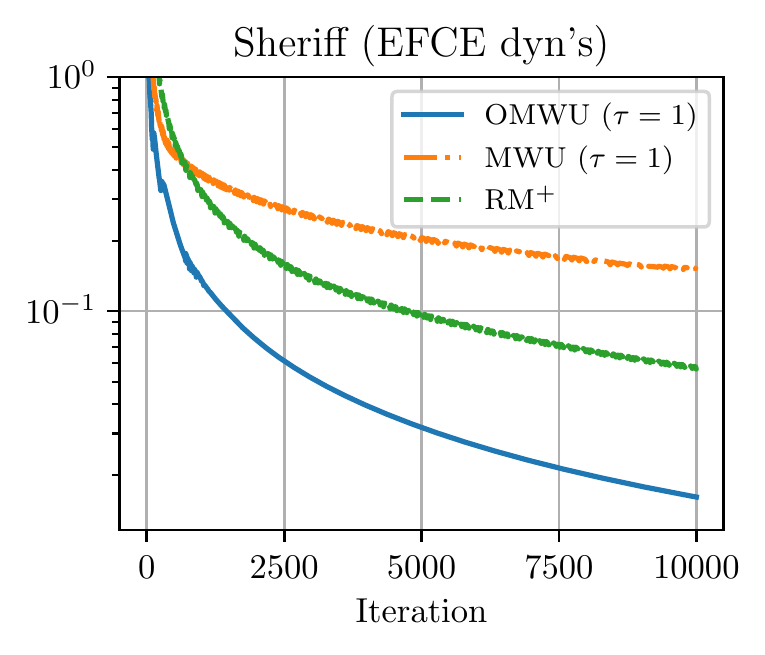}
    \includegraphics[scale=\scale]{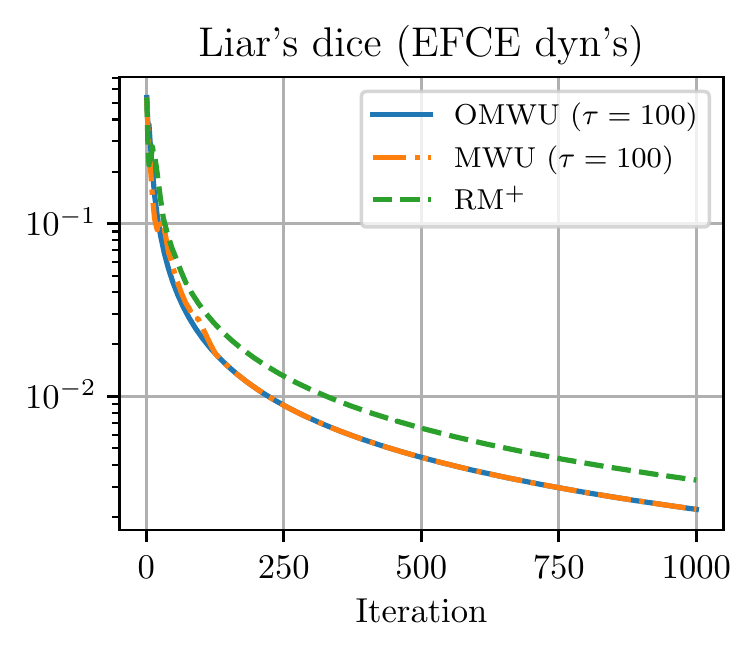}
    \includegraphics[scale=\scale]{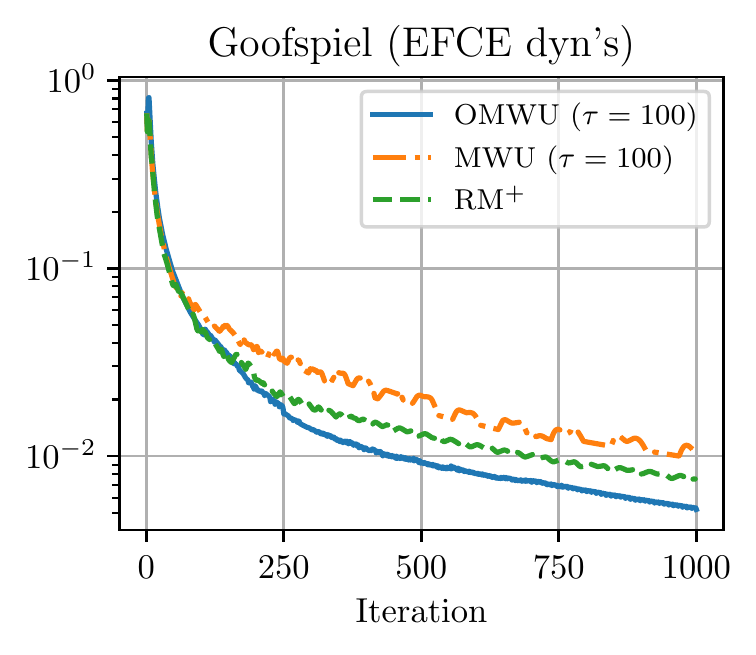}
    \caption{The performance of EFCE dynamics based on MWU, OMWU, and RM$^+$ on four general-sum EFGs.}
    \label{fig:efce_convergence_experiments}
\end{figure}

\Cref{fig:efce_convergence_experiments} shows the rate of convergence for each of the three learning dynamics we described. On the $x$-axis we indicate the number of iterations performed by each algorithm and on the $y$-axis we plot the \emph{EFCE gap}, defined as the maximum advantage that any player can gain by defecting optimally from the mediator's recommendations. It should be noted that one iteration costs the same for every algorithm, up to constant factors. We see that on every game, $\omw$ performs better than or on par with $\rmp$ and MWU. On Sheriff, a benchmark introduced specifically for the study of correlated equilibria, $\omw$ outperforms both $\rmp$ and MWU by about an order of magnitude. 

In the context of two-player zero-sum games, CFR with $\rmp$ is a formidable algorithm, typically outperforming theoretically superior dynamics. With that in mind, it is quite interesting that for EFCE computation we are able to achieve better performance using OMWU with only a modest amount of stepsize tuning. We hypothesize that this is due to the inherent differences between solving a zero-sum game via Nash equilibrium versus the problem of computing correlated equilibria.
%
One caveat to these results is that we did not use two tricks that help $\cfr$ in two-player zero-sum EFG solving: \emph{alternation} and \emph{linear averaging}. These tricks are known to retain convergence guarantees in that context~\citep{Tammelin15:Solving,Farina19:Online,Burch19:Revisiting}, but it is unclear if they still guarantee convergence in the $\efce$ setting. 

\subsection{EFCCE}
\label{sec:exp-EFCCE}

Next, we investigate the convergence of our learning dynamics for EFCCE, obtained within the same framework we developed for EFCE. We first measure the rate of convergence in an analogous to the previous subsection setup. The results are illustrated in \Cref{fig:efcce_convergence_experiments}.

%
%

\begin{figure}[thp]
    \centering
    \def\scale{0.43}
    \includegraphics[scale=\scale]{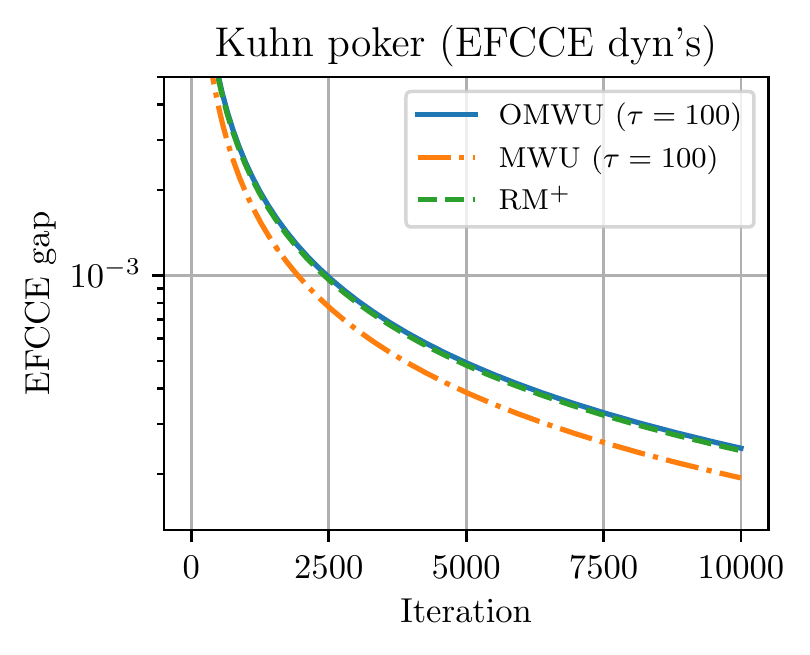}
    \includegraphics[scale=\scale]{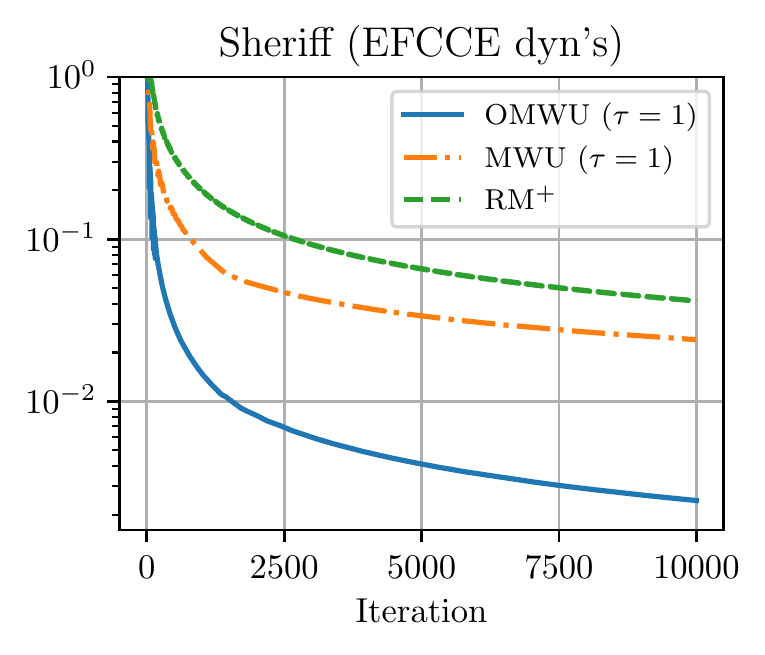}
    \includegraphics[scale=\scale]{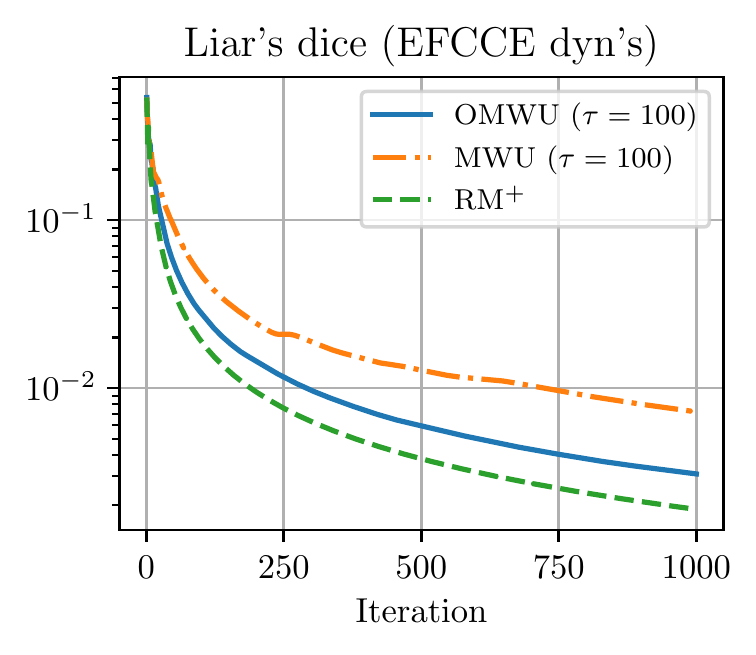}
    \includegraphics[scale=\scale]{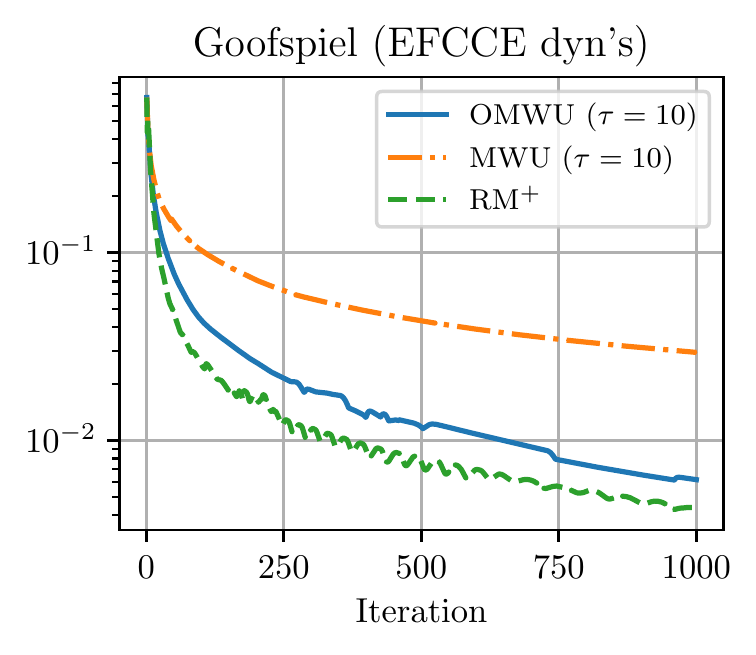}
    \caption{The performance of EFCCE dynamics based on MWU, OMWU, and RM$^+$ on four general-sum EFGs.}
    \label{fig:efcce_convergence_experiments}
\end{figure}

Interestingly, we observe a noticeable qualitative difference for convergence to EFCCE. Indeed, unlike EFCE (\Cref{fig:efce_convergence_experiments}), $\rmp$ outperforms OMWU in both Liar's dice and Goofspiel. It is also surprising that MWU converges faster than its optimistic counterpart in Kuhn poker. These results suggest a substantial difference in the convergence properties between EFCE and EFCCE. Furthermore, we illustrate in \Cref{fig:efce_vs_efcce} the running time complexity of EFCE versus EFCCE dynamics (both instantiated with $\rmp$), measured in terms of the EFCCE gap.

\begin{figure}[thp]
    \centering
    \def\scale{0.43}
    \includegraphics[scale=\scale]{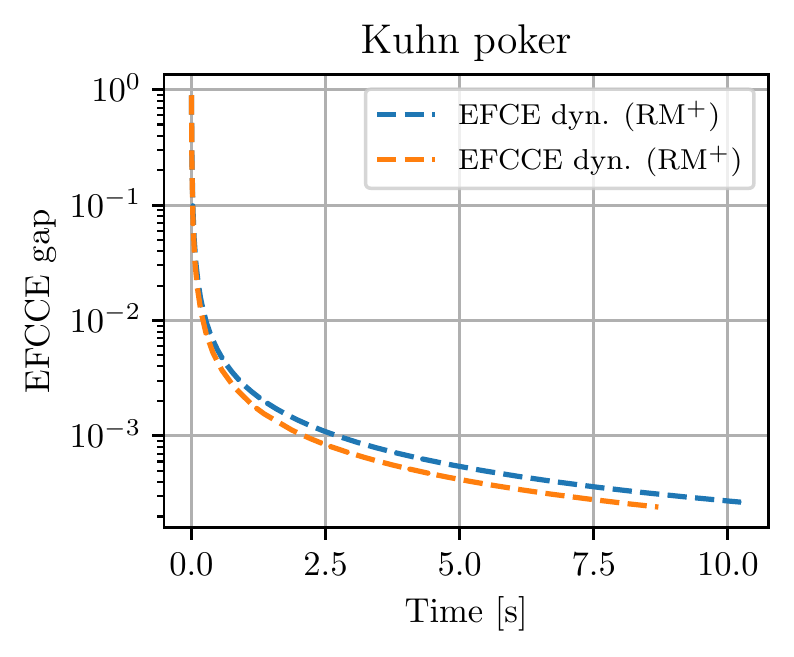}
    \includegraphics[scale=\scale]{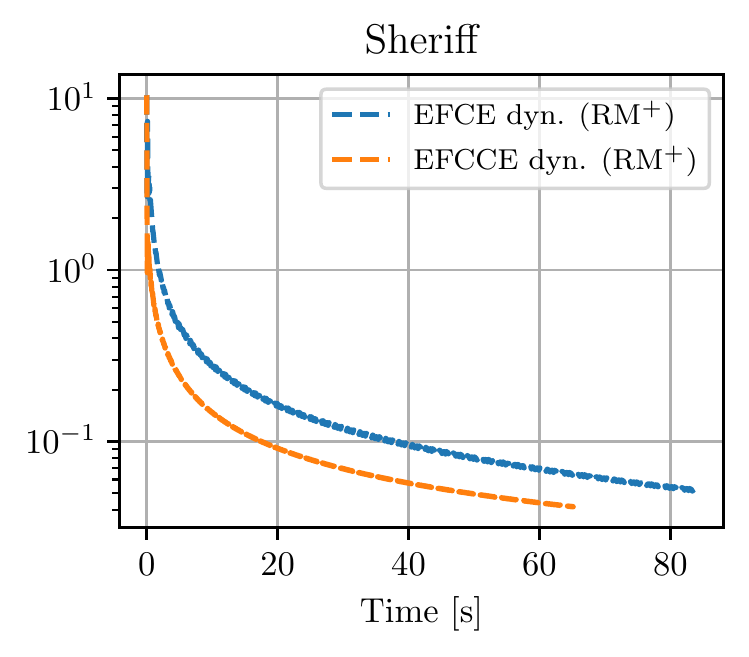}
    \includegraphics[scale=\scale]{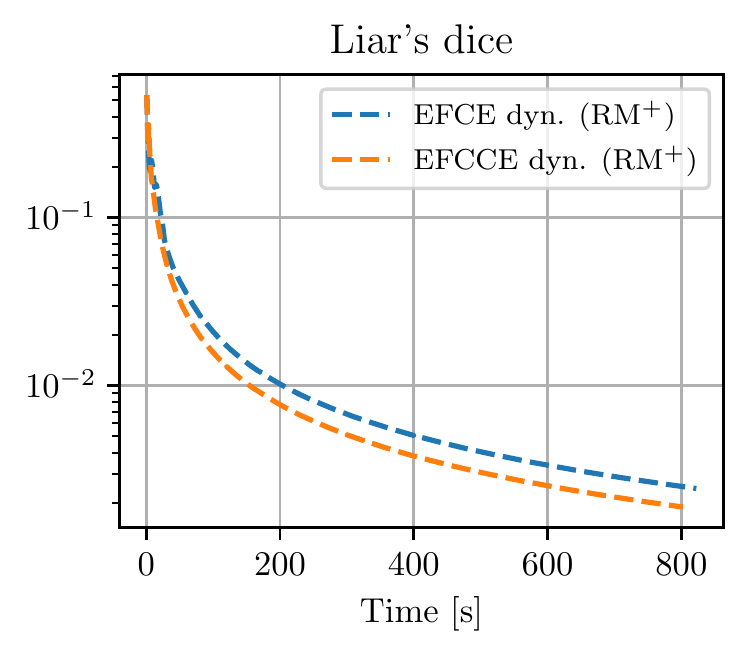}
    \includegraphics[scale=\scale]{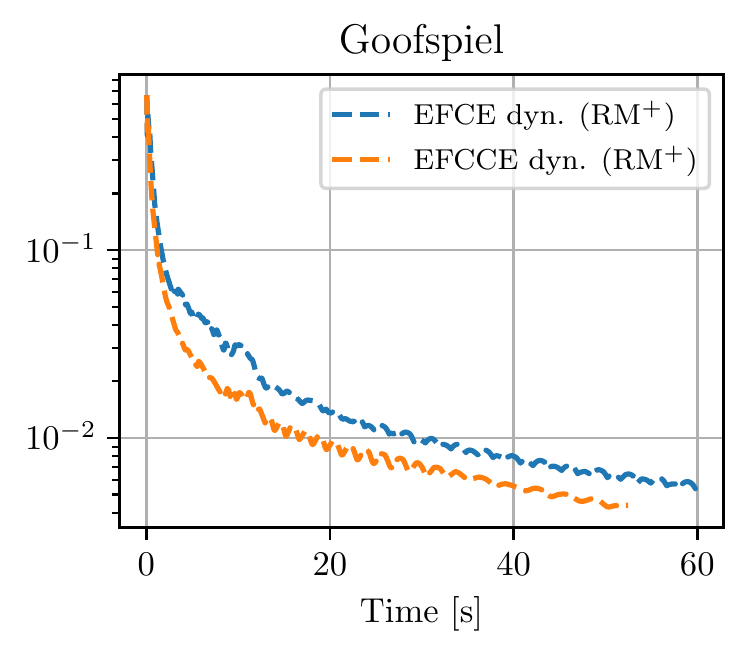}
    \caption{The convergence of EFCE and EFCCE dynamics to EFCCE, measured through the EFCCE gap.}
    \label{fig:efce_vs_efcce}
\end{figure}

In each game, the fixed point computation for the EFCE dynamics was performed through an optimized implementation of the power iteration method, interrupted when the Euclidean norm of the residual was below the value of $10^{-6}$. On the other hand, the fixed points for EFCCE were computed using our closed-form solution (\Cref{algo:FP-EFCCE}). In all four games, we see that our EFCCE dynamics outperform the EFCE dynamics in terms of the running time complexity, often by a significant margin. This is consistent with our intuition since EFCE dynamics are solving a strictly harder problem---minimizing the EFCE gap, instead of the EFCCE gap.  




\section{Conclusions}

In this paper we developed uncoupled no-regret learning dynamics so that if all agents play $T$ repetitions of the game according to our dynamics, the correlated distribution of play is an $O(T^{-3/4})$-approximate extensive-form correlated equilibrium. This substantially improves over the prior best rate of $O(T^{-1/2})$. One of our main technical contributions was to characterize the stability of the fixed points associated with trigger deviation functions through a refined perturbation analysis of a structured Markov chain, which may be of independent interest. On the other hand, for fixed points associated with extensive-form \emph{coarse} correlated equilibria we established a closed-form solution, circumventing the computation of the stationary distribution of any Markov chain. Finally, experiments conducted on standard benchmarks corroborated our theoretical findings. 

Following recent progress in normal-form games \citep{Daskalakis21:near,Anagnostides21:Near}, an important question for the future is to obtain a further acceleration of the order $\widetilde{O}(T^{-1})$. As we pointed out in \Cref{sec:related}, this would inevitably require new techniques since the known methods do not apply for the substantially more complex problem of extensive-form correlated equilibria. We believe that our characterization of the fixed points associated with trigger deviation functions could be an important step towards achieving this goal.

\begin{acks}
Tuomas Sandholm is supported by the National Science Foundation under grants IIS-1901403 and CCF-1733556.

\end{acks}

\bibliographystyle{ACM-Reference-Format}
\bibliography{paper}

\clearpage

\appendix

\section{Omitted Proofs}
\label{appendix:proofs}

This section includes all of the proofs we omitted from the main body. Let us first introduce some additional useful notation.

\subsection{Further Notation}

It will be convenient to instantiate a trigger deviation function (recall \Cref{definition:trigger_deviation_functions}) in the form of a linear mapping $\phi_{\hat{\sigma} \rightarrow \hat{\vec{\pi}}_i} : \R^{|\Sigma_{i}|} \ni \vec{x} \mapsto \mat{M}_{\hat{\sigma} \rightarrow \hat{\vec{\pi}}_i} \vec{x}$, where $\mat{M}_{\hat{\sigma} \rightarrow \hat{\vec{\pi}}_i}$ is such that for any $\sigma_r, \sigma_c \in \Sigma_{i}$,
\begin{equation}
\label{eq:canonical_trigger_deviation_functions}
    \mat{M}_{\hat{\sigma} \rightarrow \hat{\vec{\pi}}_i}[\sigma_r, \sigma_c] = 
    \begin{cases}
    1 & \textrm{if  } \sigma_c \not\succeq \hat{\sigma} \quad\&\quad \sigma_r = \sigma_c; \\
    \hat{\vec{\pi}}_i[\sigma_r] & \textrm{if  } \sigma_c = \hat{\sigma} \quad\&\quad \sigma_r \succeq j; \\
    0 & \textrm{otherwise},
    \end{cases}
\end{equation}
where $\hat{\sigma} = (j, a) \in \Sigma^*_i$. It is not hard to show that the linear mapping described in \eqref{eq:canonical_trigger_deviation_functions} is indeed a trigger deviation function in the sense of \Cref{definition:trigger_deviation_functions}. Similarly, we express a coarse trigger deviation function in the form of a linear mapping $\phi_{j \rightarrow \hat{\vec{\pi}}_i} : \R^{|\Sigma_{i}|} \ni \vec{x} \mapsto \mat{M}_{j \rightarrow \hat{\vec{\pi}}_i} \vec{x}$, where $\mat{M}_{j \rightarrow \hat{\vec{\pi}}_i}$ is such that for any $\sigma_r, \sigma_c \in \Sigma_{i}$,
\begin{equation*}
    \mat{M}_{j \rightarrow \hat{\vec{\pi}}_i}[\sigma_r, \sigma_c] = 
    \begin{cases}
    1 & \textrm{if  } \sigma_c \not\succeq j \quad\&\quad \sigma_r = \sigma_c; \\
    \hat{\vec{\pi}}_i[\sigma_r] & \textrm{if  } \sigma_c = \sigma_j \quad\&\quad \sigma_r \succeq j; \\
    0 & \textrm{otherwise}.
    \end{cases}
\end{equation*}

Furthermore, we will use the notation $\vec{x} \otimes \vec{y} = \vec{x} \vec{y}^{\top}$ to denote the \emph{outer product} of (compatible) vectors $\vec{x}$ and $\vec{y}$, while we will also write $(\mat{M})^{\flat}$ to represent the standard \emph{vectorization} of matrix $\mat{M}$.

\subsection{Proofs from \texorpdfstring{\Cref{section:prel}}{Section 2}}
\label{appendix:proofs_prel}

\efccereg*

\begin{proof}
    By assumption, we know that for any $i\in [n]$ it holds that $\reg_i^T \le \epsilon T$. Thus, by definition of $\reg_i^T$, it follows that for any $i \in [n]$ and any coarse trigger deviation function $\phi_i \in \widetilde{\Psi}_i$,
    \begin{align*}
    T\epsilon &\ge \sum_{t=1}^T\mleft(\ell_i^{(t)}\mleft(\phi_i(\vec{\pi}_i^{(t)})\mright)-\ell_i^{(t)}\mleft( \vec{\pi}_i^{(t)} \mright)\mright)
    =\sum_{t=1}^T\mleft(u_i \mleft(\phi_i( \vec{\pi}_i^{(t)} ),\vec{\pi}_{-i}^{(t)} \mright) - u_i \mleft(\vec{\pi}^{(t)} \mright)\mright) \\
    &= \sum_{t=1}^T\sum_{\pure[]\in\Pi} \mathbbm{1} \left\{ \vec{\pi} = \vec{\pi}^{(t)} \right\} \mleft(u_i \mleft(\phi_i(\vec{\pi}_i),\vec{\pi}_{-i} \mright) - u_i \mleft(\vec{\pi} \mright)\mright) \\
    &= \sum_{\pure[]\in\Pi} \sum_{t=1}^T \mleft( \mathbbm{1} \left\{ \vec{\pi} = \vec{\pi}^{(t)} \right\} \mright) \mleft(u_i \mleft(\phi_i(\vec{\pi}_i),\vec{\pi}_{-i} \mright) - u_i \mleft(\vec{\pi} \mright)\mright) \\
    &= T \sum_{\pure[]\in\Pi} \vec{\mu}[\vec{\pi}] \mleft(u_i \mleft(\phi_i(\vec{\pi}_i),\vec{\pi}_{-i} \mright) - u_i \mleft(\vec{\pi} \mright)\mright).
    \end{align*}
    This is precisely the definition of an $\epsilon$-EFCCE (\Cref{definition:efcce}), as we wanted to show.
\end{proof}

\subsection{Proofs from \texorpdfstring{\Cref{section:accelerating-Phi}}{Section 3.1}}
\label{appendix:proof-accelerating_phi}

Here we prove \Cref{theorem:accelerating-Phi}. For the convenience of the reader the theorem is restated below.

\phiaccel*

\begin{proof}
Fix any iteration $t \geq 2$. The first step is to obtain the next strategy of $\mathcal{R}_{\Phi}$: $\phi^{(t)} = \mathcal{R}_{\Phi}.\nextstr()$. Then, our regret minimizer $\mathcal{R}$ will simply output the strategy $\vec{x}^{(t)}$ such that $\vec{x}^{(t)} = \fporacle(\phi^{(t)}; \vec{x}^{(t-1)}, \kappa, \epsilon^{(t)})$.\footnote{For $t = 1$ it suffices to return any $\vec{x}^{(1)}$ such that $\vec{x}^{(1)} = \phi^{(1)}(\vec{x}^{(1)})$.} By assumption (recall \Cref{definition:FP-smooth}), we know that this is indeed well-defined and $\vec{x}^{(t)}$ will be such that (i) $\|\phi^{(t)}(\vec{x}^{(t)}) - \vec{x}^{(t)} \|_1 \leq \epsilon^{(t)}$, and (ii) $\|\vec{x}^{(t)} - \vec{x}^{(t-1)}\|_1 \leq \kappa$. This immediately implies that $\mathcal{R}$ will be $\kappa$-stable.

Afterwards, we receive feedback from the environment in the form of a utility vector $\vec{\ell}^{(t)}$, which in turn is used to construct the utility function $L^{(t)}: \phi \mapsto \langle \vec{\ell}^{(t)} , \phi(\vec{x}^{(t)}) \rangle$. Since $\Phi$ is a set of linear transformations, we can represent the corresponding utility vector as $\vec{L}^{(t)} = (\vec{\ell}^{(t)} \otimes \vec{x}^{(t)})^\flat$. This function is then given as feedback to $\mathcal{R}_{\Phi}$; that is, we invoke the subroutine $\mathcal{R}_{\Phi}.\obsut(L^{(t)})$. As a result, the (external) regret of $\mathcal{R}_{\Phi}$ can be expressed as
\begin{equation*}
    \reg_{\Phi}^T = \max_{\phi^* \in \Phi} \sum_{t=1}^T \langle \vec{\ell}^{(t)}, \phi^*(\vec{x}^{(t)}) \rangle - \sum_{t=1}^T \langle \vec{\ell}^{(t)}, \phi^t(\vec{x}^{(t)}) \rangle.
\end{equation*}
Furthermore, if $\reg^T$ is the $\Phi$-regret of $\mathcal{R}$, we have that
\begin{align}
    \reg^T - \reg_{\Phi}^T &= \sum_{t=1}^T \langle \vec{\ell}^{(t)}, \phi^{(t)}(\vec{x}^{(t)}) \rangle - \sum_{t=1}^T \langle \vec{\ell}^{(t)}, \vec{x}^{(t)} \rangle = \sum_{t=1}^T \langle \vec{\ell}^{(t)}, \phi^{(t)}(\vec{x}^{(t)}) - \vec{x}^{(t)} \rangle \notag \\
    &\leq \sum_{t=1}^T \|\vec{\ell}^{(t)}\|_* \| \phi^{(t)}(\vec{x}^{(t)}) - \vec{x}^{(t)}\| \leq \|\vec{\ell}\|_\infty \sum_{t=1}^T \epsilon^{(t)}, \label{eq:Gordon-epsilon}
\end{align}
where we used the Cauchy-Schwarz inequality, as well as the assumption that $\|\phi^{(t)}(\vec{x}^{(t)}) - \vec{x}^{(t)}\| \leq \epsilon^{(t)}$. Next, we will bound the term $\| \vec{L}^{(t)} - \vec{L}^{(t-1)}\|_{\infty}$ in terms of $\|\vec{\ell}^{(t)} - \vec{\ell}^{(t-1)}\|_{\infty}$. To this end, it follows that
\begin{align}
    \| \vec{L}^{(t)} - \vec{L}^{(t-1)}\|^2_{\infty} &= \| (\vec{\ell}^{(t)} \otimes \vec{x}^{(t)})^{\flat} - (\vec{\ell}^{(t-1)} \otimes \vec{x}^{(t-1)})^{\flat} \|^2_{\infty} \notag \\
    &= \| (\vec{\ell}^{(t)} \otimes \vec{x}^{(t)})^{\flat} - (\vec{\ell}^{(t-1)} \otimes \vec{x}^{(t)})^{\flat} + (\vec{\ell}^{(t-1)} \otimes \vec{x}^{(t)})^{\flat} - (\vec{\ell}^{(t-1)} \otimes \vec{x}^{(t-1)})^{\flat} \|^2_{\infty} \notag \\
    &= \| (( \vec{\ell}^{(t)} - \vec{\ell}^{(t-1)}) \otimes \vec{x}^{(t)})^{\flat} + (\vec{\ell}^{(t-1)} \otimes (\vec{x}^{(t)} - \vec{x}^{(t-1)}))^{\flat}\|^2_{\infty} \notag \\
    &\leq 2\| (( \vec{\ell}^{(t)} - \vec{\ell}^{(t-1)}) \otimes \vec{x}^{(t)})^{\flat}\|^2_{\infty} +2\|(\vec{\ell}^{(t-1)} \otimes (\vec{x}^{(t)} - \vec{x}^{(t-1)}))^{\flat}\|^2_{\infty} \label{eq:Gordon-triangle_Young} \\
    &= 2\|\vec{\ell}^{(t)} - \vec{\ell}^{(t-1)}\|^2_{\infty} \|\vec{x}^{(t)}\|^2_{\infty} + 2\|\vec{\ell}^{(t-1)}\|^2_{\infty} \|\vec{x}^{(t)} - \vec{x}^{(t-1)}\|^2_{\infty} \label{eq:outer-property} \\
    &\leq 2\|\vec{\ell}^{(t)} - \vec{\ell}^{(t-1)}\|^2_{\infty} + 2 \|\vec{\ell}\|^2_\infty \|\vec{x}^{(t)} - \vec{x}^{(t-1)}\|_\infty^2 \label{eq:final_bound-Phi},
\end{align}
where we used the triangle inequality together with Young's inequality in \eqref{eq:Gordon-triangle_Young}; the property that $\| (\vec{w} \otimes \vec{z})^{\flat}\|_{\infty} = \| \vec{w}\|_{\infty} \| \vec{z}\|_{\infty}$ in \eqref{eq:outer-property}; and the fact that $\|\vec{x}^{(t)}\|_\infty \leq 1$ in \eqref{eq:final_bound-Phi}. As a result, if we plug-in \eqref{eq:final_bound-Phi} to \eqref{eq:Gordon-epsilon} and we use the $(A, B)$-predictive bound of $\mathcal{R}_{\Phi}$ we can conclude that 
\begin{align*}
    \reg^T &\leq A + \|\vec{\ell}\|_\infty \sum_{t=1}^T \epsilon^{(t)} + B \sum_{t=1}^T \left(2 \|\vec{\ell}^{(t)} - \vec{\ell}^{(t-1)}\|^2_{\infty} + 2 \|\vec{\ell}\|_\infty^2 \|\vec{x}^{(t)} - \vec{x}^{(t-1)}\|^2_\infty \right) \\
    &= A + 2 B \sum_{t=1}^T \| \vec{\ell}^{(t)} - \vec{\ell}^{(t-1)}\|^2_{\infty} + 2 B \|\vec{\ell}\|_\infty^2 \sum_{t=1}^T \|\vec{x}^{(t)} - \vec{x}^{(t-1)}\|^2_\infty  + \|\vec{\ell}\|_\infty \sum_{t=1}^T \epsilon^{(t)},
\end{align*}
concluding the proof.
\end{proof}
\subsection{Proofs for \cref{section:Phi-regret}}
\label{appendix:proof-Phi-regret}

In this subsection we include the omitted proofs from \Cref{section:Phi-regret}. We commence with the proof of \Cref{proposition:R_sigma}. The corresponding construction follows that due to \citet{Farina21:Simple}, and it is highlighted in \Cref{algo:R_sigma}. 

\firstcircuit*

\begin{proof}
Consider the (linear) function $g_{\hat{\sigma}}^{(t)}: \R^{|\Sigma_j|} \ni \vec{x} \mapsto L_i^{(t)}(h_{\hat{\sigma}}(\vec{x})) - L_i^{(t)}(h_{\hat{\sigma}}(\mathbf{0}))$, and let $\vec{g}_{\hat{\sigma}}^{(t)} = (\vec{L}_i^{(t)}[\sigma_r, \hat{\sigma}])_{\sigma_r \succeq j}$ be the associated utility vector. As suggested in \Cref{algo:R_sigma}, the observed utility function $L_i^{(t)}$ at time $t$ is first used to construct $g_{\hat{\sigma}}^{(t)}$. Then, the latter function is given as input to $\mathcal{R}_{\mathcal{Q}_j}$. Thus, we may conclude that
\begin{equation*}
    \max_{\phi^* \in \Psi_{\hat{\sigma}}} \sum_{t=1}^T L_i^{(t)}(\phi^*) - \sum_{t=1}^T L_i^{(t)} \left( \phi_{\hat{\sigma} \rightarrow \vec{q}_{\hat{\sigma}}^{(t)}} \right) = \max_{\vec{q}_{\hat{\sigma}}^* \in \mathcal{Q}_{j}} \sum_{t=1}^T g_{\hat{\sigma}}^{(t)} (\vec{q}_{\hat{\sigma}}^*) - \sum_{t=1}^T g_{\hat{\sigma}}^{(t)}(\vec{q}_{\hat{\sigma}}^{(t)}).
\end{equation*}
In words, the cumulative regret incurred by $\mathcal{R}_{\hat{\sigma}}$ under the sequence of utility functions $L_i^{(1)}, \dots, L_i^{(T)}$ is equal to the regret incurred by $\mathcal{R}_{\mathcal{Q}_j}$ under the sequence of utility functions $g_{\hat{\sigma}}^{(t)}$. As a result, if we use the $(A, B)$-predictive bound assumed for the regret minimizer $\mathcal{R}_{\mathcal{Q}_j}$, it follows that the cumulative regret $\reg^T$ of $\mathcal{R}_{\hat{\sigma}}$ can be bounded as
\begin{equation*}
    \reg^T \leq A + B \sum_{t=1}^T \| \vec{g}_{\hat{\sigma}}^{(t)} - \vec{g}_{\hat{\sigma}}^{(t-1)} \|_{\infty}^2 \leq A + B \sum_{t=1}^T \| \vec{L}_i^{(t)} - \vec{L}_i^{(t-1)}\|_{\infty}^2, 
\end{equation*}
where we used the fact that $\vec{g}_{\hat{\sigma}}^{(t)} = (\vec{L}_i^{(t)}[\sigma_r, \hat{\sigma}])_{\sigma_r \succeq j}$. Finally, the claim regarding the complexity of \Cref{algo:R_sigma} is direct since we can store the vector $\vec{g}_{\hat{\sigma}}^{(t)}$ in $O(|\Sigma_j|)$ time.

\end{proof}

\begin{algorithm}[ht]
\SetAlgoLined
\DontPrintSemicolon
\KwIn{
\begin{itemize}[noitemsep]
    \item Player $i \in [n]$
    \item A trigger sequence $\hat{\sigma} = (j, a) \in \Sigma_{i}^*$
    \item An $(A, B)$-predictive regret minimizer $\mathcal{R}_{\mathcal{Q}_j}$ for the set $\mathcal{Q}_{j}$ 
\end{itemize}
}    
    \BlankLine
    \SetKwProg{Fn}{function}{:}{}
    \Fn{$\nextstr()$}{
        $\vec{q}_{\hat{\sigma}}^{(t)} \leftarrow \mathcal{R}_{\mathcal{Q}_j}.\nextstr()$\\
        \textbf{return} $\phi_{\hat{\sigma} \leftarrow \vec{q}_{\hat{\sigma}}^{(t)}}$
    }
    \Hline{}
    \SetKwProg{Fn}{function}{:}{}
    \Fn{$\obsut(L_i^{(t)})$}{
        Construct the linear function $g_{\hat{\sigma}}^{(t)}: \R^{|\Sigma_j|} \ni \vec{x} \mapsto L_i^{(t)}(h_{\hat{\sigma}}(\vec{x})) - L_i^{(t)}(h_{\hat{\sigma}}(\mathbf{0}))$\\
        $\mathcal{R}_{\mathcal{Q}_j}.\obsut(g_{\hat{\sigma}}^{(t)})$
    }
\caption{Predictive Regret Minimizer $\mathcal{R}_{\hat{\sigma}}$ for the set $\Psi_{\hat{\sigma}}$}
\label{algo:R_sigma}
\end{algorithm}

Next, we conclude the construction by combining the individual regret minimizers for all possible trigger sequences. In particular, we leverage the regret circuit of \Cref{proposition:regret_circuit-co} to obtain the following result.

\begin{proposition}
    \label{proposition:circ-pred}
    Consider an $(\alpha, \beta)$-predictive regret minimizer $\mathcal{R}_{\Delta}$ for the the simplex $\Delta(\Sigma^*_i)$, and $(A, B)$-predictive regret minimizers $\mathcal{R}_{\hat{\sigma}}$ for each $\hat{\sigma} \in \Sigma_i^*$, all with respect to the pair of dual norms $(\|\cdot\|_1, \|\cdot\|_{\infty})$. Then, there exists an algorithm which constructs a regret minimizer $\mathcal{R}_{\Psi_i}$ for the set $\co \Psi_{i}$ such that under any sequence of utility vectors $\vec{L}_i^{(1)}, \dots, \vec{L}_i^{(T)}$ its regret $\reg_{\Psi_i}^T$ can be bounded as
    \begin{equation*}
    \reg_{\Psi_i}^T \leq \alpha + A + (B + 4\beta |\Sigma_{i}|^2) \sum_{t=1}^T \| \vec{L}_i^{(t)} - \vec{L}_i^{(t-1)} \|_{\infty}^2.
    \end{equation*}
    Moreover, if the routines $\nextstr$ and $\obsut$ of $\mathcal{R}_{\Delta}$ and $\mathcal{R}_{\hat{\sigma}}$, for each $\hat{\sigma} \in \Sigma_i^*$, run in linear time on $|\Sigma_{i}|$, then the complexity of $\mathcal{R}_{\Psi}$ is $O(|\Sigma_{i}|^2)$.
\end{proposition}

The overall algorithm associated with this construction has been summarized in \Cref{algo:R_Psi}.

\begin{remark}
\label{remark:better_prediction}
To obtain better predictive bounds, the regret minimizer $\mathcal{R}_{\Delta}$ acting over the simplex in \Cref{proposition:circ-pred} will leverage the ``future'' iterates of all the individual regret minimizers. In particular, instead of using the typical one-recency bias mechanism $\vec{m}_{\lambda}^{(t)}[k] \defeq \langle \vec{L}_i^{(t-1)}, \vec{x}_k^{(t-1)} \rangle$, we will let $\vec{m}_{\lambda}^{(t)}[k] \defeq \langle \vec{L}_i^{(t-1)}, \vec{x}_k^{(t)} \rangle$. To this end, $\mathcal{R}_{\Delta}$ has to obtain the next iterate from each regret minimizer $\mathcal{R}_{\hat{\sigma}}$. This does not create complications given that the output of each $\mathcal{R}_{\hat{\sigma}}$ in the construction only depends on the observed utilities up to that time. On the other hand, it seems that there is no straightforward extension of this trick for \Cref{theorem:accelerating-Phi}, at the cost of a mismatch term of the form $\sum_{t=1}^T \| \vec{x}^{(t)} - \vec{x}^{(t-1)}  \|_\infty$.
\end{remark}

\begin{proof}[Proof of \Cref{proposition:circ-pred}]
First of all, \Cref{proposition:regret_circuit-co} implies that the accumulated regret can be bounded as 
\begin{equation}
    \label{eq:R_psi-initial}
    \reg_{\Psi_i}^{T} \leq \alpha + A + B \sum_{t=1}^T \| \vec{L}_i^{(t)} - \vec{L}_i^{(t-1)}\|_{\infty}^2 + \beta \sum_{t=1}^T \| \vec{\ell}_{\lambda}^{(t)} - \vec{m}_{\lambda}^{(t)} \|_{\infty}^2,
\end{equation}
where we used the fact that each regret minimizer $\mathcal{R}_{\hat{\sigma}}$ obtains as input the same utility function as $\mathcal{R}_{\Psi_i}$. We also used the notation $\vec{\ell}_{\vec{\lambda}}^t \in \R^{|\Sigma_i^*|}$ to represent the utility function given to $\mathcal{R}_{\Delta}$ as predicted by \Cref{proposition:regret_circuit-co}. Next, let us focus on bounding the norm $\|\vec{\ell}_{\lambda}^{(t)} - \vec{m}_{\lambda}^{(t)} \|_{\infty}^2$. In particular, it follows that for some index $s \in \{1, \dots, |\Sigma_i^*| \}$,
\begin{align*}
    \| \vec{\ell}_{\lambda}^{(t)} - \vec{m}_{\lambda}^{(t)}\|_\infty^2 &= \left( \langle \vec{L}_i^{(t)}, \vec{x}_s^{(t)} \rangle - \langle \vec{L}_i^{(t-1)}, \vec{x}_s^{(t)} \rangle \right)^2 \\
    &\leq \|\vec{L}_i^{(t)} - \vec{L}_i^{(t-1)}\|^2_\infty \|\vec{x}_s^{(t)}\|^2_1 \\
    &\leq 4 |\Sigma_i|^2 \|\vec{L}_i^{(t)} - \vec{L}_i^{(t-1)}\|^2_\infty,
\end{align*}
where we used the fact that $\|\vec{x}_s\|_1 \leq 2 |\Sigma_i|$. Thus, plugging this bound to \eqref{eq:R_psi-initial} gives the desired predictive bound. Finally, the complexity analysis for the $\nextstr$ function follows directly since the $\nextstr$ operation of each individual regret minimizer runs in $O(|\Sigma_{i}|)$, while the analysis of the $\obsut$ routine follows similarly to \cite[Theorem 4.6]{Farina21:Simple}, and it is therefore omitted.
\end{proof}

\begin{algorithm}[ht]
\SetAlgoLined
\DontPrintSemicolon
\KwIn{
\begin{itemize}[noitemsep]
    \item Player $i \in [n]$
    \item An $(A, B)$-predictive regret minimizer $\mathcal{R}_{\hat{\sigma}}$ for $\Psi_{\hat{\sigma}}$, for each $\hat{\sigma} \in \Sigma_i^*$
    \item An $(\alpha, \beta)$-predictive regret minimizer $\mathcal{R}_{\Delta}$ for $\Delta(\Sigma_i^*)$
\end{itemize}
}    
    \SetKwProg{Fn}{Function}{:}{}
    \Fn{$\nextstr()$}{
        $\vec{\lambda}_i^{(t)} \leftarrow \mathcal{R}_{\Delta}.\nextstr()$ \\
        \For{$\hat{\sigma} \in \Sigma^*_i$}{
        $\phi_{\hat{\sigma} \rightarrow \vec{q}_{\hat{\sigma}}^{(t)}} \leftarrow \mathcal{R}_{\hat{\sigma}}.\nextstr()$
        }
        \textbf{return} $\sum_{\hat{\sigma} \in \Sigma_i^*}\vec{\lambda}_i^{(t)}[\hat{\sigma}] \phi_{\hat{\sigma} \rightarrow \vec{q}_{\hat{\sigma}}^{(t)}}$ represented implicitly as $\{\vec{\lambda}_i^{(t)}[\hat{\sigma}], \vec{q}_{\hat{\sigma}}^{(t)} \}_{\hat{\sigma} \in \Sigma_i^*}$
    }
    \SetKwProg{Fn}{Function}{:}{}
    \Fn{$\obsut(L_i^{(t)})$}{
        \For{$\hat{\sigma} \in \Sigma_i^*$}{
        $\mathcal{R}_{\hat{\sigma}}.\obsut(L_i^{(t)})$
        }
        Construct the linear function $\ell^{(t)}_{\vec{\lambda}} : \vec{\lambda} \mapsto \sum_{\hat{\sigma} \in \Sigma_i^*} \vec{\lambda}[\hat{\sigma}] L_i^{(t)}\left(\phi_{\hat{\sigma} \rightarrow \vec{q}_{\hat{\sigma}}^{(t)}}\right)$\\
        $\mathcal{R}_{\Delta}.\obsut(\ell_{\lambda}^{(t)})$
    }
\caption{Predictive Regret Minimizer $\mathcal{R}_{\Psi_i}$ for the set $\co \Psi_{i}$}
\label{algo:R_Psi}
\end{algorithm}

Finally, we combine the previous pieces to prove \Cref{theorem:co-circuit}, which is recalled below.

\secondcircuit*

\begin{proof}
    The claim follows directly from \Cref{lemma:oftrl-predictive} using the fact that the range of the negative entropy DGF on the simplex $\Delta(\Sigma_i^*)$ is at most $\log |\Sigma_i|$; the predictive bound of \Cref{lemma:predictive-OMD}; \Cref{proposition:R_sigma} with the regret minimizer $\mathcal{R}_{\mathcal{Q}_j}$ instantiated using the dilatable global entropy DGF (\Cref{lemma:predictive-OMD}); and the predictive bound of the regret circuit for the convex hull derived in \Cref{proposition:circ-pred}.
\end{proof}
\subsection{Proofs for \cref{section:stability}}
\label{appendix:section_5}

We start this subsection with the proof that $\omw$ guarantees multiplicative stability.

\mulstabsimplex*

\begin{proof}
It is well-known that the update rule of $\omw$ on the simplex can be expressed in the following form:
\begin{equation*}
    \vec{x}^{(t)}[k] = \frac{e^{\eta \vec{\ell}^{(t-1)}[k] + \eta \vec{m}^{(t)}[k] - \eta \vec{m}^{(t-1)}[k]}}{\sum_{k'=1}^m e^{\eta \vec{\ell}^{(t-1)}[k'] + \eta \vec{m}^{(t)}[k'] - \eta \vec{m}^{(t-1)}[k']} \vec{x}^{(t-1)}[k']} \vec{x}^{(t-1)}[k],
\end{equation*}
for all $k \in [m]$ and $t \geq 2$. As a result, we have that 
\begin{equation*}
    \vec{x}^{(t)}[k] \leq \frac{e^{3\eta \|\vec{\ell}\|_\infty}}{\sum_{k'=1}^m e^{-3\eta \|\vec{\ell}\|_\infty} \vec{x}^{(t-1)}[k']} \vec{x}^{(t-1)}[k] = e^{6\eta \|\vec{\ell}\|_\infty} \vec{x}^{(t-1)}[k] \leq (1 + 12 \eta \|\vec{\ell}\|_\infty ) \vec{x}^{(t-1)}[k], \label{eq:mul-sta-ineq}
\end{equation*}
where we used that $\vec{\ell}^{(t-1)}[k'], \vec{m}^{(t)}[k'], \vec{m}^{(t-1)}[k'] \in [-\|\vec{\ell}\|_\infty, \|\vec{\ell}\|_\infty]$, for all $k' \in [m]$, the fact that $\sum_{k'} \vec{x}^{(t-1)}[k'] = 1$ since $\vec{x}^{(t-1)} \in \Delta^m$, and that $e^{x} \leq 1 + 2x$, for all $x \in [0, 1/2]$. Similarly, we have that 
\begin{align*}
    \vec{x}^{(t)}[k] \geq \frac{e^{-3\eta \|\vec{\ell}\|_\infty }}{\sum_{k'=1}^m e^{3\eta \|\vec{\ell}\|_\infty} \vec{x}^{(t-1)}[k']} \vec{x}^{(t-1)}[k] = e^{-6\eta \|\vec{\ell}\|_\infty} \vec{x}^{(t-1)}[k] &\geq (1 - 6 \eta \|\vec{\ell}\|_\infty) \vec{x}^{(t-1)}[k] \\
    &\geq (1 + 12 \eta \|\vec{\ell}\|_\infty)^{-1} \vec{x}^{(t-1)}[k],
\end{align*}
for $\eta \leq 1/(12 \|\vec{\ell}\|_\infty)$.
\end{proof}

\mulstabseq*

\begin{proof}
Let $\vec{S}^{(t-1)} \defeq \sum_{\tau=1}^{t-1} \vec{\ell}^{(\tau)}$. We claim that the next iterate of \eqref{eq:oftrl} with dilatable global entropy as DGF can be computed as follows. First, we compute recursively the quantities
\begin{equation}
    \label{eq:r}
    \vec{r}^{(t)}[j] \defeq \vec{\gamma}[j] \log \left( \sum_{a \in \mathcal{A}_j} \exp \left\{ \frac{\eta \vec{S}^{(t-1)}[(j,a)] + \eta \vec{m}^{(t)}[(j,a)] - \sum_{j': \sigma_{j'} = (j,a)} \vec{r}^{(t)}[j'] }{\vec{\gamma}[j]} \right\} \right)
\end{equation}
through a bottom-up tree traversal. Then, we determine the (local) behavioral strategies $\vec{b}_j \in \Delta(\mathcal{A}_j)$ at every decision point $j \in \mathcal{J}$ based on the following update rule:
\begin{equation}
    \label{eq:b}
    \vec{b}_j[a] \propto \exp \left\{ \frac{\eta \vec{S}^{(t-1)}[(j,a)] + \eta \vec{m}^{(t)}[(j,a)] - \sum_{j': \sigma_{j'} = (j,a)} \vec{r}^{(t)}[j'] }{\vec{\gamma}[j]} \right\}.
\end{equation}
Finally, the computed behavioral strategies are converted to the sequence-form representation. To argue about the multiplicative stability of the induced sequence, let us use the notation
\begin{equation}
    \label{eq:s}
    \vec{s}^{(t)}[(j,a)] \defeq \frac{1}{\vec{\gamma}[j]} \left( 2 \eta \vec{\ell}^{(t-1)}[(j,a)] - \eta \vec{\ell}^{(t-2)}[(j,a)] - \sum_{j': \sigma_{j'} = (j,a)} \mleft( \vec{r}^{(t)}[j'] - \vec{r}^{(t-1)}[j'] \mright) \right).
\end{equation}
Assuming that $\vec{m}^{(t)} \defeq \vec{\ell}^{(t-1)}$, it follows from \eqref{eq:r} that
\begin{align*}
    \vec{r}^{(t)}[j] &= \vec{\gamma}[j] \log \left( \sum_{a \in \mathcal{A}_j} \exp \left\{ \frac{\eta \vec{S}^{(t-2)}[(j,a)] + \eta \vec{\ell}^{(t-2)}[(j,a)] - \sum_{j': \sigma_{j'} = (j,a)} \vec{r}^{(t-1)}[j'] }{\vec{\gamma}[j]} \right\} e^{\vec{s}^{(t)}[(j,a)]} \right) \\
    &\leq \vec{r}^{(t-1)}[j] + \vec{\gamma}[j] \max_{a \in \mathcal{A}_j} \vec{s}^{(t)}[(j,a)].
\end{align*}
Similarly, we have that
\begin{equation*}
    \vec{r}^{(t)}[j] \geq \vec{r}^{(t-1)}[j] + \vec{\gamma}[j] \min_{a \in \mathcal{A}_j} \vec{s}^{(t)}[(j,a)] = \vec{r}^{(t-1)}[j] - \vec{\gamma}[j] \max_{a \in \mathcal{A}_j} (-\vec{s}^{(t)}[(j,a)]).
\end{equation*}
Thus, we have shown that
\begin{equation*}
    \left| \vec{r}^{(t)}[j] - \vec{r}^{(t-1)}[j] \right| \leq \vec{\gamma}[j] \max_{a \in \mathcal{A}_j} |\vec{s}^{(t)}[(j,a)]|.
\end{equation*}
Recalling the definition of $\vec{s}^{(t)}[(j,a)]$ given in \eqref{eq:s} we find that 
\begin{align}
    \left| \vec{r}^{(t)}[j] - \vec{r}^{(t-1)}[j] \right| &\leq \max_{a \in \mathcal{A}_j} \left| 2 \eta \vec{\ell}^{(t-1)}[(j,a)] - \eta \vec{\ell}^{(t-2)}[(j,a)] - \sum_{\sigma_{j'} = (j,a)} \mleft( \vec{r}^{(t)}[j'] - \vec{r}^{(t-1)}[j'] \mright) \right| \notag \\
    &\leq 3\eta + \max_{a \in \mathcal{A}_j} \sum_{j': \sigma_{j'} = (j,a)} \left| \vec{r}^{(t)}[j'] - \vec{r}^{(t-1)}[j'] \right|, \label{align:r-diff}
\end{align}
where we used the assumption that $\|\vec{\ell}^{(t-1)}\|_\infty, \|\vec{\ell}^{(t-2)}\|_\infty \leq 1$. Now \eqref{eq:b} can be equivalently expressed as 
\begin{equation*}
    \vec{b}^{(t)}_j[a] \propto \vec{b}_j^{(t-1)}[a] \exp \left\{ \frac{2 \eta \vec{\ell}^{(t-1)}[(j,a)] - \eta \vec{\ell}^{(t-2)}[(j,a)] - \sum_{j': \sigma_{j'} = (j,a) } (\vec{r}^{(t)}[j'] - \vec{r}^{(t-1)}[j'])}{\vec{\gamma}[j]} \right\}.
\end{equation*}
Using \eqref{align:r-diff} and the assumption that $\|\vec{\ell}^{(t-1)}\|_\infty, \|\vec{\ell}^{(t-2)}\|_\infty \leq 1$, it follows that 
\begin{equation*}
    \left| \frac{2 \eta \vec{\ell}^{(t-1)}[(j,a)] - \eta \vec{\ell}^{(t-2)}[(j,a)] - \sum_{j': \sigma_{j'} = (j,a) } (\vec{r}^{(t)}[j'] - \vec{r}^{(t-1)}[j'])}{\vec{\gamma}[j]} \right| = O(\eta),
\end{equation*}
where we used the definition of $\vec{\gamma}$ given in \eqref{eq:gamma}. As a result, similarly to the argument in the proof of \Cref{lemma:OMW-simplex-stability} we conclude that the sequence $(\vec{b}^{(t)}_j)$ is $O(\eta)$-multiplicative-stable. Finally, the sequence-form strategy $\vec{x}^{(t)}[(j,a)]$ is computed by taking the product of all $\vec{b}_{j'}^{(t)}[a']$ for all sequences $(j',a')$ on the path from the root to $(j,a)$. Given that there are at most $\mathfrak{D}$ sequences on every path, we may conclude that for any $\sigma \in \Sigma$,

\begin{equation*}
\vec{x}^{(t)}[\sigma] \leq (1 + O(\eta))^{\mathfrak{D}} \vec{x}^{(t-1)}[\sigma] \leq ( 1 + O(\eta \mathfrak{D})) \vec{x}^{(t-1)}[\sigma],     
\end{equation*}
for a sufficiently small $\eta = O(1/\mathfrak{D})$. Similar reasoning yields that $\vec{x}^{(t)}[\sigma] \geq (1 + O(\eta \mathfrak{D}))^{-1} \vec{x}^{(t-1)}[\sigma]$, concluding the proof.
\end{proof}

Next, we combine \Cref{lemma:OMD-stability,lemma:OMW-simplex-stability} to show \Cref{corollary:mul-stability}.

\begin{proof}[Proof of \Cref{corollary:mul-stability}]
Let us first focus on the regret minimizer $\mathcal{R}_{\hat{\sigma}}$, for some arbitrary $\hat{\sigma} = (j,a) \in \Sigma_i^*$. First, as predicted by \Cref{theorem:accelerating-Phi}, the utility function $\vec{L}_i^{(t)}$ is constructed as $\vec{L}_i^{(t)} \defeq (\vec{\ell}_i^{(t)} \otimes \vec{x}_i^{(t)} )^\flat $. \Cref{proposition:regret_circuit-co} implies that this is the same utility observed by $\mathcal{R}_{\hat{\sigma}}$. Moreover, from the construction of \Cref{algo:R_sigma} we can conclude that the utility $\vec{g}_{\hat{\sigma}}^{(t)}$ observed by $\mathcal{R}_{\mathcal{Q}_j}$ will be such that $\| \vec{g}_{\hat{\sigma}}^{(t)} \| \leq 1 $ given that $\| \vec{x}_i^{(t)}\|_\infty \leq 1$ (since $\vec{x}_i^{(t)} \in \mathcal{Q}_i$) and $\| \vec{\ell}_i^{(t)} \|_\infty \leq 1$ by the normalization assumption. Thus, we conclude from \Cref{lemma:OMD-stability} that the output sequence of $\mathcal{R}_{\mathcal{Q}_j}$ will be $O(\eta \mathfrak{D}_i)$-multiplicative-stable. Furthermore, the construction of \Cref{algo:R_sigma} immediately implies that the output sequence of $\mathcal{R}_{\hat{\sigma}}$ will also be $O(\eta \mathfrak{D}_i)$.

Next, we establish the claim regarding the stability of $\mathcal{R}_{\Delta}$. Indeed, it is easy to see that the utility $\vec{\ell}^{(t)}_{\lambda}$ observed by $\mathcal{R}_{\Delta}$ is such that $\| \vec{\ell}_{\lambda}\|_{\infty} = O(\|\mathcal{Q}_i\|_1)$, and the same holds for the prediction $\vec{m}^{(t)}_{\lambda}$. Thus, \Cref{lemma:OMW-simplex-stability} completes the proof.
\end{proof}

Next, we focus on the proof of \Cref{theorem:stability-EFCE}. To this end, we leverage the approach of \citet{Kruckman10:Elementary}, who provided an alternative proof of the classic Markov chain tree theorem using linear-algebraic techniques. We commence by stating some elementary properties of the determinant.

\begin{fact}
    \label{fact:det_properties}
The following properties hold:
\begin{itemize}[leftmargin=5mm]
    \item The determinant is a multilinear function with respect to the rows and columns of the matrix:
    \begin{equation*}
     \det(\vec{u}_1, \dots, \alpha \vec{u}_k + \beta\vec{u}'_k, \dots, \vec{u}_m) = \alpha \det(\vec{u}_1, \dots, \vec{u}_k, \dots, \vec{u}_m) + \beta \det(\vec{u}_1, \dots, \vec{u}'_k, \dots, \vec{u}_m),
    \end{equation*}
    for any $\vec{u}_1,\dots,\vec{u}_m \in \R^m$, $\vec{u}'_k \in \R^m$, and $\alpha,\beta\in\R$;
    \item If any two rows or columns of $\mat{A}$ are equal, then $\det(\mat{A}) = 0$;
    \item The determinant remains invariant under permutations.
\end{itemize}
\end{fact}

Given a matrix $\mat{A}$, the minor $\minor^{(i, j)}(\mat{A})$ is the matrix formed from $\mat{A}$ after deleting its $i$-th row and its $j$-th column. Then, the \emph{cofactor} is defined as $\co^{(i, j)}(\mat{A}) = (-1)^{i + j} \det \left( \minor^{(i, j)}(\mat{A}) \right)$, while the \emph{adjugate} (or adjoint) matrix $\adj(\mat{A})^{\top}$ is the matrix with entries the corresponding cofactors of $\mat{A}$; that is, $\adj(\mat{A})[(i, j)] \defeq \co^{(j, i)}(\mat{A})$. With this notation at hand, we are ready to state the following characterization due to \cite[Theorem 3.4]{Kruckman10:Elementary}:

\begin{theorem}[\cite{Kruckman10:Elementary}]
    \label{theorem:stationary_solution}
    Consider an ergodic $m$-state Markov chain with transition matrix $\mat{M}$. If $\vec{x} \in \mathbb{R}^m$ is such that $\vec{x}[i] \defeq \adj(\mathcal{L})[(i, i)]$, where $\mathcal{L} \defeq \mat{M} - \mat{I}_m$ is the Laplacian of the system, $\vec{x}$ is an eigenvector of $\mat{M}$ with a corresponding eigenvalue of $1$. That is, $\mat{M} \vec{x} = \vec{x}$.
\end{theorem}

A key step of our proof for \Cref{theorem:stability-EFCE} uses this theorem in order to characterize the stationary distribution of a certain (ergodic) Markov chain. Incidentally, an alternative characterization can be provided using the classic Markov chain tree theorem. In particular, a central component of the latter theorem is the notion of a \emph{directed tree}:
\begin{definition}[Directed Tree]
    \label{definition:directed_tree}
    A graph $G = (V, E)$ is said to be a \emph{directed tree} rooted at $u \in V$ if (i) it does not contain any cycles, and (ii) $u$ has \emph{no} outgoing edges, while every other node has exactly \emph{one} outgoing edge.
\end{definition}
We will represent with $\mathcal{D}_i$ the set of all graphs which have property (ii) with respect to a node $i \in [m]$. Moreover, we will use $\mathcal{T}_i$ to represent the subset of $\mathcal{D}_i$ which also has property (i) of \Cref{definition:directed_tree}. For a matrix $\mat{D} \in \mathcal{D}_i$, we define a matrix $\map(\mat{D})$ so that $\map(\mat{D})_{(j, k)} = 1$ if $(k, j) \in 
E(\mat{D})$, and $0$ otherwise. The following lemma will be of particular use for our purposes.

\begin{lemma}[\cite{Kruckman10:Elementary}]
    \label{lemma:cycles}
    Consider some $m \times m$ matrix $\mat{D} \in \mathcal{D}_i$, and let $R_i$ be the determinant of the Laplacian matrix $\mathcal{L} \defeq \map(\mat{D}) - \mat{I}$ after replacing the $i$-th column with the $i$-th standard unit vector $\vec{e}[i]$. Then, $R_i = (-1)^{m-1}$ if $\mat{D} \in \mathcal{T}_i$, i.e. $\mat{D}$ contains no (directed) cycles. Otherwise, $R_i = 0$.
\end{lemma}

Before we proceed with the technical proof of \Cref{lemma:convex_characterization}, we also state a useful elementary fact.

\begin{fact}
    \label{fact:obvious}
    The adjugate matrix at $(i, i)$ is equal to the determinant of $\mat{A}$ after we replace the $i$-th column with the vector $\vec{e}[i]$.
\end{fact}

\rankone*
\begin{proof}
Let us consider the Laplacian matrix $\mathcal{L} = \mat{M} - \mat{I}_m$, and the quantities $\Sigma_i := \adj(\mathcal{L})[(i, i)]$. We shall first characterize the structure of $\Sigma_i$'s. By symmetry, we can focus without loss of generality on the term $\Sigma_1$. We know from \Cref{fact:obvious} that $\Sigma_1$ can be expressed as 
\begin{equation}
    \label{eq:first_determinant}
    \Sigma_1 = \det(\vec{e}[1], \vec{v} + \vec{c}_2 - \vec{e}[2], \dots,  \vec{v} + \vec{c}_m - \vec{e}[m]),
\end{equation}
where $\vec{c}_j$ represents the $j$-th column of $\mat{C}$. Now if $\vec{e}_{j, k} \defeq \vec{e}[j] - \vec{e}[k]$, given that $\mat{M}$ is column-stochastic we have that
\begin{equation*}
    \vec{e}[j] - \vec{v} - \vec{c}_j = \sum_{k=1}^m (\vec{e}[j] - \vec{e}[k]) \vec{v}[k] + \sum_{k=1}^m (\vec{e}[j] - \vec{e}[k]) \vec{c}_{j}[k] = \sum_{k=1}^m \vec{e}_{j, k} \vec{v}[k] + \sum_{k=1}^m \vec{e}_{j, k} \vec{c}_{j}[k].
\end{equation*}
Next, if we plug-in this expansion to \eqref{eq:first_determinant} it follows that 
\begin{equation}
    \label{eq:Sigma_1}
    \Sigma_1 = \det \left( \vec{e}[1], \sum_{k=1}^m \vec{e}_{k, 2} \vec{v}[k] + \sum_{k=1}^m \vec{e}_{k, 2} \vec{c}_2[k], \dots, \sum_{k=1}^m \vec{e}_{k, m} \vec{v}[k] + \sum_{k=1}^m \vec{e}_{k, m} \vec{c}_m[k] \right).
\end{equation}
By multilinearity of the determinant (\Cref{fact:det_properties}), $\Sigma_1$ can be expressed as the sum of terms, with a single term of the form
\begin{equation}
    \label{eq:term_c}
    \det \left( \vec{e}[1], \sum_{k=1}^m \vec{e}_{k, 2} \vec{c}_2[k], \dots, \sum_{k=1}^m \vec{e}_{k, m} \vec{c}_m[k] \right),
\end{equation}
independent on $\vec{v}$, while any other term can be expressed in the form
\begin{equation}
    \label{eq:term_v}
\det \left( \vec{e}[1], \vec{z}_2, \dots, \sum_{k=1}^m \vec{e}_{k, j} \vec{v}[k], \dots, \vec{z}_m \right),    
\end{equation}
for some index $j$, where $\vec{z}_{\ell}$ is either $\sum_{k=1}^m \vec{e}_{k, \ell} \vec{v}[k]$ or $\sum_{k=1}^m \vec{e}_{k, \ell} \vec{c}_{\ell}[k]$. Now let us first analyze each term of \eqref{eq:term_v}. We will show that it can be equivalently expressed so that the vector $\vec{v}$ appears only in a single column. Indeed, consider any other column in the matrix involved in the determinant of \eqref{eq:term_v}, expressed in the form $\sum_{k=1}^m \vec{e}_{k, \ell} \vec{v}[k]$, for some index $\ell \neq j$, if such column exists. Then, if we subtract the $j$-th column from that column it would take the form 
\begin{equation*}
    \sum_{k=1}^m \vec{e}_{k, \ell} \vec{v}[k] - \sum_{k=1}^m \vec{e}_{k, j} \vec{v}[k] = \sum_{k=1}^m (\vec{e}[j] - \vec{e}[\ell]) \vec{v}[k] = \lambda \vec{e}_{j, \ell},
\end{equation*}
where recall that $\lambda$ is the sum of the entries of vector $\vec{v}$, while this subtraction operation does not modify the value of the determinant. Thus, by multinearity, the determinant \eqref{eq:term_v} is equal to 
\begin{equation}
    \label{eq:lambda-term}
    \lambda^{p} \det \left( \vec{e}[1], \vec{z}'_{2}, \dots, \sum_{k=1}^m \vec{e}_{k, j} \vec{v}[k], \dots, \vec{z}'_m \right),
\end{equation}
where $\vec{z}'_\ell$ is either $\sum_{k=1}^m \vec{e}_{k, \ell} \vec{c}_\ell[k]$ or $\vec{e}_{j, \ell}$, and $0 \leq p \leq m-2$. Next, if we use again the multilinearity property, the term in \eqref{eq:lambda-term} can be expressed as a sum of terms each of which has the form
\begin{equation*}
    \left(\lambda^p \vec{v}[q] \prod_{(s, w) \in S} \mat{C}[(s, w)] \right) \det(\vec{e}[1], \vec{e}_{\cdot, 2}, \dots, \vec{e}_{\cdot, m}),
\end{equation*}
where $|S| = m - p - 2$. (For notational simplicity we used the notation $\vec{e}_{\cdot, 2}, \dots, \vec{e}_{\cdot, m}$ to suppress the first index.) Thus, the induced determinant $\det(\vec{e}[1], \vec{e}_{\cdot, 2}, \dots, \vec{e}_{\cdot, m})$ matches after a suitable permutation the form of \Cref{lemma:cycles} associated with some matrix $\mat{D} \in \mathcal{D}_i$. As a result, it can either be $0$ or $(-1)^{m-1}$, depending on whether the corresponding graph has a (directed) cycle. Similar reasoning applies for the determinant in \eqref{eq:term_c}, which can be expressed as a sum of terms 
\begin{equation*}
(-1)^{m-1} \prod_{(s, w) \in S} \mat{C}[(s, w)],
\end{equation*}
where $|S| = m - 1$. Overall, we have shown that each $\Sigma_i$ (due to symmetry) can be expressed in the form
\begin{equation}
    \label{eq:nominator}
    (-1)^{m-1} \sum_{j \in F_i} \lambda^{p_j} (\vec{v}[q_j])^{b_j}  \prod_{(s,w) \in S_j} \mat{C}[(s,w)], 
\end{equation}
where for all $j$ it holds that $b_j \in \{0, 1\}$, and $|S_j| = m - p_j - b_j- 1$. Next, we will focus on characterizing the term $\Sigma \defeq \sum_{i=1}^m \Sigma_i$. In particular, the stationary distribution $\vec{\pi}$ of $\mat{M}$ is such that
\begin{equation}
    \left( \mat{C} + \vec{v} \vec{1}^{\top} \right) \vec{\pi} = \vec{\pi} \iff \mat{C} \vec{\pi} + \vec{v} = \vec{\pi} \iff (\mat{I}_m - \mat{C}) \vec{\pi} = \vec{v},
\end{equation}
where we used that $\vec{1}^{\top} \vec{\pi} = 1$ since $\vec{\pi} \in \Delta^m$. Moreover, we claim that the matrix $\mat{I}_m - \mat{C}$ is invertible. Indeed, the sum of the columns of $\mathbf{C}$ is $1 - \lambda$, and subsequently it follows that the maximum eigenvalue of $\mathbf{C}$ is $(1 - \lambda)$. In turn, this implies that all the eigenvalues of $\mat{I}_m - \mat{C}$ are at least $\lambda > 0$. 
As a result, we can use Cramer's rule to obtain an explicit formula for the solution of the linear system with respect to the first coordinate of $\vec{\pi}$:
\begin{equation}
    \label{eq:cramer}
    \vec{\pi}[1] = \frac{\det(\vec{v}, \vec{e}[2] - \vec{c}_2, \dots, \vec{e}[m] - \vec{c}_m)}{\det(\vec{e}[1] - \vec{c}_1, \vec{e}[2] - \vec{c}_2, \dots, \vec{e}[m] - \vec{c}_m)}.
\end{equation}
Moreover, it follows that 
\begin{align}
    \det(\vec{v}, \vec{e}[2] - \vec{c}_2, \dots, \vec{e}[m] - \vec{c}_m) \notag &= \det(\vec{v}, \vec{e}[2] - \vec{c}_2 - \vec{v}, \dots, \vec{e}[m] - \vec{c}_m - \vec{v}) \notag \\
    &= \det(\vec{v} + (\lambda \vec{e}[1] - \vec{v}), \vec{e}[2] - \vec{c}_2 - \vec{v}, \dots, \vec{e}[m] - \vec{c}_m - \vec{v}) \label{eq:zero_det} \\
    &= \lambda \det(\vec{e}[1], \vec{e}[2] - \vec{c}_2 - \vec{v}, \dots, \vec{e}[m] - \vec{c}_m - \vec{v}), \notag
\end{align}
where in \eqref{eq:zero_det} we used the fact that $\det(\lambda \vec{e}[1] - \vec{v}, \dots, \vec{e}[m] - \vec{c}_m - \vec{v}) = 0$. Thus, if we use the definition of $\Sigma_1$, \Cref{fact:obvious}, and \eqref{eq:cramer}, we arrive at the following conclusion:
\begin{equation*}
\vec{\pi}[1] = \lambda \frac{\Sigma_1}{\det\left( \mat{I}_m - \mat{C} \right)}.
\end{equation*}
But we can also infer from \Cref{theorem:stationary_solution} that $\vec{\pi}_1 = \Sigma_1/\Sigma$, implying the following identity:
\begin{equation}
    \label{eq:formula}
    \det(\mat{I}_m - \mat{C}) = \lambda \sum_{i=1}^m \Sigma_i.
\end{equation}
In fact, we have shown this formula for \emph{any} vector $\lambda \vec{p}$, where $\vec{p}$ is a probability distribution and $\lambda > 0$. Thus, it must also hold for $\mathbf{v} \defeq \frac{\lambda}{m} \vec{1}$. That is, 
\begin{equation}
    \label{eq:denominator}
    \det(\mat{I}_m - \mat{C}) = \lambda (-1)^{m-1} \sum_{j \in F} C_j \lambda^{p_j + b_j} \prod_{(s,w) \in S_j} \mat{C}[(s,w)],
\end{equation}
where $|S_j| \leq m - 1 - p_j$, $C_j = C_j(m)$ is a positive parameter independent on the entries of $\vec{v}$ and $\mat{C}$, and $F = \bigcup_i F_i$. Finally, given that the vector $\vec{\pi} \in \Delta^m$ with $\vec{\pi}[i] = \Sigma_i/\Sigma$ is the (unique) stationary distribution of $\mat{M}$, the claim follows directly from \eqref{eq:nominator}, \eqref{eq:formula}, and \eqref{eq:denominator}.
\end{proof}

\begin{corollary}
    \label{corollary:w_formula}
Let $\mat{M}$ be the transition matrix of an $m$-state Markov chain such that $\mat{M} \defeq \vec{v} \vec{1}^{\top} + \mat{C}$, where $\mat{C}$ is a matrix with strictly positive entries and columns summing to $1 - \lambda$, and $\vec{v}$ is a vector with strictly positive entries summing to $\lambda$. Moreover, let $\vec{v} = \vec{r}/l$, for some $l > 0$. Then, if $\vec{\pi}$ is the stationary distribution of $\mat{M}$, there exists, for each $i \in [m]$, a (non-empty) finite set $F_i$ and $F = \bigcup_i F_i$, and corresponding parameters $b_j \in \{0, 1\}, 0 \leq p_j \leq m-2, |S_j| = m - p_j - b_j - 1$, for each $j \in F_i$, such that the $i$-th coordinate of the vector $ \vec{w} \defeq l \vec{\pi}$ can be expressed as
\begin{equation}
    \label{eq:w_formula}
    \vec{w}[i] = \frac{\sum_{j \in F_i} \lambda^{p_j + 1} (\vec{r}[q_j])^{b_j} l^{1 - b_j}  \prod_{(s,w) \in S_j} \mat{C}[(s,w)]}{\sum_{j \in F} C_j \lambda^{p_j + b_j} \prod_{(s,w) \in S_j}\mat{C}[(s,w)]},
    \end{equation}
    where $C_j = C_j(m)$ is a positive constant.
\end{corollary}
\begin{proof}
The proof follows directly from the formula derived in \Cref{lemma:convex_characterization}.
\end{proof}

This expression for the stationary distribution was derived specifically to characterize the multiplicative stability of the fixed points associated with $\efce$. In particular, this will be shown directly from \Cref{corollary:Markov_stability}, which is recalled next.

\markovstable*

\begin{proof}
Consider some coordinate $i \in [m]$, and let 
\begin{equation*}
    V_j \defeq \lambda^{p_j + 1} (\vec{r}[q_j])^{b_j} l^{1 - b_j}  \prod_{(s,w) \in S_j} \mat{C}[(s,w)],
\end{equation*}
for some $j \in F_i$. Also let $V_j'$ be the corresponding quantity with respect to $\mat{M}'$. Then, by assumption we have that
\begin{equation*}
    V_j' \leq (1 + \kappa)^{p_j + 1} (1 + \gamma) (1 + \kappa)^{|S_j|} V_j \leq (1 + \gamma) (1 + \kappa)^{m} V_j,
\end{equation*}
where we used the fact that $|S_j| + p_j + 1 \leq m$ by \Cref{corollary:w_formula}. Moreover, for a sufficiently small $\kappa = O(1/m)$, we can infer that $V_j' \leq (1 + \gamma)(1 + O(\kappa m)) V_j' = (1 + (\gamma + O(\kappa m))) V_j$. In turn, this implies that $\sum_{j \in F_i} V_j' \leq (1 + (\gamma + O(\kappa m))) \sum_{j \in F_i} V_j'$. Moreover, we can show that the denominator of \eqref{eq:w_formula} induces an extra additive factor of $O(\kappa m)$ in the multiplicative stability, implying that $\vec{w}'[i] \leq (1 + (\gamma + O(\kappa m))) \vec{w}[i]$, for any $i \in [m]$. Similarly, it follows that $\vec{w}'[i] \geq (1 + (\gamma + O(\kappa m)))^{-1} \vec{w}[i]$.
\end{proof}

Next, we will use this statement to prove \Cref{proposition:stability_extend}, which is recalled below.

\stabextend*

We note that it is tacitly assumed that the vectors $\vec{\lambda}_i^{(t)}, \vec{q}_{\hat{\sigma}}^{(t)}$ and $\vec{x}_{(j \in J)}$, involved in \Cref{proposition:stability_extend}, have strictly positive coordinates; this is indeed the case under our dynamics (\Cref{fig:algo}).

\begin{proof}[Proof of \Cref{proposition:stability_extend}]
Let us focus on the stability analysis of \Cref{algo:Extend} as the rest of the claim follows from \cite[Proposition 4.14]{Farina21:Simple}. In particular, for consistency with the terminology of \Cref{corollary:Markov_stability}, let us define
\begin{equation*}
    \mat{C}[(a_r, a_c)] \defeq \vec{\lambda}_i[(j^*, a_c)] \vec{q}_{(j^*, a_c)}[(j^*, a_r)] + \left(1 - \sum_{\hat{\sigma} \preceq (j^*, a_c)} \vec{\lambda}_i[\hat{\sigma}] \right) \mathbbm{1} \{ a_r = a_c \},
\end{equation*}
and $l \defeq \vec{x}_i[\sigma_p]$. We will show that the conditions of \Cref{corollary:Markov_stability} are satisfied: 
\begin{itemize}
    \item[(i)] The entries of matrix $\mat{C}$ are $O(\kappa)$-multiplicative-stable. In particular, this follows from the fact that $1 - \sum_{\hat{\sigma} \preceq (j^*, a_c)} \vec{\lambda}_i[\hat{\sigma}] = \sum_{\tilde{\sigma} \in \tilde{\Sigma}_i} \vec{\lambda}_i[\tilde{\sigma}]$, for some $\tilde{\Sigma}_i \subseteq \Sigma^*_i$, since $\vec{\lambda}_i \in \Delta(\Sigma^*_i)$. The latter term is clearly $\kappa$-multiplicative-stable;
    \item[(ii)] The sum of the entries of $\vec{v}^t \defeq \vec{r}^t/l^t$ is $\kappa$-multiplicative-stable. To see this, note that the sum of each column of $\mat{C}$ can be expressed as $\sum_{\tilde{\sigma} \in \tilde{\Sigma}_i} \vec{\lambda}_i[\tilde{\sigma}]$, and as a result, since the matrix $\mat{C} + \frac{1}{l} \vec{r} \vec{1}^{\top}$ is stochastic, we can infer that the sum of the entries of $\vec{v}$ can also be expressed as $\sum_{\tilde{\sigma} \in \tilde{\Sigma}_i} \vec{\lambda}_i[\tilde{\sigma}]$ since $\vec{\lambda}$ is a vector on the simplex. But the latter term is clearly $\kappa$-multiplicative-stable, as desired;
    \item[(iii)] The sequence $(\vec{r}^{(t)})$ is $\gamma + O(\kappa)$-multiplicative-stable. This assertion can be directly verified from the definition of $\vec{r}$ in \Cref{algo:Extend};
    \item[(iv)] The sequence of scalars $(l^{(t)})$ is $\gamma$-multiplicative-stable. Indeed, this follows directly from the assumption that the sequence $(\vec{x}_i^{(t)})$ is $\gamma$-multiplicative-stable.
\end{itemize}
As a result, we can apply \Cref{corollary:Markov_stability} to conclude the proof.
\end{proof}

\fpth*

\begin{proof}
Our argument proceeds inductively. For a root information set $j \in \mathcal{J}_{i}$, \Cref{proposition:stability_extend} implies $O(\kappa |\mathcal{A}|)$-multiplicative-stability for any induced partial fixed point; this follows given that the $\emptyseq$-partial fixed point is trivially $0$-multiplicative-stable. Next, the theorem follows inductively given that by \Cref{proposition:stability_extend} each sequence can only incur an additive factor of $O(\kappa |\mathcal{A}|)$ in the multiplicative stability bound with respect to the preceding sequences.
\end{proof}

\begin{remark}
\label{remark:precise_bound}
More precisely, if $F_{i} \defeq \max_{j_1 \prec j_2 \prec \dots \prec j_d} \sum_{i=1}^d |\mathcal{A}_{j_i}|$, with $j_1, \dots, j_d \in \mathcal{J}_{i}$, we can show that the sequence of fixed points is $O(\kappa F_{i})$-multiplicative-stable. Observe that $F_{i}$ can be trivially upper bounded by $|\mathcal{A}_{i}| \mathfrak{D}_{i}$, as well as the number of sequences $|\Sigma_{i}|$.
\end{remark}

\begin{algorithm}[ht]
\SetAlgoLined
\DontPrintSemicolon
\KwIn{
\begin{itemize}[noitemsep]
    \item $\phi_i = \sum_{\hat{\sigma} \in \Sigma_i^*} \vec{\lambda}_i[\hat{\sigma}] \phi_{\hat{\sigma} \rightarrow \vec{q}_{\hat{\sigma}}} \in \co \Psi_{i}$ 
    \item $J \subseteq \mathcal{J}_{i}$ trunk for player $i$ 
    \item $j^* \in \mathcal{J}_{i}$ information set not in $J$ with an immediate predecessor in $J$
    \item $\vec{x}_i \in \R_{\geq 0}^{|\Sigma_{i}|}$ $J$-partial fixed point of $\phi$
\end{itemize}
}    
\KwOut{$\vec{x}_i' \in \mathbb{R}_{\geq 0}^{|\Sigma_{i}|}$ $(J \cup \{j^*\})$-partial fixed point of $\phi$}
Let $\vec{r} \in \R_{\geq 0}^{|\mathcal{A}_{j^*}|}$ be defined as $\vec{r}[a] \defeq \sum_{j' \preceq \sigma_{j^*}} \sum_{a' \in \mathcal{A}_{j'}} \vec{\lambda}_i[(j', a')] \vec{q}_{(j', a')}[(j^*, a)] \vec{x}_i[(j', a')]$ \\
Let $\mat{W} \in \vec{x}_i[\sigma_{j^*}] \mathbb{S}^{|\mathcal{A}_{j^*}|}$ be the matrix with entries $\mat{W}[a_r, a_c]$ defined, for $a_r, a_c \in \mathcal{A}_{j^*}$, as $\vec{r}[a_r] + \left( \vec{\lambda}_i[(j^*, a_c)] \vec{q}_{(j^*, a_c)}[(j^*, a_r)] + \left(1 - \sum_{\hat{\sigma} \preceq (j^*, a_c)} \vec{\lambda}_i[\hat{\sigma}] \right) \mathbbm{1} \{ a_r = a_c \} \right) \vec{x}_i[\sigma_{j^*}]$ \label{line:W} \\
\If{$\vec{x}_i[\sigma_{j^*}] = 0$}{
$\vec{w} \leftarrow \mathbf{0} \in \mathbb{R}_{\geq 0}^{|\mathcal{A}_{j^*}|}$
}
\Else{
$\vec{b} \in \Delta(\mathcal{A}_{j^*}) \leftarrow$ stationary distribution of $\frac{1}{\vec{x}_i[\sigma_{j^*}]}\mat{W}$ \\
$\vec{w} \rightarrow \vec{x}_i[\sigma_{j^*}] \vec{b}$
}
$\vec{x}_i' \leftarrow \vec{x}_i$\\
\For{$a \in \mathcal{A}_{j^*}$}{
$\vec{x}_i'[(j^*, a)] \leftarrow \vec{w}[(j^*, a)]$
}
\caption{$\textsc{Extend}(\phi_i, J, j^*, \vec{x})$; \citep{Farina21:Simple}}
\label{algo:Extend}
\end{algorithm}

\begin{algorithm}[ht]
\SetAlgoLined
\DontPrintSemicolon
\KwIn{ $\phi_i = \sum_{\hat{\sigma} \in \Sigma_i^*} \vec{\lambda}_i[\hat{\sigma}] \phi_{\hat{\sigma} \rightarrow \vec{q}_{\hat{\sigma}}} \in \co \Psi_{i}$ 
}    
\KwOut{$\vec{q}_i \in \mathcal{Q}_{i}$ such that $\vec{q}_i = \phi_i(\vec{q}_i)$}
$\vec{q}_i \leftarrow \mathbf{0} \in \mathbb{R}^{|\Sigma_{i}|}, \vec{q}_i[\emptyseq] \leftarrow \emptyseq$ \\
$J \leftarrow \emptyseq$ \\
\For{$j \in \mathcal{J}_{i}$ in top-down order}{
$\vec{q}_i \leftarrow \textsc{Extend}(\phi_i, J, j, \vec{q}_i^*)$ \\
$J = J \cup \{j\}$
}
\textbf{return} $\vec{q}^*_i$
\caption{$\textsc{FixedPoint}(\phi_i)$; \citep{Farina21:Simple}}
\label{algo:FP-EFCE}
\end{algorithm}

\subsection{Proofs from \texorpdfstring{\Cref{subsection:completing}}{Section 4.3}}
    \label{appendix:everything}

We begin this subsection with the proof of \Cref{claim:aux}, which is recalled below.

\utilinfty*

\begin{proof}
For a profile of mixed sequence-form strategies $(\vec{q}_{1}, \dots, \vec{q}_{n})$, the utility of player $i$ can be expressed as
\begin{equation*}
    u_{i}(\vec{q}_{1}, \dots, \vec{q}_{n}) = \sum_{z \in \mathcal{Z}} p_{c}(z) u_{i}(z) \prod_{k=1}^n \vec{q}_{k}(\sigma_{k, z}).
\end{equation*}
As a result, given that (by assumption) $|u_{i}(z)| \leq 1$ for all $z \in \mathcal{Z}$, it follows that \begin{align}
    \| \vec{\ell}_i^{(t)} - \vec{\ell}_i^{(t-1)}\|_{\infty} &\leq \sum_{z \in \mathcal{Z}} \left\lvert \prod_{k \neq i}^n \vec{q}_k^{(t)}(\sigma_{k, z}) - \prod_{k \neq i}^n \vec{q}_k^{(t-1)}(\sigma_{k, z}) \right\rvert \notag \\
    &\leq \sum_{z \in \mathcal{Z}} \sum_{k \neq i}^n \left\lvert \vec{q}_k^{ (t)}(\sigma_{k, z}) - \vec{q}_k^{(t-1)}(\sigma_{k, z}) \right\rvert, \label{eq:q-bound}
\end{align}
where in the last bound we used the well-known inequality 
\begin{equation*}
    |(a_1 a_2 \dots a_m) - (b_1 b_2 \dots b_m)| \leq \sum_{i=1}^m |a_i - b_i| (a_1 \dots a_{i-1}) (b_{i+1} \dots b_m) \leq \sum_{i=1}^m |a_i - b_i|,
\end{equation*}
for any $a_1, \dots, a_m, b_1,\dots, b_m \in [0, 1]$. Finally, from \eqref{eq:q-bound} we can conclude that 
\begin{equation*}
    \| \vec{\ell}_i^{(t)} - \vec{\ell}_i^{(t-1)} \|_{\infty} \leq \sum_{k \neq i}^{n} \sum_{z \in \mathcal{Z}} \left\lvert \vec{q}_k^{(t)}(\sigma_{k, z}) - \vec{q}_k^{(t-1)}(\sigma_{k, z}) \right\rvert \leq |\mathcal{Z}| \sum_{k \neq i}^{n} \| \vec{q}_k^{(t)} - \vec{q}_k^{(t-1)} \|_1.
\end{equation*}
Finally, the claim follows from a standard application of Young's inequality.
\end{proof}

Next, we include the proof of \Cref{theorem:main}.

\begin{proof}[Proof of \Cref{theorem:main}]
For a player $i \in [n]$ we let $\vec{\mu}_i^{(t)}$ be any probability distribution on the set $\Pi_{i}$ such that $\E_{\vec{\pi}_i \sim \mu_i^{(t)}}[\vec{\pi}_i] = \vec{q}_i^{(t)}$, where $\vec{q}_i^{(t)}$ is the output of the regret minimizer operating over the \emph{mixed} sequence-form strategy polytope $\mathcal{Q}_{i}$, realized with the dynamics associated with \Cref{corollary:1/4}. Moreover, let $\vec{\mu}^{(t)} \defeq \vec{\mu}_1^{(t)} \otimes \dots \otimes \vec{\mu}_n^{(t)}$ be the associated joint probability distribution, and $\bar{\vec{\mu}}^{(t)} \defeq \frac{1}{T} \sum_{t=1}^T \vec{\mu}^{(t)}$ be their average over time. Then, since the expression in \Cref{definition:efce} is linear (recall that the set of transformations $\Psi_{i}$ is linear), it follows from the linearity of expectation that $\bar{\vec{\mu}}^{(t)}$ is an $\epsilon$-EFCE, where, if $\reg_i^T$ is the cumulative $\Psi_i$-regret of player $i$ with respect to $\mathcal{Q}_i$, it holds that $\epsilon \defeq \frac{1}{T} \max_i \reg_i^{T}$. Finally, the proof follows given that $\reg_i^T = \mathcal{P} T^{1/4}$, for every player $i \in [n]$, where $\mathcal{P}$ is a parameter polynomial in the game (\Cref{corollary:1/4}).
\end{proof}

\subsection{Proofs from \texorpdfstring{\Cref{section:efcce}}{Section 5}}
\label{appendix:efcce}

Before we proceed with the proof of \Cref{thm:closedForm}, we first show the following useful claim.

\begin{lemma}\label{lemma: phi decomposition}
For any $\phi_i = \sum_{\info'\in\Info}\lambdav_i[\info']\tdev[\info'][\q_{\info'}]\in\co\Ph$, $\q_i \in\Q_i$, and $\sigma=(j,a)\in\Seqs$,
\[
\phi_i(\q_i)[\sigma]-\q_i[\sigma]=\mleft(\sum_{\info'\preceq \info}\lambdav_i[\info']\q_{\info'}[\sigma]\q_i[\sigma_{j'}]\mright)-d_{\sigma}\q_i[\sigma].
\]
\end{lemma}
\begin{proof}
By definition of the linear mapping $\tdev[\info'][\q_{\info'}]$, we have that
\begin{align*}
    \phi_i(\q_i)[\sigma] & = \sum_{\info'\in\Info}\lambdav_i[\info']\tdev[\info'][\q_{\info'}](\q_i)[\sigma]
    \\ & = 
    \sum_{\info'\in\Info}\lambdav_i[\info']\begin{cases} \q_{\info'}[\sigma] \q_i[\sigma_{j'}] & \text{if }\sigma\succeq\info'\\ \q_i[\sigma]  & \text{\normalfont otherwise}\end{cases}
    \\ & = 
    \mleft(1-\sum_{\info'\preceq\sigma}\lambdav_i[\info']\mright)\q_i[\sigma] 
    + \sum_{\info'\preceq \sigma}\lambdav_i[\info']\q_{\info'}[\sigma]\q_i[\sigma_{j'}].
\end{align*}
A rearrangement of the last equation completes the proof.
\end{proof}

\closedForm*
\begin{proof}
Consider some arbitrary $\phi_i \in\co\Ph$. The proof is divided into three claims: (i) the vector $\q_i \in\bbR^{|\Seqs|}$ obtained through \Cref{algo:FP-EFCCE} is such that $\q_i \in\Q_i$ (\emph{i.e.}, it is a proper sequence-form strategy); (ii) the sequence-form strategy $\q_i$ obtained through \Cref{algo:FP-EFCCE} is such that $\phi_i(\q_i)=\q_i$; and (iii) \Cref{algo:FP-EFCCE} runs in time $O(|\Sigma_i| \mathfrak{D}_i)$.

\textit{Part 1: $\q_i$ is a sequence-form strategy}. First, by construction \Cref{line:q star init}) we have that $\q_i[\emptyseq]=1$. Thus, we need to show that, for each $\info\in\Info$, it holds that $\sum_{a \in \mathcal{A}_j }\q_i[(j,a)]=\q_i[\sigma_j]$ (recall \Cref{definition:sequence_form}). Indeed, for any $\info\in\Info$ such that $d_{\sigma}=0$, it is immediate to see that the above constraint holds by construction (\Cref{line: set to uniform}). On the other hand, for each $\info\in\Info$ such that $d_{\sigma}\ne 0$, we have that
\begin{align*}
    \sum_{a\in \mathcal{A}_{j}} \q_i[(\info,a)]  &=
    \frac{1}{d_{\sigma}}\mleft(\sum_{a \in \mathcal{A}_{j} }\sum_{\info'\preceq\info}\lambdav_i[\info']\q_{\info'}[(j,a)]\q_i[\sigma_{j'}]\mright)
    \\ &  =
    \frac{1}{d_{\sigma}}\mleft(\sum_{\info'\preceq\info}\lambdav_i[\info']\q_i[\sigma_{j'}]\mleft(\sum_{a \in \mathcal{A}_{j} }\q_{\info'}[(j,a)]\mright)\mright)
    \\ &  =
    \frac{1}{d_{\sigma}}\mleft(\sum_{\info'\preceq\info}\lambdav_i[\info']\q_i[\sigma_{j'}] \cdot\begin{cases} \q_{\info'}[\sigma_j] & \text{if } j' \prec j\\ 1 & \text{otherwise}\end{cases}\mright),
\end{align*}
where the first equality holds by \Cref{eq: fixed point}, and the last equality holds since $\q_{\info'}\in \Q_{\info'}$.
%
Next, we distinguish between two cases: if $d_{\sigma_j}=0$, then $\lambdav_i[\info']=0$ for each $\info'\prec\info$. Therefore, since we are assuming $d_{\sigma}\ne 0$, it must be the case that $d_{\sigma}=\lambdav_i[\info]\ne 0$. This yields that
\begin{align*}
\sum_{a \in \mathcal{A}_j}\q_i[(\info,a)] & =
\frac{1}{d_{\sigma}}\mleft(\sum_{\info'\preceq\info}\lambdav_i[\info']\q_i[\sigma_{j'}]\cdot\begin{cases} \q_{\info'}[\sigma_j] & \text{if } I' \prec I\\ 1 & \text{otherwise}\end{cases}\mright)
\\ & = \frac{1}{\lambdav_i[\info]}\mleft(\lambdav_i[\info]\q_i[\sigma_j]\mright)
 = \q_i[\sigma_j].
\end{align*}
On the other hand, if $d_{\sigma_j} \ne 0$, then $\q_i[\sigma_j]$ was set according to \Cref{eq: fixed point}, and thus, 
\begin{equation}\label{eq:parseq fixed point}
\q_i[\sigma_j] = \frac{1}{d_{\sigma_j}} \mleft(\sum_{\info'\prec \info}\lambdav_i[\info']\q_i[\sigma_{j'}]\q_{\info'}[\sigma_j]\mright).
\end{equation}
By definition of $d_\sigma$ (Line~\ref{line:dsigma}), it holds that $d_{\sigma}=d_{\sigma_j} + \lambdav_i[\info]$. Thus,
\begin{align*}
    \sum_{a \in \mathcal{A}_j }\q_i[(\info,a)]  & 
     =
    \frac{1}{d_{\sigma}}\mleft(\sum_{\info'\preceq\info}\lambdav_i[\info']\q_i[\sigma_{j'}]\cdot\begin{cases} \q_{\info'}[\sigma_j] & \text{if } j' \prec j \\ 1 & \text{otherwise}\end{cases}\mright)
    \\ &  =
    \frac{1}{d_{\sigma_j}+\lambdav_i[\info]}\Bigg(\lambdav_i[\info]\q_i[\sigma_j]
    +\sum_{\info'\prec\info}\lambdav_i[\info']\q_i[\sigma_{j'}]\q_{\info'}[\sigma_j]\Bigg)
    \\ &  =
    \frac{1}{d_{\sigma_j}+\lambdav_i[\info]} \mleft(\lambdav_i[\info]\q_i[\sigma_j] + d_{\sigma_j} \q_i[\sigma_j]\mright)
    =
    \q_i[\sigma_j],
\end{align*}
where the second to last equality is obtained from~\eqref{eq:parseq fixed point}. This concludes the first part of the proof.

\textit{Part 2: $\q_i$ is a fixed point of $\phi_i$.} Fix a sequence $\sigma=(j,a)\in\Seqs$. We want to show that $\phi(\q_i)[\sigma]-\q_i[\sigma]=0$. If $\sum_{\info'\preceq\info}\lambdav_i[\info']=0$, then it immediately follows that $\phi_i(\q_i)[\sigma]=\q_i[\sigma]$. Otherwise, applying \cref{lemma: phi decomposition} and substituting $\q_i[\sigma]$ according to \Cref{eq: fixed point} yields that
\begin{align*}
    \phi_i(\q_i)[\sigma]-\q_i[\sigma] & 
    = 
    \mleft(\sum_{\info'\preceq \info}\lambdav_i[\info']\q_{\info'}[\sigma]\q_i[\sigma_{j'}]\mright)-d_{\sigma}\q_i[\sigma]
    \\ &  =  
    \mleft(\sum_{\info'\preceq \info}\lambdav_i[\info']\q_{\info'}[\sigma]\q_i[\sigma_{j'}]\mright)
    - 
    \frac{d_{\sigma}}{d_{\sigma}}\mleft(\sum_{\info'\preceq \info}\lambdav_i[\info']\q_{\info'}[\sigma]\q_i[\sigma_{j'}]\mright)
    = 0.
\end{align*}
This concludes the second part of the proof.

\textit{Part 3: time complexity.} For each sequence in $\Sigma_i^*$ \cref{algo:FP-EFCCE} has to visit at most $\mathfrak{D}_i$ information sets as part of \Cref{line:dsigma,eq: fixed point}. This completes the proof.
\end{proof}

\stabilityefcce*

\begin{proof}
By assumption, we know that $\vec{\lambda}_i[j] > 0$ for all $j \in \mathcal{J}_i$. Thus, it will always be the case that $d_{\sigma} > 0$, for any $\sigma \in \Sigma_{i}$. Hence, \Cref{algo:FP-EFCCE} will never visit the first ``if'' branch. 

Now fix any $t \geq 2$. We will show by induction that $\vec{q}_i^{(t)}[\sigma]$ is such that $\vec{q}_i^{(t)}[\sigma] \leq (1 + \kappa)^{3 \mathfrak{D}_i[\sigma] - 2} \vec{q}_i^{(t-1)}[\sigma]$ and $\vec{q}_i^{(t-1)}[\sigma] \leq (1 + \kappa)^{3 \mathfrak{D}_{i}[\sigma] - 2} \vec{q}_i^{(t)}[\sigma]$, where $\mathfrak{D}_{i}[\sigma] \geq 1$ is the depth of sequence $\sigma \in \Sigma^*_{i}$ with respect to $i$'s subtree. For the base case, let $\sigma = (j, a)$ be a sequence such that $j \in \mathcal{J}_{i}$ corresponds to a root information set of player $i$. Then, it follows from \Cref{algo:FP-EFCCE} that $d_{\sigma} = \vec{\lambda}_i[j]$, in turn implying that $\vec{q}_i^{(t)}[\sigma] = \vec{q}^{(t)}_j[\sigma]$. Thus, $\vec{q}_i^{(t)}[\sigma]$ is indeed $\kappa$-multiplicative-stable. Next, consider some sequence $\sigma = (j, a)$ at depth $\mathfrak{D}_{i}[\sigma] \geq 2$ such that all ancestor sequences---\emph{i.e.} all $\sigma_{j'}$ for $j' \preceq j$---satisfy the inductive hypothesis. Then, we have that
\begin{align}
    \vec{q}_i^{(t)}[\sigma] &= \frac{\sum_{j' \preceq j} \vec{\lambda}_i^{(t)}[j'] \vec{q}_{j'}^{(t)}[\sigma] \vec{q}_i^{(t)}[\sigma_{j'}]}{\sum_{j' \preceq j}\vec{\lambda}_i^{(t)}[j']} \label{eq:EFCCE-algo} \\
    &\leq (1 + \kappa)^3 \frac{\sum_{j' \preceq j} \vec{\lambda}_i^{(t-1)}[j'] \vec{q}_{j'}^{(t-1)}[\sigma] \vec{q}_i^{(t)}[\sigma_{j'}]}{\sum_{j' \preceq j}\vec{\lambda}_i^{(t-1)}[j']} \label{eq:EFCCE-stable} \\
    &\leq (1 + \kappa)^3 (1 + \kappa)^{3 \mathfrak{D}_{i}[\sigma] - 5} \vec{q}_i^{(t-1)}[\sigma] \label{eq:EFCCE-induction} \\
    &= (1 + \kappa)^{3 \mathfrak{D}_{i}[\sigma] - 2} \vec{q}_i^{(t-1)} [\sigma] \notag,
\end{align}
where \eqref{eq:EFCCE-algo} derives from the formula of \Cref{algo:FP-EFCCE}; \eqref{eq:EFCCE-stable} uses the $\kappa$-multiplicative-stability of the sequences $(\vec{\lambda}_i^{(t)})$ and $(\vec{q}_{j}^{(t)})$, for any $j \in \mathcal{J}_i$; and \eqref{eq:EFCCE-induction} leverages the inductive hypothesis. Similar reasoning yields:
\begin{align*}
    \vec{q}_i^{(t)}[\sigma] &= \frac{\sum_{j' \preceq j} \vec{\lambda}_i^{(t)}[j'] \vec{q}_{j'}^{(t)}[\sigma] \vec{q}_i^{(t)}[\sigma_{j'}]}{\sum_{j' \preceq j}\vec{\lambda}_i^{(t)}[j']} \\
    &\geq \frac{1}{(1 + \kappa)^3} \frac{\sum_{j' \preceq j} \vec{\lambda}_i^{(t-1)}[j'] \vec{q}_{j'}^{(t-1)}[\sigma] \vec{q}_i^{ (t)}[\sigma_{j'}]}{\sum_{j' \preceq j}\vec{\lambda}_i^{(t-1)}[j']} \\
    &\geq \frac{1}{(1 + \kappa)^{3}} \frac{1}{(1 + \kappa)^{3 \mathfrak{D}_i[\sigma] - 5}} \vec{q}_i^{(t-1)}[\sigma] \\
    &\geq \frac{1}{(1 + \kappa)^{3 \mathfrak{D}_i[\sigma] - 2}} \vec{q}_i^{ (t-1)}[\sigma].
\end{align*}
Thus, if $\mathfrak{D}_{i}$ is the depth of $\mathcal{J}_i$, we conclude that $\vec{q}_i^{(t)}[\sigma] \leq (1 + \kappa)^{3\mathfrak{D}_{i} - 2} \vec{q}_i^{(t-1)}[\sigma] \leq e^{3 \mathfrak{D}_{i} \kappa - 2 \kappa} \vec{q}_i^{(t-1)}[\sigma] \leq (1 + 6 \mathfrak{D}_{i} \kappa) \vec{q}_i^{(t-1)}[\sigma]$, where we used the inequalities $1 + x \leq e^x$ for all $x \in \R$, and $e^x \leq 1 + 2x$ for $x \in [0, 1/2]$, applicable as long as $\kappa \leq 1/(12\mathfrak{D}_i)$. Similarly, we obtain that $\vec{q}_i^{(t)} \geq (1 + 12 \mathfrak{D}_i \kappa)^{-1} \vec{q}_i^{(t-1)}$, concluding the proof.
\end{proof}

\section{Sequential Decision Making and Stable-Predictive \texorpdfstring{$\cfr$}{}}
\label{appendix:opt-cfr}

The main purpose of this section is to provide a stable-predictive variant of $\cfr$ following the construction in~\citep{Farina19:Stable}. The main result is given in \Cref{theorem:opt-cfr}. We begin by introducing the basic setting of \emph{sequential decision making}.

A sequential decision process can be represented using a tree consisting of two types of nodes: \emph{decision nodes} and \emph{observation nodes}. The set of all decision nodes will be denoted by $\mathcal{J}$, while the set of observation nodes by $\mathcal{K}$. At every decision node $j \in \mathcal{J}$ the agent has to select a strategy $\vec{x}_j$ in the form of a probability distribution over all possible actions $\mathcal{A}_j$. On the other hand, at every observation point $k \in \mathcal{K}$ the agent may receive a feedback in the form of a signal in the set $\mathcal{S}_k$. At every decision point $j \in \mathcal{J}$ of the sequential decision process, the strategy $\vec{x}_j \in \Delta(\mathcal{A}_j)$ secures a utility of the form $\langle \vec{\ell}_j, \vec{x}_j \rangle$, for some utility vector $\vec{\ell}_j$. The expected utility throughout the entire decision process can be expressed as $\sum_{j \in \mathcal{J}} \pi_j \langle \vec{\ell}_j, \vec{x}_j \rangle$, where $\pi_j$ is the probability that the agent reaches decision point $j$. It is important to point out that in all extensive-form games of perfect recall the agents face a sequential decision process. A central ingredient for our construction of stable-predictive CFR is a decomposition of the strategy space, described in detail below.

\paragraph{Decomposition of the Sequence-Form Strategy Space} Our construction will rely on a recursive decomposition of the sequence-form strategy space $\mathcal{X}^{\triangle}$:

\begin{itemize}
    \item Consider an observation node $k \in \mathcal{K}$, and let $\mathcal{C}_k$ be the children decision points of $k$. Then, $\mathcal{X}_k^{\triangle}$ can be decomposed as the following Cartesian product:
    \begin{equation}
        \label{eq:observation-decomposition}
        \mathcal{X}_k^{\triangle} \defeq \bigtimes_{j \in \mathcal{C}_k} \mathcal{X}_j^{\triangle};
    \end{equation}
    \item Consider a decision point $j \in \mathcal{J}$, and let $\mathcal{C}_j = \{k_1, \dots, k_{m_j} \}$ be the children observation points of $j$, with $m_j = |\mathcal{A}_j|$. Then, $\mathcal{X}_j^{\triangle}$ can be decomposed as follows:
    \begin{equation}
        \label{eq:decision-decomposition}
        \mathcal{X}^{\triangle}_j \defeq \left\{ 
        \begin{pmatrix}
        \vec{\lambda}[1] \\
        \vdots \\
        \vec{\lambda}[m_j] \\
        \vec{\lambda}[1] \vec{x}_{1} \\
        \vdots \\
        \vec{\lambda}[m_j] \vec{x}_{m_j}
        \end{pmatrix}
        : (\vec{\lambda}[1], \dots, \vec{\lambda}[m_j]) \in \Delta^{m_j}, \vec{x}_{1} \in \mathcal{X}_{k_1}^{\triangle}, \dots, \vec{x}_{{m_j}} \in \mathcal{X}^{\triangle}_{k_{m_{j}}}
        \right\}.
    \end{equation}
\end{itemize}
In view of this decomposition, the basic ingredients for the overall construction are given in \Cref{proposition:observation} and \Cref{proposition:decision}. We should note that in the sequel the stability and the predictive bounds will be tacitly assumed with respect to the pair of norms $(\| \cdot \|_1, \| \cdot \|_{\infty})$.

\begin{proposition}
    \label{proposition:observation}
    Consider an observation node $k \in \mathcal{K}$, and assume access to a $\kappa_j$-multiplicative-stable $(\alpha_j, \beta_j)$-predictive regret minimizer $\mathcal{R}_j^{\triangle}$ over the sequence-form strategy space $\mathcal{X}_j^{\triangle}$, for each $j \in \mathcal{C}_k$. Then, we can construct a $\max_j\{\kappa_j\}$-multiplicative-stable $(A, B)$-predictive regret minimizer $\mathcal{R}_k^{\triangle}$ for the sequence-form strategy space $\mathcal{X}_k^{\triangle}$, where $A = \sum_{j \in \mathcal{C}_k} \alpha_j$ and $B = \sum_{j \in \mathcal{C}_k} \beta_j$.
\end{proposition}

\begin{proof}
Given the decomposition of \eqref{eq:observation-decomposition}, the composite regret minimizer can be constructed using a regret circuit for the Cartesian product~\citep{Farina19:Regret}. In particular, it is direct to verify that $\reg_k^{\triangle, T} = \sum_{j \in \mathcal{C}_k} \reg_j^{\triangle, T}$, where $\reg_k^{\triangle, T}$ is the regret accumulated by the composite regret minimizer, and $\reg_j^{\triangle, T}$ the regret of each individual regret minimizer $\mathcal{R}_j^{\triangle}$. In particular, by assumption we know that 
\begin{equation*}
    \reg_j^{\triangle, T} \leq \alpha_j + \beta_j \sum_{t=1}^T \| \vec{\ell}_j^{\triangle, (t)} - \vec{\ell}_j^{\triangle, (t-1)}\|^2_{\infty}.
\end{equation*}
As a result, we can conclude that
\begin{equation*}
    \reg_k^{\triangle, T} \leq \left( \sum_{j \in \mathcal{C}_k} \alpha_j \right) + \left( \sum_{j \in \mathcal{C}_k} \beta_j \right) \sum_{t=1}^T \|\vec{\ell}_k^{\triangle, (t)} - \vec{\ell}_k^{\triangle, (t-1)}\|^2_{\infty},
\end{equation*}
where we used that $\| \vec{\ell}_j^{\triangle, (t)} - \vec{\ell}_j^{\triangle, (t-1)}\|_{\infty} \leq \| \vec{\ell}_k^{\triangle, (t)} - \vec{\ell}_k^{\triangle, (t-1)}\|_{\infty}$. Finally, the $\max_j \{\kappa_j\}$-multiplicative-stability of $\mathcal{R}_k^{\triangle}$ follows directly from the $\kappa_j$-multiplicative-stability of each $\mathcal{R}_j^{\triangle}$.
\end{proof}

In the following construction the regret circuit for the convex hull uses an advanced prediction mechanism, analogously to that we explained in \Cref{remark:better_prediction}. 

\begin{proposition}
    \label{proposition:decision}
    Consider a decision node $j \in \mathcal{J}$, and assume access to a $K$-multiplicative-stable $(\alpha_k, \beta_k)$-predictive regret minimizer $\mathcal{R}_k^{\triangle}$ over the sequence-form strategy space $\mathcal{X}_k^{\triangle}$, for each $k \in \mathcal{C}_j$. Moreover, assume access to a $\kappa$-multiplicative-stable $(\alpha, \beta)$-predictive regret minimizer $\mathcal{R}_{\Delta}$ over the simplex $\Delta(\mathcal{A}_j)$. Then, we can construct a $(\kappa + \kappa K + K)$-multiplicative-stable $(A, B)$-predictive regret minimizer $\mathcal{R}_j^{\triangle}$ for the sequence-form strategy space $\mathcal{X}_j^{\triangle}$, where 
    \begin{equation*}
        \begin{split}
        A &= \alpha + \max_{k \in \mathcal{C}_j}\{\alpha_k\}; \\ 
        B &= \max_{k \in \mathcal{C}_j} \{ \beta_k\} + \beta \| \mathcal{Q}\|^2_1,
        \end{split}
    \end{equation*}
    where $\|\mathcal{Q}\|_1$ an upper bound on the $\ell_1$ norm of all $\vec{x} \in \mathcal{X}^{\triangle}$.
\end{proposition}

\begin{proof}
For this construction we will use the regret circuit for the convex hull, stated in \Cref{proposition:regret_circuit-co}. First, we have that, by assumption, the regret $\reg_k^{\triangle, T}$ accumulated by each regret minimizer $\mathcal{R}_k^{\triangle}$ can be bounded as
\begin{equation*}
    \reg_k^{\triangle, T} \leq \alpha_k + \beta_k \sum_{t=1}^T \|\vec{\ell}_k^{\triangle, (t)} -\vec{\ell}_k^{\triangle, (t-1)} \|^2_{\infty}.
\end{equation*}
Moreover, by construction, each regret minimizer $\mathcal{R}_k^{\triangle}$ receives the same utility as $\mathcal{R}_j^{\triangle}$; this, along with the guarantee of \Cref{proposition:regret_circuit-co}, imply that 
\begin{equation}
    \label{eq:decision_node-1}
    \reg_j^{\triangle, T} \leq \alpha + \max_{k \in \mathcal{C}_j} \{ \alpha_k \} + \max_{k \in \mathcal{C}_j} \{ \beta_k \} \sum_{t=1}^T \|\vec{\ell}_j^{\triangle, (t)} - \vec{\ell}_j^{\triangle, (t-1)} \|^2_{\infty} + \beta \sum_{t=1}^T \| \vec{\ell}_\lambda^{(t)} - \vec{\ell}_\lambda^{(t-1)}\|^2_{\infty},
\end{equation}
where $\vec{\ell}_\lambda^{(t)}$ represents the utility function received as input by $\mathcal{R}_{\Delta}$. Next, similarly to the analysis of \Cref{proposition:circ-pred}, we can infer that for some $k \in \mathcal{C}_j$,
\begin{align*}
    \|\vec{\ell}_\lambda^{(t)} - \vec{\ell}_\lambda^{(t-1)}\|_{\infty} = | \langle \vec{\ell}_j^{\triangle, (t)} - \vec{\ell}_j^{\triangle,(t-1)}, \vec{x}_k^{(t)} \rangle | \leq \| \vec{\ell}_j^{\triangle, (t)} - \vec{\ell}_j^{\triangle, (t-1)} \|_\infty \| \vec{x}_k^{(t)} \|_1 \leq \| \vec{\ell}_j^{\triangle, (t)} - \vec{\ell}_j^{\triangle, (t-1)} \|_\infty \| \mathcal{Q}\|_1
\end{align*}
where we used that $\| \vec{x}_k^{(t)} \|_1 \leq \|\mathcal{Q}\|_1 $. As a result, if we plug-in this bound to \eqref{eq:decision_node-1} we can conclude that 
\begin{equation*}
    \reg_j^{\triangle, T} \leq \left( \alpha + \max_{k \in \mathcal{C}_j}\{\alpha_k\} \right) + \left( \max_{k \in \mathcal{C}_j} \{ \beta_k\} + \beta \|\mathcal{Q}\|_1^2 \right) \sum_{t=1}^T \| \vec{\ell}_j^{\triangle, (t)} - \vec{\ell}_j^{\triangle, (t-1)}\|^2_{\infty}.
\end{equation*}
Finally, the $(\kappa + \kappa K + K)$-multiplicative-stability of $\mathcal{R}_j^{\triangle}$ can be directly verified from the decomposition given in \eqref{eq:decision-decomposition}.
\end{proof}

\begin{remark}
Given the decomposition provided in \Cref{eq:decision-decomposition}, the regret circuit for the convex hull should operate on the appropriate ``lifted'' space, but this does not essentially alter the analysis of the regret since the augmented entries in the lifted space remain invariant; this is illustrated and further explained in~\citep[Figure 7]{Farina19:Regret}.
\end{remark}

Finally, we inductively combine \Cref{proposition:observation} and \Cref{proposition:decision} in order to establish the main result of this section: a stable-predictive variant of $\cfr$.

\begin{theorem}[Optimistic $\cfr$]
    \label{theorem:opt-cfr}
    If every local regret minimizer $\mathcal{R}_j^{\triangle}$ is updated using $\omw$ with a sufficiently small learning rate $\eta$, for each $j \in \mathcal{J}$, we can construct an $(A, B)$-predictive regret minimizer $\mathcal{R}^{\triangle}$ for the space of sequence-form strategies $\mathcal{X}^{\triangle}$, such that
    \begin{equation}
        \label{eq:predictivity}
        \begin{split}
            A &= O\left( \frac{\log |\mathcal{A}|}{\eta} \|\mathcal{Q}\|_1 \right) ; \\
            B &= O( \eta \|\mathcal{Q}\|^3_1),
        \end{split}
    \end{equation}
    where $|\mathcal{A}| \defeq \max_{j \in \mathcal{J}} |\mathcal{A}_j|$; $\|\vec{\ell}\|_\infty$ is an upper bound on the $\ell_{\infty}$ norm of the utilities observed by $\mathcal{R}^{\triangle}$; $\|\mathcal{Q}\|_1$ is an upper bound on the $\ell_1$ norm of any $\vec{x} \in \mathcal{X}^{\triangle}$; and $\mathfrak{D}$ is the depth of the decision process. Moreover, the sequence of strategies produced by $\mathcal{R}^{\triangle}$ is $O(\eta \mathfrak{D} \|\mathcal{Q}\|_1 \|\vec{\ell}\|_\infty)$-multiplicative-stable.
\end{theorem}
\begin{proof}
First of all, it is easy to see that all losses observed by the ``local'' regret minimizers---\emph{i.e.}, the \emph{counterfactual losses}~\citep[Section 4]{Farina19:Stable}---have $\ell_\infty$ bounded by $O(\|\mathcal{Q}\|_1 \| \vec{\ell}\|_\infty)$. As a result, we can conclude from \Cref{lemma:OMW-simplex-stability} that the output of each local regret minimizer $\mathcal{R}_j^{\triangle}$ under $\omw$ with a sufficiently small learning rate $\eta$ is $O(\eta \|\mathcal{Q}\|_1 \|\vec{\ell}\|_\infty )$-multiplicative-stable. Along with \Cref{proposition:decision}, we can inductively infer that the output of $\mathcal{R}^{\triangle}$ is $O(\eta \mathfrak{D} \| \mathcal{Q} \|_1 \|\vec{\ell}\|_\infty) $-multiplicative-stable, for a sufficiently small $\eta = O(1/(\mathfrak{D} \| \mathcal{Q} \|_1 \|\vec{\ell}\|_\infty))$. This established the claimed bound for the multiplicative stability. 

For the predictive bound, first recall that the range of the entropic regularizer on the $m$-dimensional simplex is $\log m$. Thus, by \Cref{lemma:oftrl-predictive} we know that each local regret minimizer at information set $j \in \mathcal{J}$ instantiated with $\omw$ with learning rate $\eta$ will be $(\log (|\mathcal{A}_j|/\eta, \eta)$-predictive. As a result, the predictive bound in \eqref{eq:predictivity} follows inductively from \Cref{proposition:decision}.
\end{proof}

Naturally, the same bounds apply for constructing a regret minimizer for the subspace $\mathcal{X}_j^{\triangle}$, for any decision point $j \in \mathcal{J}$, as required in \Cref{proposition:R_sigma}.

\section{Description of Game Instances used in the Experiments}
\label{app:games}

In this section we give a description of the game instances used in our experiments. The parameters associated with each game are summarized in \Cref{table:parameters}.

\paragraph{Kuhn poker} First, we experimented on a \emph{three-player} variant of the popular benchmark game known as \emph{Kuhn poker}~\citep{Kuhn50:Simplified}. In our version, a deck of three cards---a Jack, a Queen, and a King---is employed. Players initially commit a single chip to the pot, and privately receive a single card.
The first player can either {\em check} or {\em bet} (\emph{i.e.} place an extra chip). Then, the second player can in turn check or bet if the first player checked, or {\em folded/called} in response to the first player's bet. If no betting occurred in the previous rounds, the third player can either check or bet. In the contrary case, the player can either fold or call. Following a bet of the second player (or respectively the third player), the first player (or respectively the first and the second players) has to decide whether to fold or to call. At the \emph{showdown}, the player with the \emph{highest} card---who has not folded in a previous round---gets to win all the chips committed in the pot.

\paragraph{Sheriff} Our second benchmark is a bargaining game inspired by the board game \emph{Sheriff of Nottingham}, introduced by \cite{Farina19:Correlation}. In particular, we used the \emph{baseline} version of the game. This game consists of two players: the \emph{Smuggler} and the \emph{Sheriff}. The smuggler must originally come up with a number $n \in \{0, 1, 2, 3\}$ which corresponds to the number of illegal items to be loaded in the cargo. It is assumed that each illegal item has a fixed value of $1$. Subsequently, $2$ rounds of bargaining between the two players follow. At each round, the Smuggler decides on a bribe ranging from $0$ to $3$, and the Sheriff must decide whether or not the cargo will be inspected given the bribe amount. The Sheriff's decision is binding only in the last round of bargaining. In particular, if during the last round the Sheriff accepts the bribe, the game stops with the Smuggler obtaining a utility of $n$ minus the bribe amount $b$ that was proposed in the last bargaining round, while the Sheriff receives a utility equal to $b$. On the other hand, if the Sheriff does not accept the bribe in last bargaining round and decides to inspect the cargo, there are two possible alternatives. If the cargo has no illegal items (\emph{i.e.} $n = 0$), the Smuggler receives the fixed amount of $3$, while the utility of the Sheriff is set to be $-3$. In the contrary case, the utility of the smuggler is assumed to be $-2n$, while the utility of the Sheriff is $2n$.


\paragraph{Liar's dice} The final benchmark we experimented on is the game of \emph{Liar's dice}, introduced by ~\citet{Lisy15:Online}. In the three-player variant, the beginning of the game sees each of the three players privately roll an unbiased $3$-face die. Then, the players have to sequentially make claims about their private information. In particular, the first player may announce any face value up to $3$, as well as the minimum number of dice that the player claims are showing that value among the dice of \emph{all} players. Then, each player can either make a higher bid, or challenge the previous claim by declaring the previous agent a ``liar''. More precisely, it is assumed that a bid is higher than the previous one if either the face value is higher, or if the claimed number of dices is greater. If the current claim is challenged, all the dices must be revealed. If the claim was valid, the last bidder wins and receives a reward of $+1$, while the challenger suffers a negative payoff of $-1$. Otherwise, the utilities obtained are reversed. Any other player will receive $0$ utility. 

\paragraph{Goofspiel} This game was introduced---in its current form---by~\citet{Ross71:Goofspiel}. In Goofspiel every player has a hand of cards numbered from $1$ to $r$, where $r$ is the \emph{rank} of the game. An additional stack of $r$ cards is shuffled and singled out as winning the current prize. In each turn a prize card is revealed, and each player privately chooses one of its cards to bid. The player with the highest card wins the current prize; in case of a tie, the prize card is discarded. After $r$ turns have been completed, all the prizes have been dealt out and players obtain the sum of the values of the prize cards they have won. It is worth noting that, due to the tie-breaking mechanism, even two-player instances are general-sum. We also consider instances with \emph{limited information}---the actions of the other players are observed only at the end of the game. This makes the game strategically more involved as players have less information regarding previous opponents' actions. 

\begin{table}[H]\centering
    \begin{tabular}{lrrrr}
        \toprule
            \textbf{Game} & \textbf{Players} & \textbf{Decision points} & \textbf{Sequences} & \textbf{Leaves}\\
        \midrule
            Kuhn poker & 3 & 36 & 75 & 78\\
            Sheriff & 2 & 73 & 222 & 256\\
            Goofspiel & 3 & 837 & 934 & 1296\\
            Liar's dice & 3 & 1536 & 3069 & 13797\\
        \bottomrule
    \end{tabular}
    \caption{The parameters of each game.}
    \label{table:parameters}
\end{table}

\end{document}